\documentclass[a4paper, 12pt]{article}

\usepackage[english]{babel}
\usepackage[utf8x]{inputenc}
\usepackage[T1]{fontenc}

\usepackage[a4paper,top=3cm,bottom=2cm,left=0.75in,right=0.75in,marginparwidth=1.75cm]{geometry}
\usepackage{setspace} 
\usepackage{amsmath, amssymb, amsthm}
\usepackage{graphicx}
\usepackage[colorinlistoftodos]{todonotes}
\usepackage[colorlinks=true, allcolors=blue]{hyperref}
\usepackage{amsfonts}
\usepackage{optidef}
\usepackage{dsfont}
\usepackage{enumitem}
\usepackage{booktabs}
\usepackage{adjustbox}
\usepackage{subcaption} 

\usepackage{natbib}
\newtheorem{thm}{Theorem}[section]
\newtheorem{lemma}[thm]{Lemma}
\newtheorem{assumption}[thm]{Assumption}
\newtheorem{prop}[thm]{Proposition}
\newtheorem{cor}[thm]{Corollary}

\theoremstyle{definition}

\newtheorem{remark}[thm]{Remark}
\newtheorem{design}[thm]{Design}
\numberwithin{equation}{section}

\newcommand{\prob}{\mathbb{P}}
\newcommand{\wh}{\widehat}
\newcommand{\var}{\mathrm{Var}}
\newcommand{\lag}{\langle}
\newcommand{\rag}{\rangle} 
 
\newcommand{\eproof}{\end{proof}}
\newcommand{\convp}{\overset{p}{\to}}
\newcommand{\convd}{\overset{d}{\to}}
\newcommand{\one}{\mathds{1}}
\newcommand{\inv}{^{-1}}
\newcommand{\ov}{\overline}
\newcommand{\wt}{\Tilde}

\DeclareMathOperator*{\argmin}{argmin}
\DeclareMathOperator{\vect}{vec}
\DeclareMathOperator{\poly}{Poly}

\newcommand{\numgroups}{B} 
\newcommand{\kvec}{k} 

\newcommand{\xvec}{(x_{it}^1, \dots, x_{it}^{\numgroups})}

\newcommand{\thetaci}{\theta(c_i)}

\newcommand{\truethetaci}{\theta^0(c_i^0)}
\newcommand{\thetaoracle}{\Tilde{\theta}}
\newcommand{\estthetaci}{\wh{\theta}(\wh{c}_i)}

\newcommand{\samplerisk}{\wh{Q}}

\newcommand{\mapping}{\sigma}
\newcommand{\qtilde}{\Tilde{Q}}
\newcommand{\qoracle}{\Tilde{Q}}
\newcommand{\qconc}{\wh{Q}}
\newcommand{\distab}{d(\ell, a, b)}
\newcommand{\neighborhood}{\mathcal{N}}
\newcommand{\boundfn}{B_i^T}
\newcommand{\clusterdist}{d}
\newcommand{\clusterdistmin}{d_{min}}
\newcommand{\erroracle}{\Tilde{e}}
\newcommand{\ellaux}{\Tilde{\ell}}
\newcommand{\evalboundtwo}{\underline{\rho}}
\newcommand{\covxmatrix}{\wh{M}}
\newcommand{\cltvec}{v}
\newcommand{\estvariance}{\wh{V}}
\newcommand{\deltatheta}{\Delta \wh{\theta}}

\newcommand{\cp}{C_p}
\newcommand{\biascorr}{\wh{b}}
\newcommand{\baserisk}{\wh{Q}^0}

\newcommand{\cest}{\wh{c_i}}
\newcommand{\ctrue}{c_i^0}
\newcommand{\clusterspace}{\mathcal{C}}
\newcommand{\thetak}{\wh{\theta}^k}
\newcommand{\gammak}{\wh{\gamma}^k}

\newcommand{\modelk}{\wh{m}^k}

\newcommand{\modeltrue}{m^{k_0}}
\newcommand{\model}{m}
\newcommand{\rate}{r_{NT}}
\newcommand{\avgdist}{d}

\newcommand{\unifboundxxsup}{C_{NT}}

\newcommand{\eiglowerbound}{\underline{\lambda}}
\newcommand{\rootrate}{b_T}
\newcommand{\thetarate}{a_{T}}
\newcommand{\qopt}{\wh{Q}^{opt}}
\newcommand{\thetaopt}{\wh{\theta}^{opt}}

\title{Blocked Clusterwise Regression}
\author{Max Cytrynbaum \\ MIT}

\begin{document}
\maketitle

\begin{abstract}
A recent literature in econometrics models unobserved cross-sectional heterogeneity in panel data by assigning each cross-sectional unit a one-dimensional, discrete latent type. 
Such models have been shown to allow estimation and inference by regression clustering methods. 
This paper is motivated by the finding that the clustered heterogeneity models studied in this literature can be badly misspecified, even when the panel has significant discrete cross-sectional structure. 
To address this issue, we generalize previous approaches to discrete unobserved heterogeneity by allowing each unit to have multiple, imperfectly-correlated latent variables that describe its response-type to different covariates.
We give inference results for a k-means style estimator of our model and develop information criteria to jointly select the number clusters for each latent variable. 
Monte Carlo simulations confirm our theoretical results and give intuition about the finite-sample performance of estimation and model selection.
We also contribute to the theory of clustering with an over-specified number of clusters and derive new convergence rates for this setting. 
Our results suggest that over-fitting can be severe in k-means style estimators when the number of clusters is over-specified. 
\end{abstract}

\section{Introduction}

We often believe that there may be significant cross-sectional heterogeneity in the structural relationship between observed covariates and response.
Even with panel data, however, estimating distinct regression coefficients for each cross-sectional unit can be noisy or infeasible when the time dimension is small.
Clusterwise regression methods (e.g. \cite{LinNg2012}, \cite{BonhommeManresa2015}), which model individual heterogeneity as a function of a one-dimensional discrete latent type, have recently become popular as a viable compromise between the common parameter assumption and full heterogeneity. 
However, as we show, this discretization of heterogeneity can be badly misspecified even when the panel has significant cross-sectional structure.
By introducing multiple, imperfectly correlated latent types, we can relieve this issue and significantly enrich the set of panel structures that can be handled by clustering methods.
In particular, our approach is motivated by a class of data-generating processes where units are clustered along multiple latent dimensions or ``response-types'' to distinct blocks of the covariate vector. 
We motivate this generalization with several examples from finance and production function estimation. 
The main contribution of this paper is to modify existing clustering methods for use in this larger family of models and show that they can likewise be used to perform inference on regression parameters. \\  

Following \cite{BonhommeManresa2015} (henceforth BM), we establish consistency and asymptotic normality for a k-means style estimator in our setting. 
The estimation algorithm is iterative and alternates between (1) solving a least-squares problem to estimate cluster parameters for each block and (2) updating latent types for each cross-sectional unit based on a unit-wise predictive criterion. 
As in BM, our proof proceeds by establishing asymptotic equivalence with the oracle estimator where each unit's latent types are known. 
We extend the approach in \cite{AndoBai2016} to give a $\cp$ style information criterion to choose the number of clusters (types) for all the latent variables simultaneously. \\ 

In general, the true number of clusters in a given data set is unknown. 
Thus, the behavior of estimators with a misspecified number of clusters is important both for model selection theory as well as for our understanding the finite-sample properties of clustering estimators. 
Here, we make some contributions to the theory of models with an over-specified number of clusters, improving the convergence rates given in \cite{LiuOverspecifiedGroups} for the linear regression setting. 
In contrast to the well-specified case, difficulty obtaining the ``fast rate'' $O_p(\frac{1}{NT})$ when we over-specify the number of clusters suggests that over-fitting may be severe when the number of clusters is over-specified. 
We conjecture that $\sqrt{T}$-consistency may be optimal for over-specified models.

\subsection{Motivating Example - Production Function Estimation}
Consider panel data on firms' production levels and factor usage. 
We are interested in estimating the firm-specific production functions
\begin{equation} \label{equation:production_function}
y_{it} = \theta_{i1}L_{it} + \theta_{i2} K_{it} + \theta_{i3} M_{it} + \theta_{i4} \text{Elec}_{it} + e_{it}    
\end{equation}

where $y_{it}$ is a measure of output and $L_{it}, K_{it}, M_{it}$ are labor, capital and materials (all in logs), and $\text{Elec}_{it}$ is a measure of electricity usage.  
Suppose that the heterogeneity in factor elasticities can be well approximated by 
\[
\theta_{i\ell} \in \{\theta^{low}_{\ell}, \theta^{mid}_{\ell}, \theta^{high}_{\ell} \} \quad 1 \leq \ell \leq 4
\]

Ignoring endogeneity in input choice, we consider estimation of \ref{equation:production_function} with clusterwise regression. 
The problem with this approach is readily apparent - although each firm can only have one of $3 \cdot 4 = 12$ elasticity types, $\theta_i$ can take up to $3^4 = 81$ distinct values. 
Thus, estimating this model with clusterwise regression would require $k=81$ clusters to be well-specified. 
For a panel of $200$ firms, this would lead to estimation with approximately $N/81 \leq 3$ firms in each regression, in spite of significant cross-sectional homogeneity. 
However, with $k=81$ clusters the model is also significantly over-parameterized. 
For instance, there will be $27$ distinct clusters with each level of labor elasticity coefficient. \\

The problem is that current clustering models assume limited heterogeneity in the individual parameter vectors $\theta_i$. 
In our example, however, cross-sectional heterogeneity takes the form of a few discrete elasticity levels for \emph{each input factor}, while the support of $\theta_i$ itself is large. 
This suggests a model with multiple latent heterogeneity types. 
For instance
\[
\theta_i = (\theta_1(c_{i1}), \theta_2(c_{i2}), \theta_3(c_{i3}), \theta_4(c_{i4}))
\]
with latent type $c_{i\ell}$ for $1 \leq \ell \leq 4$ controlling the elasticity level of factor $\ell$.  

\subsection{Related Literature and Outline}
Early contributions to the econometric literature on clustering include \cite{Sun2005} and \cite{BuchinskyHotz2005}. 
Linear panel data models with discrete unobservable heterogeneity have recently been studied in \cite{LinNg2012}, \cite{BonhommeManresa2015}, \cite{Su2016}, \cite{Su2016Homogeneity}, \cite{DzemskiOkui2018}. 
Our asymptotic normality results for the well-specified case closely follow the analysis pioneered in \cite{BonhommeManresa2015}. 
\cite{AndoBai2016} extends clustering methods to linear factor models and gives an information criterion for choosing the number of clusters. 
We develop a similar $\cp$-style criterion in our setting. 
Outside of the linear case, \cite{Zhang2019} and \cite{Chen2019} study clustered linear conditional quantile regression, and \cite{BonhommeManresa2019Discretizing} considers discrete latent types as an approximation to continuous unobserved heterogeneity. 
\cite{LiuOverspecifiedGroups} studies clustering in M-estimation with an over-specified number of groups. 
We build on their techniques and significantly sharpen their rate results for the linear case. 
In contemporaneous work, \cite{Cheng2019} consider a clustering model with two latent types in a GMM setting. 
By contrast, we allow for $\numgroups > 1$ latent types in a linear model with individual fixed effects. 
Clusterwise regression was initially proposed in \cite{Spath1979} as ``Algorithm 39 - Clusterwise Linear Regression.''\\

Further afield, this paper is related to a number of literatures in statistics and computer science, such as the literature on clustering functional data, e.g. \cite{Wasserman2005CATS}, \cite{Yamamoto2014}, \cite{LintonVogt2019}, and subspace clustering, e.g. \cite{Candes2012}.  
In statistics, related methods include homogeneity pursuit, proposed in  \cite{Ke2015}. See also \cite{Ke2016} and \cite{Lian2019}. 
In the Bayesian literature, clusterwise regression is also known as multilevel regression; see \cite{GelmanHill2007}. \\

The paper is organized as follows - we introduce our model and estimator in section \ref{section:model}. 
Asymptotic properties of the estimator and consistency of model selection are given in section \ref{section:asymptotic}. 
Section \ref{section:overspecification} discusses models with an over-specified number of clusters. 
Monte Carlo simulations are given in section \ref{section:monte_carlo}, and proofs in section \ref{section:proofs}. 
Supplementary appendix \ref{section:supplement} collects technical lemmas and other ancillary discussions. 

\section{Model and Estimation} \label{section:model}
\subsection{Model}
Let $y_{it}$ and $x_{it}$ denote repsonse and covariates for $t = 1, \dots, T$ time periods and $i = 1, \dots, N$ cross-sectional observations. 
The covariate vector $x_{it} \in \mathbb{R}^p$ is divided into $1 \leq \ell \leq \numgroups$ blocks $x_{it}^{\ell}$, where $x_{it} = (x_{it}^1, \dots, x_{it}^{\numgroups})$, and $\numgroups$ denotes the total number of blocks. 
We let $\kvec = (k_1, \dots, k_{\numgroups})$, where $k_{\ell}$ denotes the number of distinct latent types (clusters) associated with the $\ell^{th}$ block.
Possible cluster assignments are denoted $c = (c_1, \dots, c_{\numgroups}) \in \prod_{\ell} [k_{\ell}] \equiv \mathcal{C}$.
For instance, a unit in cluster $1$ in the first block and cluster $3$ in the second block would have $c=(1, 3)$. \\

Each cross-sectional unit belongs to exactly one cluster for each block.
We let $\gamma: [N] \to \prod_{\ell} [k_{\ell}]$ denote an assignment of cross-sectional units to cluster vectors, so that $\gamma(i) = c_i$.  
The set of all possible cluster assignments is denoted $\Gamma$.
In the main specification, we assume that the response $y_{it}$ is given by 

\begin{equation}\label{problem}
y_{it} = x_{it}'\theta(c_i) + e_{it}
\end{equation}

with the $\ell^{th}$ block parameter selected by the latent variable $c_{i\ell}$ 
\begin{equation} \label{equation:cluster_structure}
\theta(c_i) = (\theta_1(c_{i1}), \dots, \theta_{\numgroups}(c_{i \numgroups}))
\end{equation}

Thus, each covariate grouping $\ell$ is associated with $k_{\ell}$ parameter sub-vectors with $\theta_{\ell}(c_{\ell}) \in \mathbb{R}^{d_{\ell}}$.
Our goal is to jointly estimate the true parameter $\theta^0$ and true cluster assignments $\gamma^0(i) = (c^0_{i1}, \dots, c^0_{i\numgroups})$ of each cross-sectional unit.
In appendix \ref{proof:fe}, we also give results for the model with individual fixed effects
\begin{equation}
y_{it} = x_{it}'\theta(c_i) + a_i + e_{it} \\
\end{equation}
\quad \\
\textbf{Relationship with Clusterwise Regression}: The model above nests clusterwise regression (as in \cite{LinNg2012}) when $\numgroups=1$. 
For $\numgroups > 1$, it is statistically equivalent to clusterwise regression when the conditional pdf $\prob(c^0_{i(-\ell)} | c^0_{i\ell})$ is degenerate (perfectly correlated types). 
Similarly, clusterwise regression ($\numgroups=1$) with exponentially many clusters $k = \prod_{\ell=1}^{\numgroups} k_{\ell}$ and exponentially many constraints nests our model.  
For instance, let $p=\numgroups$ and assume $\theta_{\ell} \in \{\pm 1\}$ for each $\ell$. 
Then our model would be equivalent to a clusterwise regression model with $2^p$ clusters and $p2^{p-1}$ linear equality constraints. \\

\textbf{Example - Exchange Rates}:
Consider financial market data with $y_{it}$ the exchange rate against USD of country $i$ at time $t$, $p^{oil}_{t}$ the price of crude oil, $b_{it}$ a measure of the country's business cycle and $r_{t}$ the US discount rate, we model 
\[
y_{it} = \theta_{i1} p^{oil}_{t} + \theta_{i2} b_{it} + \theta_{i3} r_{t} + e_{it}
\]
Due to differences in national industry composition, the magnitude and composition of foreign trade, financial openness and so on, we may expect heterogenous marginal responses $\theta_{i \ell}$ of $y_{it}$ to each of the factors above. 
As in the introduction, we may model $\theta_{i\ell} \in \{\theta_{\ell}(1), \dots, \theta_{\ell}(k_{\ell})\}$, corresponding to $k_{\ell}$ different sensitivity levels to factor $\ell$. 
However, we don't expect these unobserved types to be perfectly correlated across factors. 
For instance, we might expect both Venezuela and China to have large $\theta_{i1}$ but very different $\theta_{i3}$. 
A factor error structure could be accommodated using the techniques in \cite{AndoBai2016}. 

\subsection{Estimator}
We define our estimator of the parameter $\theta^0$ and cluster assignment $\gamma^0$ as

\begin{equation} \label{estimator}
 (\wh{\theta}, \wh{\gamma}) = \argmin_{\gamma \in \Gamma, \theta \in \Theta} \frac{1}{NT}\sum_{i=1}^N\sum_{t=1}^T (y_{it} - x_{it}'\theta(c_i))^2
\end{equation}

We let $\samplerisk(\theta, \gamma)$ denote the sample risk in \ref{estimator}. 
There are many algorithms available for the least squares partitioning problem above\footnote{See, for instance, the discussion in BM Appendix S1.}.
One benchmark approach, known as Lloyd's algorithm (\cite{Lloyd1982}) in the setting of k-means clustering, takes a coordinate ascent approach to the problem in \ref{estimator}, alternately updating the parameters $\theta$ and assignments $\gamma$ until convergence. \\

\textbf{Lloyd's Algorithm} - Fix a division of the covariate vector $x_{it} = \xvec$ into blocks with $x_{it}^{\ell} \in \mathbb{R}^{d_{\ell}}$ and fix the number of clusters $\kvec = (k_1, \dots, k_{\numgroups})$ in each block.  
Our approach is a modification of Lloyd's algorithm for k-means clustering. 
We perform coordinate ascent on the sample objective $\samplerisk(\theta, \gamma)$ by alternating parameter updates and cluster assignment updates until convergence. \\  

(1) Randomly initialize parameters $\theta^1$ and cluster assignments $\gamma^1$. \\
(2) Given $\theta^s$, set $\gamma^{s + 1} = \argmin_{\gamma \in \Gamma} \samplerisk(\theta^s, \gamma)$. \\
(3) Given $\gamma^{s + 1}$, update $\theta^s \to \theta^{s + 1}$. \\
(4) Repeat (2) and (3) until convergence. \\

Since problem \ref{estimator} is not generally convex, we repeat steps (1) through (4) from different random initializations (in parallel), and take the estimate that achieves the lowest sample risk $\samplerisk$.
See appendix \ref{appendix_computation} for more discussion of the implementation of this algorithm and related computational issues. 

\section{Asymptotic Properties} \label{section:asymptotic}

In this section, we investigate the asymptotic properties of the estimator $(\wh{\theta}, \wh{\gamma})$ defined above as $N, T \to \infty$. 
In what follows, we assume the data is generated from the model \ref{problem} with $(\theta^0, \gamma^0)$ the true slope parameters and cluster assignment function. 
We let $\| \cdot \|$ denote the usual Euclidean norm. 

\subsection{Consistency}

\begin{assumption} \label{assumptions:consistency} 
We make the following assumptions
\begin{enumerate}[label={(\alph*)}, ref ={\ref*{assumptions:consistency}.(\alph*)}, itemindent=.5pt, itemsep=.5pt]
    \item For each $\ell \in [\numgroups]$, the parameter space $\Theta_{\ell} \subset \mathbb{R}^{d_{\ell} \times k_{\ell}}$ is compact \label{consistency:compact}
    \item $\frac{1}{NT^2} \sum_{i=1}^N \sum_{t=1}^T \sum_{s=1}^T e_{it} e_{is} x_{it}' x_{is} \convp 0$ \label{consistency:dependence}
    \item $\|\theta^0_{\ell a} - \theta^0_{\ell b} \| \equiv \distab > 0$ for each pair of clusters $a, b \in [k_{\ell}]$\label{consistency:separation}
    \item Define $M(c, c', \gamma) \equiv \frac{1}{NT} \sum_{i, t} x_{it} x_{it}' \mathds{1}(c_i^0 = c')\mathds{1}(c_i = c)$ and let $\rho(c, c', \gamma) \equiv \lambda_{min}(M(c, c', \gamma))$. 
    Then there exists $\delta > 0$ such that $\inf_{c', \gamma} \max_{c} \rho(c, c', \gamma) \geq \delta - o_P(1)$ as $N, T \to \infty$. \label{consistency:eigenvalue} 
\end{enumerate}
\end{assumption}

Assumption \ref{consistency:compact} is the usual parameter space compactness condition. 
Assumption \ref{consistency:dependence} can be seen as limiting the time-series dependence of errors and covariates, averaged over cross-sectional units.
Condition \ref{consistency:separation} ensures that the clusters within each grouping are non-identical. 
The final assumption \ref{consistency:eigenvalue} is the analogue in our setting of assumption S2(a). in BM. 
This condition is used to ensure curvature of the sample risk function $\wh{Q}$. 
If there is a common parameter ($\numgroups = 1$, $k = 1$), this is the usual non-collinearity condition for pooled panel regression. 
See section \ref{appendix:eigenvalue_discussion} in the appendix for further discussion, as well as section S4.2 in the supplementary appendix of BM.\\ 

\textbf{Cluster Label Ambiguity.} The minimizer $\argmin_{\gamma \in \Gamma, \theta \in \Theta} \samplerisk(\theta, \gamma)$ is only unique up to permutations of the labels in $\mathcal{C}$ and their associated parameter vectors in $\theta$.
Thus, the $c \in \mathcal{C}$ used to label estimated clusters in each block is arbitrary, and to resolve this ambiguity we need to fix a correspondence $\sigma_{\ell}: [k_{\ell}] \to [k_{\ell}]$ between true and estimated cluster parameters for each $\ell$\footnote{ Note that, for finite $T$, it can be the case that $\wh{\theta}(\wh{c}_i) \not = \wh{\theta}(\wh{c}_j)$, but $c_i^0 = c_j^0$, so the estimates $\wh{\theta}(\wh{c}_i)$ do not in general induce a well-defined estimator of any fixed cluster parameter $\theta^0_{\ell a}$.}.
Let
\begin{equation} \label{equation:permutation_function}
\mapping_{\ell}(a) \equiv \argmin_{b \in [k_{\ell}]} \|\wh{\theta}_{\ell b} - \theta^0_{\ell a} \|
\end{equation}

and define the estimator of the true parameter $\theta^0_{\ell a}$ to be $\wh{\theta}_{\ell \mapping({a})}$.
Note that this is infeasible without access to the true parameters $\theta^0$. 

\begin{lemma} \label{lemma:permutation}
Under the assumptions in \ref{assumptions:consistency}, $\mathbb{P}(\sigma_{\ell} \, \text{invertible}) \to 1$ as $N, T \to \infty$
\end{lemma}

By the lemma, we can relabel the estimated clusters $\wh{\theta}_{\ell \sigma(a)} \to \wh{\theta}_{\ell a}$, and this is well-defined w.h.p as $N, T \to \infty$.

\begin{thm} \label{thm:consistency}
Under the assumptions in \ref{assumptions:consistency}, for all groupings $\ell$ and $a \in [k_{\ell}]$, we have 
\[
\|\theta^0_{\ell a} - \wh{\theta}_{\ell a} \| = o_P(1) 
\]
equivalently
\[
\min_{x \in \clusterspace} \|\wh{\theta}(x) - \theta(c)\| = o_p(1) \quad \forall c \in \clusterspace
\]
as $N, T \to \infty$. 
\end{thm}

See section \ref{proof:consistency} for the proof of the theorem and lemma. 

\subsection{Asymptotic Equivalence}

In this section, we establish asymptotic equivalence of $\wh{\theta}$ to the infeasible oracle estimator with known clusters. 
We need the following assumptions in addition to those already stated in \ref{assumptions:consistency}.

\begin{assumption} \label{assumptions:inference} 
Make the following assumptions
\begin{enumerate}[label={(\alph*)}, ref ={\ref*{assumptions:inference}.(\alph*)}, itemindent=.5pt, itemsep=.5pt]
    \item $\frac{1}{NT^2} \sum_{i} \sum_{t, s} \|x_{it} \|^2 \|x_{is}\|^2 = O_P(1)$ \label{assumptions:inference:xnorm}
    \item Define $M_{NT}^c = \frac{1}{NT} \sum_{i, t} \one(c_i^0 = c) x_{it}x_{it}'$.
    Then there exists $\evalboundtwo > 0$ such that for all $a > 0$, this sequence of matrices satisfies $\min_{c \in \mathcal{C}} \lambda_{min}(M_{NT}^c) \convp \evalboundtwo$ as $N, T \to \infty$. \label{assumptions:inference:evals}
    \item There exist constants $b > 0$ and $d_1 > 0$ and sequence $\alpha(t) \leq e^{-bt^{d_1}}$ such that for all $i \in [N]$ $\{x_{it}\}_t$ and $\{x_{it} e_{it}\}_t$ are strongly mixing with coefficients $\alpha(t)$. \label{assumptions:inference:mixing}
    \item There exist constants $f>0$ and $d_2>0$ such that for all $i \in [N]$ and all $z > 0$, for all components $x_{it}^j$, $x_{it}^{j'}$ of the vector $x_{it}$ we have $\prob(|x_{it}^j x_{it}^{j'}  - E(x_{it}^j x_{it}^{j'}) | > z)$ and $\prob( |e_{it} x_{it}^j - E e_{it} x_{it}^j |  > z)$ are bounded above by $e^{1 - (z/f)^{d_2}}$. \label{assumptions:inference:tails} 
    \item The uniform limits $\max_{i \in [N]} \frac{1}{T} \sum_t E[e_{it} x_{it}] \to 0$ and $\min_{i \in [N]} \frac{1}{T} \sum_t \mathbb{E}(x_{it}'(\theta(c) - \theta(c')))^2 \to \clusterdist(c, c')$ hold as $T \to \infty$, and $d(c, c') \geq \clusterdistmin > 0$ for $c \not = c'$. \label{assumptions:inference:limits}
    \item There exists $M' > 0$ such that for all $a > 0$
    \[
    \max_{i \in [N]} \prob \left (\frac{1}{T} \sum_t \|x_{it} \|^2 > M' \right ) = o(T^{-a})
    \] \label{assumptions:inference:xbound}
\end{enumerate} 
\end{assumption}

We will show that $\wh{\theta}$ is asymptotically equivalent to the infeasible oracle estimator where true cluster membership $c_i^0$ is known for all $i$.  
Define the problem
\begin{align}
\qoracle(\theta) &\equiv \frac{1}{NT} \sum_{i = 1}^N \sum_{t = 1}^T (y_{it} - x_{it}' \theta(c_i^0))^2 \nonumber \\
\thetaoracle &= \argmin_{\theta \in \Theta} \qoracle(\theta) \label{equation:qoracle}
\end{align}

The following theorem shows that $\wh{\theta}$ and $\thetaoracle$ are asymptotically equivalent.
\begin{thm} \label{thm:inference}
Let the assumptions in \ref{assumptions:consistency} and \ref{assumptions:inference} hold. 
Then for any $a > 0$ and as $N, T \to \infty$, we have
\begin{equation} \label{equation:equivalence}
\wh{\theta} = \thetaoracle + o_P(T^{-a})
\end{equation}

Moreover, individual cluster estimates satisfy

\begin{equation} \label{equation:cluster_assignment}
\prob \left (\exists i \in [N] \, s.t. \, \wh{c}_i \not = c_i^0 \right ) = o(1) + o(NT^{-a})
\end{equation}
\end{thm}

See appendix \ref{proof:inference} for the proof. 
Because of this theorem, for asymptotic sequences with $N$ growing at a sub-polynomial rate relative to $T$, it suffices to characterize the asymptotic distribution of the estimator $\thetaoracle$. 

\subsection{Inference}
\textbf{Notation} - To aid the exposition, we start with a few definitions. 
For $A \in \mathbb{R}^{p \times q}$, let $\vect(A) \equiv ((A^1)', \dots, (A^q)')' \in \mathbb{R}^{pq}$. 
Thinking of $\theta = \{\theta_1, \dots, \theta_{\numgroups}\}$ as a collection of matrices $\theta_{\ell} \in \mathbb{R}^{d_{\ell} \times k_{\ell}}$, we denote $\vect(\theta) = (\vect(\theta_1)', \dots, \vect(\theta_{\numgroups})')' \in \mathbb{R}^{d_{\theta}}$, where $d_{\theta} \equiv \sum_{\ell} k_{\ell} d_{\ell}$ is the total dimension of $\vect(\theta)$. 
For $1 \leq \ell \leq \numgroups$ and $a \in [k_{\ell}]$, we use the block index convention that $\vect(\theta)_{\ell a}$ refers to the $d_{\ell}$ dimensional sub-vector in the $a^{th}$ position of the $\ell^{th}$ block. 
Using the notation above, for $1 \leq \ell, s \leq \numgroups$ and $a \in [k_{\ell}]$, $b \in [k_s]$ define $\covxmatrix \in \mathbb{R}^{d_{\theta} \times d_{\theta}}$ and $\cltvec \in \mathbb{R}^{d_{\theta}}$ by  
\begin{align}
\covxmatrix_{\ell a, s b} &= \frac{1}{NT} \sum_{i=1}^N \sum_{t=1}^T x_{it \ell} x_{its}' \one(c^0_{is} = b) \one(c^0_{i\ell} = a) \label{equation:covxmatrix} \\
\cltvec_{\ell a} &= \frac{1}{NT} \sum_{i=1}^N \sum_{t=1}^T y_{it} \one(c^0_{i \ell} = a) x_{it\ell} \\
w_{\ell a} &= \frac{1}{NT} \sum_{i=1}^N \sum_{t=1}^T e_{it} \one(c^0_{i \ell} = a) x_{it\ell}
\end{align}

\begin{prop}
The solution $\thetaoracle$ to problem \ref{equation:qoracle} satisfies
\begin{equation}
\covxmatrix \vect(\thetaoracle) = \cltvec
\end{equation}
\end{prop}
The proof follows by taking the first order conditions of \ref{equation:qoracle} and rearranging.
Note that the first order conditions $\nabla_{\theta_{\ell a}} \qoracle(\thetaoracle) = 0$ can potentially vary with all the other parameters $\theta_{sb}$ in the model (for $s \not = \ell$). 
Therefore, in contrast to the $\numgroups = 1$ case considered in the existing literature, the estimator $\thetaoracle$ is \emph{not} equivalent to simply running $k$ separate regressions over the partition of the cross-sectional units.  \\ 

Consider the following assumptions that allow us to characterize the asymptotic distribution of the infeasible $\thetaoracle$. 

\begin{assumption} \label{assumptions:clt} 
We make the following assumptions
\begin{enumerate}[label={(\alph*)}, ref ={\ref*{assumptions:clt}.(\alph*)}, itemindent=.5pt, itemsep=.5pt]
    \item There is a matrix $\Omega \succ 0$ such that for all $\ell, s \in [\numgroups]$ and $1 \leq a \leq k_{\ell}$, $1 \leq b \leq k_s$
    \begin{equation}
    \frac{1}{NT} \sum_{i, j = 1}^N \sum_{t, t' = 1}^T E[e_{it} e_{jt'} \one(c_{i \ell}^0 = a) \one(c_{js}^0 = b) x_{it\ell} x_{jt's}'] \to \Omega_{\ell a, s b} \label{equation:covariance}
    \end{equation}
    as $N, T \to \infty$. \label{assumptions:clt:omega} 
    \item $\mathbb{E}[e_{it} \one(c_{i \ell}^0 = a) x_{it \ell}] = 0$ for all $\ell, a$. \label{assumptions:clt:exogeneity}
    \item $\covxmatrix \convp M$ as $N, T \to \infty$, with $M \succ 0$. \label{assumptions:clt:prob}
    \item $\sqrt{NT} w \convd \mathcal{N}(0, \Omega)$ as $N, T \to \infty$. \label{assumptions:clt:clt}
\end{enumerate} 
\end{assumption}

\begin{thm} \label{thm:clt}
Suppose that the assumptions in \ref{assumptions:clt} are satisfied. 
Also suppose there is some $r > 0$ such that $\sqrt{N}T^{-r} = o(1) $ as $N, T \to \infty$. 
Then we have   

\begin{equation}
\sqrt{NT}(\vect(\wh{\theta} - \theta^0)) \convd \mathcal{N}(0, M \inv \Omega M)
\end{equation}
\end{thm}

The proof of this theorem is given in appendix \ref{proof:inference}. \\

Consider the case where cross-sectional units are independent, then under assumption \ref{assumptions:clt:exogeneity}, the terms in \ref{equation:covariance} with $i \not = j$ vanish. 
In this case, we propose the HAC estimators
\begin{align}
\wh{\Omega}_{\ell a, s b} &= \frac{1}{NT} \sum_{i = 1}^N \sum_{t, t' = 1}^T \wh{e}_{it} \wh{e}_{it'} \one(\wh{c}_{i \ell} = a) \one(\wh{c}_{is} = b) x_{it\ell} x_{it's}' \nonumber \\
\estvariance &=  \covxmatrix \inv \wh{\Omega} \covxmatrix \inv \label{equation:variance_estimator} 
\end{align}

where $\covxmatrix$ is as in equation \ref{equation:covxmatrix}.  
Variance estimators of this form were originally proposed in \cite{Arellano1987}, and their asymptotic theory for $N, T \to \infty$ jointly was first analyzed in \cite{Hansen2007Robust}. 
For further discussion on adapting the results of \cite{Hansen2007Robust} to our setting, see appendix \ref{proof:variance_estimation}. 

\subsection{Model Selection}

In this section we let $k^0 = (k_1^0, \dots, k_{\numgroups}^0)$ denote the true number of clusters and develop a Cp-like criterion to estimate $k^0$. 
We suppose that prior information can be used to bound the true number of clusters from above $k^0 \leq k_{max}$. 
So far, we have defined the sample risk $\samplerisk(\theta, \gamma)$ with domain the true parameter space i.e. $\theta \in \prod_{\ell} \mathbb{R}^{d_{\ell} \times k^0_{\ell}}$ and $\gamma: [N] \to \prod_{\ell} [k^0_{\ell}]$. 
However, note that $\samplerisk = \frac{1}{NT} \sum_{i, t} (y_{it} - x_{it}'\theta(\gamma(i)))^2$ only varies through $\theta(\gamma(i)) \in \mathbb{R}^p$. 
Thus, we can extend the domain of $\samplerisk$ to models with $k \not = k^0$, since $\theta(\gamma(i)) \in \mathbb{R}^p$ for any conformable $(\theta, \gamma)$.\footnote{Formally, let $\samplerisk: \bigcup_{k \geq 0} \left \{\prod_{\ell} \mathbb{R}^{d_{\ell} \times k_{\ell}} \times [\, [N] \to \prod_{\ell} [k^0_{\ell}]\, ]\right \} \to \mathbb{R}_{\geq 0}$ with $\samplerisk(\theta, \gamma) = \frac{1}{NT} \sum_{i, t} (y_{it} - x_{it}'\theta(\gamma(i)))^2$}\\

We slightly strengthen some assumptions above
\begin{assumption} \label{assumptions:mo} 
Impose the following assumptions
\begin{enumerate}[label={(\alph*)}, ref ={\ref*{assumptions:mo}.(\alph*)}, itemindent=.5pt, itemsep=.5pt]
    \item For all $c \in \clusterspace^{k_0}$, $\frac{1}{NT} \sum_{i, t} e_{it} x_{it} \one(\ctrue = c) = O_p(1/\sqrt{NT})$ \label{assumptions:mo:ci_uncorrelated}
    \item With $\rho(c, c', \gamma)$ defined as in assumption \ref{assumptions:inference:evals}, there exists $\delta > 0$ such that for all $k \geq k^0$ we have 
    \[
     \rho^k_{NT} \equiv \min_{c' \in \clusterspace^{k_0}} \min_{\gamma \in \Gamma^k} \max_{c \in \clusterspace^k} \rho(c, c', \gamma) \geq \delta - o_p(1)   
    \] \label{assumptions:mo:evals} 
    \item As $N, T \to \infty$ \label{assumption:mo:lambdamin} 
    \[
    \inf_{j \in [N]} \lambda_{min} \left (\frac{1}{T} \sum_t E[x_{jt} x_{jt}'] \right ) \to \eiglowerbound > 0
    \]
    \item For some $0 < \epsilon < \frac{1}{2} \wedge \frac{d_1 d_2}{d_1 + d_2}$ we have $\log N = o(T^{\epsilon})$ as $N, T \to \infty$ \label{assumptions:mo:epsilon}, where $d_1, d_2$ are the mixing and tail parameters defined in assumptions \ref{assumptions:inference:mixing} and \ref{assumptions:inference:tails}  
\end{enumerate} 
\end{assumption}

We can think of assumption \ref{assumptions:mo:ci_uncorrelated} as stating that a CLT holds for $e_{it} x_{it} \one(c_i^0 = c)$ for each $c \in \clusterspace^{k_0}$. 
This will be easiest to satisfy when $E[e_{it} x_{it} \one(c_i^0 = c)] = 0$, a stronger form of unconfoundedness. 
Assumption \ref{assumptions:mo:evals} is the extension of assumption \ref{consistency:eigenvalue} in our consistency proof to the case of models with misspecified number of clusters.
See section \ref{appendix:eigenvalue_discussion} of the supplementary appendix for a discussion of this condition. 
In the stationary case with identically distributed cross-sectional units, having $E[x_{it} x_{it}']$ of full rank is sufficient for assumption \ref{assumption:mo:lambdamin}. 
Finally, assumption \ref{assumptions:mo:epsilon} requires that $\log N$ is sub-polynomial in $T$ as $N, T \to \infty$. \\

\textbf{Information Criterion} - Let $(\wh{\theta}^k, \wh{\gamma}^k)$ be the minimizer of $\samplerisk$ with $(k_i)_i$ clusters and denote $\samplerisk(k) = \samplerisk(\wh{\theta}^k, \wh{\gamma}^k)$.  
Then we define the $\cp$ criterion
\begin{align} \label{equation:cp}
\cp(k) &\equiv \samplerisk(k) + f(N, T) \sum_i k_i \\
\wh{k} &= \argmin_{k \leq k_{\text{max}}} \cp(k) \nonumber 
\end{align}

We have the following result on consistency of model selection

\begin{thm} \label{thm:cp}
Suppose the assumptions in \ref{assumptions:mo} hold. 
Let $f(N, T)$ be such that $f(N, T) \to 0$ and for some $\epsilon$ as in assumption \ref{assumptions:mo:epsilon}, $f(N, T) T^{1-3\epsilon} \to \infty$ as $N, T \to \infty$.
Then 
\[
\prob(\wh{k} = k^0) \to 1
\]
as $N, T \to \infty$
\end{thm}

For the proof, see appendix \ref{proof:cp}. 

\begin{remark}[Choice of $f$]
While any function $f(N, T)$ satisfying the conditions of the theorem will give asymptotically consistent model selection, the choice of $f$ will significantly affects finite sample performance. 
To put $f(N, T)$ on the same scale as $\samplerisk(k)$, we use $f(N, T) = \wh{\sigma}^2 g(N, T)$ in our simulations, where $\wh{\sigma}^2$ is a consistent estimate of the long run variance $\lim_{N, T} \frac{1}{NT} \sum_{i, t} E[e_{it}^2]$. 
By lemma \ref{lemma:cp:risk_order} in the appendix, $\wh{\sigma}^2 \equiv \samplerisk(k_{max})$ is such a consistent estimator.  
We find good performance with $g(N, T) = \frac{\log T}{T}$ in our simulations. 
Alternatively, e.g. $g(N, T) = \frac{\log T}{T^{1-\epsilon'}}$ for small $\epsilon'$ can be used to be technically consistent with the theory. 
\end{remark}

\section{Overspecification of $k$} \label{section:overspecification}
In this section, we report new results on the performance of k-means style estimators with an over-specified number of clusters. 
The work in this section builds on and sharpens the results in \cite{LiuOverspecifiedGroups} for the case of linear regression. 
Our proof of model selection consistency in Theorem \ref{thm:cp} heavily relies on the following result. 

\begin{thm} \label{thm:overspecified}
Suppose the assumptions in \ref{assumptions:consistency} hold. 
Then 
\begin{align}
\sup_{i \in [N]} \|\thetak(\wh{c}^k_i) - \theta^0(c_i^0)\|^2 &= o_p(T^{-1 + 4 \epsilon}) \\
\min_{x \in \clusterspace^k} \|\thetak(x) - \theta^0(c)\|^2 &= o_p(T^{-1 + 3 \epsilon}) \\
\frac{1}{N} \sum_i (\thetak(\wh{c}_i^k) - \theta^0(c_i^0))^2 &= o_p(T^{-1 + 3 \epsilon}) \label{equation:overspecified_thm:avg_loss}
\end{align}
\end{thm}
\begin{proof}
Follows from Proposition \ref{prop:mo:maxmin_rate}, Proposition \ref{prop:worst_case_rate}, and Corollary \ref{cor:mo:rates} in the appendix. 
\end{proof}

\begin{remark}
The preceding result can be compared with \cite{LiuOverspecifiedGroups} Theorem 1 and Lemma 5.16, which give $o_p(1)$, $o_p(1)$, and $o_p(T^{\frac{-1}{2(1 + d)}})$ rates (respectively) for each of the losses above, with $d = \frac{d_1 d_2}{d_1 + d_2}$. 
Note that their setting is a more general model of clustered M-estimation. 
Our rate improvements come from (1) optimizing the Fuk-Nagaev inequality in \cite{Rio2011} for our purposes (see lemma \ref{lemma:mo:deviations}) and (2) an inductive strategy that allows us to ``boost'' $O_p(T^{-1/4})$ rates arbitrarily close to $O_p(T^{-1/2})$. 
For a description of this approach, see lemma \ref{lemma:mo:recursion} as well as the propositions and corollary referenced above. 
\end{remark}

\begin{remark}
The rate established in equation \ref{equation:overspecified_thm:avg_loss} above is used to bound the magnitude of over-fitting for estimators with $k > k^0$, as in the second part of corollary \ref{cor:mo:rates}.  
A result of this form is necessary to determine the complexity penalty in \ref{thm:cp}. 
In particular, in contrast to the result in \cite{LiuOverspecifiedGroups}, theorem \ref{thm:overspecified} gives feasible rates for $f(T)$ that do not depend on mixing parameters and tail bounds of $e_{it}$ and $x_{it}$, which may be difficult or impossible to estimate. 
\end{remark}

\begin{remark}
Difficulty obtaining the fast rate $\samplerisk(k) - \baserisk = O_p(\frac{1}{NT})$ for $k > k^0$ suggests that over-fitting may be severe under over-specification of $k$. 
Difficulty obtaining $\sqrt{NT}$-consistency of $\thetak$ when $k > k^0$ suggests a type of incidental parameter problem. 
In fact, in the linear case it is known\footnote{See the example in BM, appendix S3.1.} that under $N \to \infty$, finite $T$ asymptotics, estimators with $k > k^0$ can suffer a bias of order up to $\frac{1}{\sqrt{T}}$.
\end{remark}

\section{Monte Carlo Simulations} \label{section:monte_carlo}

In this section, we describe the results of our Monte Carlo simulations.
All tables are reported in section \ref{section:tables} of the appendix.
Throughout, we denote
\begin{enumerate}[itemsep=0.4pt]
    \item $\text{Param. MSE} = \frac{1}{N} \sum_{i=1}^N \|\estthetaci - \truethetaci\|^2$ 
    \item $\text{Function MSE} = \frac{1}{NT} \sum_{i=1}^N \sum_{t=1}^T \|\estthetaci'x_{it} - \truethetaci'x_{it}\|^2$ 
    \item $\text{Cluster Loss} = \frac{1}{N} \sum_{i=1}^N \one(\wh{c}_i \not = c_i^0)$
\end{enumerate}

We use two specifications for the joint distribution $(x_{it}, y_{it})$. 
(1) Specifications labeled \emph{AR(1)} take $e_{it} \sim \text{AR}(1)$ and $x_{it} \sim \text{VAR}(1)$.
The AR process $e_{it}$ has normal innovations, and $x_{it}$ has multivariate normal innovations with constant, diagonal covariance matrix. 
The respective autocorrelation parameters are $\rho_e = 0.3$, $\rho_x = 0.5$. 
(2) Specifications labeled \emph{HK} use a heteroskedastic design inspired by \cite{Hansen2007Robust}. 
With $x_{it}$ as above, we use $e_{it} = \rho e_{it-1} + v_{it} \cdot \sqrt{\frac{1}{2} + \frac{\|x_{it}\|^2}{2p}}$, with independent normal innovations $v_{it}$.
Innovation variances are normalized so that $\var(e_{it}) =1$ for all designs and all $(i, t)$. 
All simulations use $500$ independent samples. 
See appendix \ref{appendix_computation} for additional details on the computational specification. 

\subsection{Estimator Performance}

\begin{design}[Cluster Separation]
We let $p=4$, $k = (2, 2)$, $\numgroups=2$ and parameters
\[
\theta_1^0 = \begin{pmatrix} 1 & \cos\alpha \\ 0 & \sin\alpha \end{pmatrix} \quad \theta_2^0 = \begin{pmatrix} 0 & -\sin\alpha \\ 1 & \cos\alpha \end{pmatrix} 
\]
where the columns of $\theta_{\ell}$ list parameters of block $1 \leq \ell \leq \numgroups$. 
Thus, as $\alpha \to 0$, the cluster parameters in each block rotate towards each other. 
Cluster estimation accuracy radically worsens for small $\alpha$. 
Coverage is around $80-90\%$ for well-separated clusters. 
As $\alpha \to 0$, our confidence intervals do not account for variation due to cluster estimation, and coverage is poor. 
Parameter loss is inverse U-shaped in cluster separation $\alpha$. 
For small $\alpha$, classification of $c_i^0$ becomes worse, giving large losses on some units. 
For $\alpha$ near $0$, misclassification contributes less to parameter loss since the cluster centers are very close. 
Results are shown in Table \ref{table:separation}. 
\end{design}

\begin{design}[Sample Size $(N, T)$]
We use the specification in the simulation above with $\alpha=\frac{\pi}{2}$.
Cluster loss is quite insensitive to $N$, in line with the theory. 
Increasing $T$ has a much larger effect than $N$ on parameter loss and coverage. 
For $T=5$, we find coverage actually decreases with larger $N$, which could be an example of the over-fitting issue discussed in section \ref{section:overspecification}. 
Results are shown in Table \ref{table:nt}. 
\end{design}

\begin{design}[Number of Clusters] \label{design:number}
Again with $p=4$ and $\numgroups=2$, we let $k=(k_1, k_2)$ vary. 
We define clusters $\theta_{1a} = (\cos(\frac{2 \pi}{5} \cdot a), \sin(\frac{2 \pi}{5} \cdot a))'$ for $1 \leq a \leq k_1$ and similarly for the second block. 
All performance measures decrease as the number of clusters increase. 
Results are shown in Table \ref{table:clusters}. 
\end{design}

\begin{design}[Misspecification]
In this simulation, we repeat the design above using $B=1$ and $k = k_1 \cdot k_2$, the minimal number of clusters for consistent estimation using the single latent type assumption ($B=1$) considered in the literature. 
As expected, there is a significant power loss. 
Results are shown in Table \ref{table:misspec}. 
\end{design}

\begin{design}[Block Dimension Imbalance]
We let $p=12$, $B=2$ and $(k_1, k_2) = (2, 2)$. 
We vary the grouping of covariates, taking the first block to be $(x_{it1}, \dots, x_{itm})$ for $m \in \{1, \dots, 6\}$. 
Coverage-small denotes the average coverage for parameters belonging to the small block $(x_{it1}, \dots, x_{itm})$ and conversely for Coverage-large.
Cluster loss-large and Cluster loss-small are defined similarly. 
Classification and coverage are worse for the block of smaller dimension $m$ when $m / (p-m)$ is very small, but quickly equalize as $m$ gets larger.  
Results are shown in Table \ref{table:imbalance}. 
\end{design}

\begin{design}[Covariate Dimension]
In this simulation, we take $B = p$ (one latent variable for each covariate) and study the effect of increasing $p$. 
We let $k_{\ell} = 2$ for $1 \leq \ell \leq p$ and clusters $\theta_{\ell a} = \pm 1$ ($a \in \{1, 2\}$). 
Performance only slightly deteriorates as $p$ increases.
Results are shown in Table \ref{table:dimension}
\end{design}

\subsection{Model Selection}

\begin{design}[Model Selection - Number of Clusters]
We implement the $\cp$ criterion and study its performance on the DGP in design \ref{design:number} above. 
We use penalty sequence $f(N, T) = \wh{\sigma}^2 \frac{\log T}{T}$, as in section \ref{section:model}.
Model loss is calculated using average (over $k_1, \dots k_{\numgroups}$) distance from the truth $\|\wh{k} - k^0\|_1 / \numgroups$. 
We set $k_{max} = (6, 6)$ and use $200$ independent samples. 
For $k^0=(2, 3)$, we estimate $\mathbb{E}\frac{\|\wh{k}-k^0\|}{\numgroups} = 0.03$. 
For $k^0 = (4, 4)$, we find $\mathbb{E}\frac{\|\wh{k}-k^0\|}{\numgroups} = 0.73$, with all estimates $\wh{k} = (3, 3), (3, 4)$, or $(4, 3)$.
We view this performance as reasonable given that the clusters are quite close in this design, though the results suggest we may be slightly over-penalizing. 
\end{design}

\section{Conclusion}
Clustering methods have recently become popular as a way of modeling limited heterogeneity in panel data. 
This paper motivates a family of panel structures, nesting the standard regression clustering model, that have significant cross-sectional homogeneity but are nevertheless ill-suited to estimation by the clustering methods currently considered in the literature 
We propose a modified procedure that simultaneously clusters on multiple discrete latent types, significantly expanding the set of panel structures that can be accommodated by these methods. 
We employ Lloyd's algorithm to compute the estimator, and give consistency and asymptotic normality results for the resulting estimates. 

\clearpage

\appendix
\section{Proofs} \label{section:proofs}

Throughout the following proofs, unless otherwise specified $\max_i$, $\sum_i$, $\sum_t$, $\sum_{c}$ denote $\max_{i \in [N]}$, $\sum_{i \in [N]}$, $\sum_{t \in [T]}$, $\sum_{c \in \mathcal{C}}$ respectively. 

\subsection{Proof of Theorem \ref{thm:consistency}} \label{proof:consistency}
We define 
\begin{equation} \label{equation:qtilde_def}
\qtilde(\theta, \gamma) = \frac{1}{NT} \sum_{i=1}^N \sum_{t=1}^T (x_{it}'(\theta(c_i) - \theta^0(c_i^0))^2 + \sum_{i=1}^N\sum_{t=1}^T e_{it}^2
\end{equation}

and recall that
\[
\wh{Q}(\theta, \gamma) = \frac{1}{NT} \sum_{i=1}^N \sum_{t=1}^T (x_{it}'(\theta^0(c_i^0) - \theta(c_i)) + e_{it})^2
\]

we begin by showing uniform convergence of risk surfaces

\begin{lemma} \label{unif:conv}
$\sup_{\theta \in \Theta, \gamma \in \Gamma} \left [\wh{Q}(\theta, \gamma) - \qtilde(\theta, \gamma) \right ] \convp 0$ as $N, T \to \infty$ 
\end{lemma} 

\begin{proof}
Define $\Delta \theta_i = \theta^0(c_i^0) - \theta(c_i)$ and note that 
\[
\wh{Q}(\theta, \gamma) - \qtilde(\theta, \gamma) = \frac{2}{NT}\sum_{i, t} e_{it} x_{it}' \Delta \theta_i
\]
Then we can compute
\begin{align}
\left (\wh{Q}(\theta, \gamma) - \qtilde(\theta, \gamma) \right)^2 &= \left [\frac{1}{N} \sum_i \Delta \theta_i' \left ( \frac{1}{T} \sum_t e_{it} x_{it} \right ) \right]^2 \leq \frac{1}{N} \sum_i \| \Delta \theta_i \|^2 \frac{1}{T^2} \left \| \sum_t e_{it} x_{it} \right \|^2  \nonumber \\
& \lesssim \frac{1}{NT^2} \sum_i \sum_{t, s} e_{it} e_{is} x_{it}' x_{is} = o_P(1) \label{equation:cross_term_bound}
\end{align}

The first inequality follows from Jensen and Cauchy-Schwarz, the next uses assumption \ref{consistency:compact} on compactness, and the final equality is assumption \ref{consistency:dependence}.  
Taking $\sup_{\theta \in \Theta, \gamma \in \Gamma}$ on both sides of the inequality gives the statement of the lemma. 
\end{proof}

Now we make the usual observation that $\qtilde(\wh{\theta}, \wh{\gamma}) - \qtilde(\theta^0, \gamma^0) = o_P(1)$ since
\begin{align*}
\qtilde(\wh{\theta}, \wh{\gamma}) &= \wh{Q}(\wh{\theta}, \wh{\gamma}) + [\qtilde(\wh{\theta}, \wh{\gamma}) - \wh{Q}(\wh{\theta}, \wh{\gamma})] \leq \wh{Q}(\wh{\theta}, \wh{\gamma}) + \sup_{\theta \in \Theta, \gamma \in \Gamma} \left [(\qtilde - \wh{Q})(\theta, \gamma) \right ]\\
&= \wh{Q}(\wh{\theta}, \wh{\gamma}) + o_P(1) \leq \wh{Q}(\theta^0, \gamma^0) + o_P(1) = \qtilde(\theta^0, \gamma^0) + o_P(1) \\
& \implies 0 \leq \qtilde(\wh{\theta}, \wh{\gamma}) - \qtilde(\theta^0, \gamma^0) \leq o_P(1)
\end{align*}

The second equality follows from lemma \ref{unif:conv}, and the third equality from the definition of the estimator.
The next step is to show curvature of the auxiliary sample risk $\qtilde$.
The following curvature calculation is almost identical to the proof in appendix S6 in BM. 
For arbitrary $\theta \in \Theta$ and $\gamma \in \Gamma$, we have 
\begin{align}
\qtilde(\theta, \gamma) - \qtilde(\theta^0, \gamma^0) &= \frac{1}{NT} \sum_{i, t} (\thetaci - \truethetaci)'x_{it}x_{it}'(\thetaci - \truethetaci)  \label{equation:consistency:curvature} \\
&= \frac{1}{NT} \sum_{i, t} \sum_{c \in \mathcal{C}} \sum_{c' \in \mathcal{C}} (\theta(c) - \theta^0(c'))'x_{it}x_{it}'(\theta(c) - \theta^0(c')) \mathds{1}(c_i^0 = c')\mathds{1}(c_i = c) \nonumber \\
&= \sum_{c, c'} (\theta(c) - \theta^0(c'))' \left ( \frac{1}{NT} \sum_{i, t} x_{it} x_{it}' \mathds{1}(c_i^0 = c')\mathds{1}(c_i = c) \right )(\theta(c) - \theta^0(c')) \nonumber
\end{align}

Let $M(c, c', \gamma) \equiv \frac{1}{NT} \sum_{i, t} x_{it} x_{it}' \mathds{1}(c_i^0 = c')\mathds{1}(c_i = c)$ and define $\rho(c, c', \gamma) \equiv \lambda_{min}(M(c, c', \gamma))$. 
Then the last line is bounded below by 
\begin{align*}
 \sum_{c, c'} \|\theta(c) - \theta^0(c')\|^2 \rho(c, c', \gamma) & \geq \sum_{c, c'} \rho(c, c', \gamma) \inf_{x \in \mathcal{C}} \|\theta(x) - \theta^0(c')\|^2 \\
 & \geq \sum_{c'} \inf_{\Tilde{c}, \gamma} \max_{c} \rho(c, \Tilde{c}, \gamma) \inf_{x \in \mathcal{C}} \|\theta(x) - \theta^0(c')\|^2 \\
 & \geq \left (\inf_{\Tilde{c}, \gamma} \max_{c} \rho(c, \Tilde{c}, \gamma) \right )\max_{c'} \inf_{x \in \mathcal{C}} \|\theta(x) - \theta^0(c')\|^2 
\end{align*}

Then we see that 
\begin{align*}
o_P(1) &= \qtilde(\wh{\theta}, \wh{\gamma}) - \qtilde(\theta^0, \gamma^0) \geq \left (\inf_{\Tilde{c}, \gamma} \max_{c} \rho(c, \Tilde{c}, \gamma) \right )\max_{c'} \inf_{x \in \mathcal{C}} \|\wh{\theta}(x) - \theta^0(c')\|^2 \\
&\geq \delta \max_{c} \inf_{x \in \mathcal{C}} \|\wh{\theta}(x) - \theta^0(c)\|^2 - o_P(1) \\
& \implies \max_{c} \inf_{x \in \mathcal{C}} \|\wh{\theta}(x) - \theta^0(c)\|^2 = o_P(1) 
\end{align*}

So that $ \max_{c} \inf_{x \in \mathcal{C}} \|\wh{\theta}(x) - \theta^0(c)\|^2 = o_P(1)$.
The second equality uses assumptions \ref{consistency:eigenvalue} and \ref{consistency:compact}.
As noted in the main text, the problem $\argmin_{\gamma \in \Gamma, \theta \in \Theta} \samplerisk(\theta, \gamma)$ is invariant to permutations of the labels in $\mathcal{C}$ and their associated parameter vectors in $\theta$.
The next step is resolve this degeneracy by giving a well-defined estimator of $\theta^0_{\ell a}$ for each $\ell \in [\numgroups]$, $a \in [k_{\ell}]$. \\

\begin{lemma} \label{lemma:permutation_proof}
Define $\sigma(c) \equiv \argmin_{x \in \mathcal{C}} \|\wh{\theta}(x) - \theta^0(c) \|^2$. 
The map $\sigma_{\ell}(x) \equiv \sigma(c)_{\ell}$ for any $c$ such that $c_{\ell} = x$ is well defined. 
\end{lemma}
\begin{proof}
Existence of the map is clear; we show it is a function.
Note that $\min_{x \in \mathcal{C}} \|\wh{\theta}(x) - \theta^0(c) \|^2 = \sum_{\ell = 1}^{\numgroups} \min_{x_{\ell} \in [k_{\ell}]} \|\wh{\theta}_{\ell}(x_{\ell}) - \theta_{\ell}^0(c_{\ell}) \|^2 $. 
Then for any $c \in \mathcal{C}$, we have $\sigma(c)_{\ell} = f(c_{\ell}, \theta^0, \wh{\theta})$, so $c_{\ell} = c'_{\ell} \implies \sigma(c)_{\ell} = \sigma(c')_{\ell}$. 
\end{proof}

In fact, since $\min_{x \in \mathcal{C}} \|\wh{\theta}(x) - \theta^0(c) \|^2 = o_P(1)$, the proof of the lemma above shows that $\|\theta_{\ell a}^0 - \wh{\theta}_{\ell \sigma_{\ell}(a)}\| = o_P(1)$ for all groupings $\ell$ and each cluster $a \in [k_{\ell}]$, completing the main statement of the theorem. \\

We show that for each $\ell$, $\sigma_{\ell}$ is a bijection w.h.p.
Since $\sigma_{\ell}: [k_{\ell}] \to [k_{\ell}]$, it suffices to show injection. 
Let $a, b \in [k_{\ell}]$, then we have
\[
\|\theta^0_{\ell a} - \theta^0_{\ell b}\| \leq  \|\theta^0_{\ell a} - \wh{\theta}_{\ell \sigma_{\ell}(a)}\| + \|\wh{\theta}_{\ell \sigma_{\ell}(a)} - \wh{\theta}_{\ell \sigma_{\ell}(b)}\| + \|\wh{\theta}_{\ell \sigma_{\ell}(b)} - \theta^0_{\ell}(b)\| \\
 \leq \|\wh{\theta}_{\ell \sigma_{\ell}(a)} - \wh{\theta}_{\ell \sigma_{\ell}(b)}\| + X_{N, T}
\]

Where $X_{N, T} = o_P(1)$ as $N, T \to \infty$. 
Then $\{\sigma_{\ell}(a) = \sigma_{\ell}(b) \implies a = b \} \subset \{X_{N, T} < \distab\}$ by assumption \ref{consistency:separation}, and the latter event has probability going to $1$ as $N, T \to \infty$. 
Since $\cap_{\ell} \{\sigma_{\ell} \, \text{injective}\}$ is an intersection of finitely many events of the form above, we have 

\[
\mathbb{P}(\sigma_{\ell} \, \text{injective} \, \forall \ell) \to 1
\]

as $N, T \to \infty$. 

\subsection{Proof of Theorem \ref{thm:inference}} \label{proof:inference} 

The proof closely follows the strategy used in \cite{BonhommeManresa2015}. 
We define the problem

\begin{equation} \label{equation:conc_problem}
\qconc(\theta) = \inf_{\gamma \in \Gamma} \wh{Q}(\theta, \gamma)
\end{equation}

And let $\wh{c}_i(\theta)$ denote the cluster assignments that minimize the RHS of \ref{equation:conc_problem}. 
Thus, $\qconc$ is the original problem from \ref{estimator} with the cluster assignments concentrated out. 
The proof of theorem \ref{thm:inference} crucially relies on the following lemma

\begin{lemma} \label{lemma:cluster_errors}
For $\eta > 0$, define $\neighborhood_{\eta} = \{\theta \in \Theta: \max_{c \in \mathcal{C}} \|\theta(c) - \theta^0(c) \| < \eta \}$.
Then there exists $\eta > 0$ such that for all $a > 0$

\[
\sup_{\theta \in \neighborhood_{\eta}} \frac{1}{N} \sum_i^N \mathds{1}(\wh{c}_i(\theta) \not = c_i^0) = o_P(T^{-a})
\]

\end{lemma} 

\begin{proof}

First note that for each $i \in [N]$
\begin{align*}
\mathds{1}(\wh{c}_i(\theta) \not = c_i^0) &= \sum_{c \not = c_i^0} \one(\wh{c}_i(\theta) = c) \leq \sum_{c \not = c_i^0} \one \left (\frac{1}{T} \sum_{t=1}^T (y_{it} - x_{it}'\theta(c))^2 \leq \frac{1}{T} \sum_{t =1}^T (y_{it} - x_{it}'\theta(c_i^0))^2 \right) \\
& = \sum_{c \not = c_i^0} \one \left (\frac{1}{T} \sum_{t=1}^T (x_{it}'(\theta^0(c_i^0) - \theta(c)) + e_{it})^2 \leq \frac{1}{T} \sum_{t =1}^T (x_{it}'(\theta^0(c_i^0) - \theta(c_i^0)) + e_{it})^2 \right ) \\ 
&\leq \sum_{c \in \mathcal{C}} \max_{c' \not = c} \one \left (\frac{1}{T} \sum_{t=1}^T (x_{it}'(\theta^0(c') - \theta(c)) + e_{it})^2 \leq \frac{1}{T} \sum_{t =1}^T (x_{it}'(\theta^0(c') - \theta(c')) + e_{it})^2 \right ) \\ 
&\equiv \sum_{c \in \mathcal{C}} \max_{c' \not = c} Z_{ic}(c', \theta)
\end{align*}

We can rewrite inequality inside the indicator as (for $c \in \mathcal{C}$) as 
\begin{align*}
\boundfn(\theta) \equiv \frac{1}{T}\sum_{t=1}^T 2e_{it} x_{it}' (\theta(c') - \theta(c)) + [x_{it}'(\theta^0(c') - \theta(c))]^2 - [x_{it}'(\theta^0(c') - \theta(c'))]^2 \leq 0 
\end{align*}

Then we calculate
\begin{align}
|\boundfn(\theta) &- \boundfn(\theta^0)| \leq \left |\frac{1}{T}\sum_{t=1}^T 2e_{it} x_{it}' (\theta(c') - \theta^0(c') + \theta^0(c) - \theta(c)) \right | \nonumber \\ 
&+ \left | \frac{1}{T} \sum_{t = 1}^T [x_{it}'(\theta^0(c') - \theta(c))]^2 - [x_{it}'(\theta^0(c') - \theta(c'))]^2 - [x_{it}'(\theta^0(c') - \theta^0(c))]^2 \right | \label{equation:boundfn}
\end{align}

The second term is bounded above by
\begin{align*}
\left | \frac{1}{T} \sum_{t = 1}^T [x_{it}'(\theta^0(c') - \theta(c'))]^2 \right | &+ \left | \frac{1}{T} \sum_{t = 1}^T [x_{it}'(\theta^0(c') - \theta^0(c) + \theta^0(c) - \theta(c))]^2 - [x_{it}'(\theta^0(c') - \theta^0(c))]^2 \right | \\ 
&\leq \left | \frac{1}{T} \sum_{t = 1}^T [x_{it}'(\theta^0(c') - \theta(c'))]^2 \right | + \left | \frac{1}{T} \sum_{t = 1}^T [x_{it}'(\theta^0(c) - \theta(c))]^2 \right | \\
&+ 2 \left | \frac{1}{T} \sum_{t = 1}^T [x_{it}'(\theta^0(c') - \theta^0(c))] [x_{it}'(\theta^0(c) - \theta(c))] \right |
\end{align*}

Using $\theta \in \neighborhood_{\eta}$ and applying the triangle inequality, Cauchy-Schwarz, and assumption \ref{consistency:compact}, the last expression can be bounded above by
\begin{align*}
2 \eta^2 \frac{1}{T}\sum_t \|x_{it}\|^2 + 2M\eta \frac{1}{T}\sum_{t=1}^T \|x_{it}\|^2 = 2 \eta(M + \eta) \frac{1}{T}\sum_{t=1}^T \|x_{it}\|^2 \leq  4 \eta M \frac{1}{T}\sum_{t=1}^T \|x_{it}\|^2 
\end{align*}

Similarly, one can show that the first term in \ref{equation:boundfn} is bounded by $4 \eta \left \| \frac{1}{T} \sum_{t = 1}^T e_{it} x_{it} \right \|$.

This shows that for any $c \not = c'$ 
\begin{align*}
\sup_{\theta \in \neighborhood_{\eta}} &Z_{ic}(c', \theta) \leq \sup_{\theta \in \neighborhood_{\eta}} \one(\boundfn(\theta^0) \leq |\boundfn(\theta) - \boundfn(\theta^0) |) \\
&\leq \one \left (\boundfn(\theta^0) \leq 4 \eta M \frac{1}{T}\sum_t \|x_{it}\|^2 + 4 \eta \left \| \frac{1}{T} \sum_{t = 1}^T e_{it} x_{it} \right \| \right ) \\
&= \one \left ( \frac{1}{T}\sum_{t=1}^T 2e_{it} x_{it}' (\theta^0(c') - \theta^0(c)) + [x_{it}'(\theta^0(c') - \theta^0(c))]^2 \leq 4 \eta M \frac{1}{T}\sum_t \|x_{it}\|^2 + 4 \eta \left \| \frac{1}{T} \sum_{t = 1}^T e_{it} x_{it} \right \| \right ) \\
&\leq \one \left ( \frac{1}{T}\sum_{t=1}^T [x_{it}'(\theta^0(c') - \theta^0(c))]^2 \leq 4 \eta M \frac{1}{T}\sum_t \|x_{it}\|^2 + (4 \eta + 2M) \left \| \frac{1}{T} \sum_{t = 1}^T e_{it} x_{it} \right \| \right )
\end{align*}
Where $\text{Diam}(\Theta) \leq M$ by assumption \ref{consistency:compact}.
Let $M'$ be the constant from \ref{assumptions:inference:xbound}
Then taking expectations
\begin{align} 
& \max_{i} \mathbb{E} \sup_{\theta \in \neighborhood_{\eta}} Z_{ic}(c', \theta) \nonumber \\
&\leq  \max_{i} \prob \left ( \frac{1}{T}\sum_{t=1}^T [x_{it}'(\theta^0(c') - \theta^0(c))]^2 \leq 4 \eta M M' + (4 \eta + 2M) \eta \right ) \nonumber \\
&+ \max_{i}\prob \left ( \frac{1}{T}\sum_t \|x_{it}\|^2 > M' \right ) +  \max_{i} \prob \left (  \left \| \frac{1}{T} \sum_{t = 1}^T e_{it} x_{it} \right \| > \eta \right ) \label{equation:zbound} 
\end{align}

To bound these terms we will use Lemma B.5 from BM, which is an application of \cite{Rio2017}, on concentration of strongly mixing sequences. 
We restate the lemma here 

\begin{lemma}[BM Lemma B.5] \label{lemma:bonhomme}
Let $z_t$ be a strongly mixing process with zero mean, with strong mixing coefficients $\alpha(t)$ satisfying \ref{assumptions:inference:mixing} and tails $\prob(|z_t| > z) \leq e^{1 - (z/f)^{d_2}}$. 
Then for all $a > 0$ and $z > 0$, we have as $T \to \infty$

\[
T^a \cdot \prob \left ( \left | \frac{1}{T} \sum_{t=1}^T z_t \right  | > z \right ) \leq r(T) = o(1)  
\]

Moreover, the function $r$ only depends on the constants $b, f, d_1, d_2$ from assumption \ref{assumptions:inference:mixing} and \ref{assumptions:inference:tails}.
\end{lemma}

We want to apply this result to the terms in \ref{equation:zbound} above. 
Observe that if $\{x_{it}\}$ is strongly mixing with mixing coefficients $\alpha(t)$ then $\{(x_{it}'(\theta^0(c) - \theta^0(c')))^2\}$ is also strongly mixing with coefficients uniformly bounded above by $\alpha(t)$. 
This follows because continuous transformations can only decrease the mixing coefficients.
For completeness, we can show that the tail assumptions in \ref{assumptions:inference:tails} imply that that $z_t \equiv (x_{it}'(\theta(c) - \theta(c')))^2 - E(x_{it}'(\theta(c) - \theta(c')))^2$ also satisfies the tail bound required in the lemma. 
Let $\Delta \theta \equiv \theta^0(c) - \theta^0(c')$ and recall $p = \text{dim}(x_{it})$, then 
\begin{align}
\prob((x_{it}'(\theta(c) - \theta(c')))^2 &- E(x_{it}'(\theta(c) - \theta(c')))^2 > z) = \prob(\Delta \theta'(x_{it}x_{it}' - Ex_{it}x_{it}') \Delta \theta > z) \nonumber \\
&= \prob \left (\sum_{j, j'} \Delta \theta^j \Delta \theta^{j'}(x_{it}^k x_{it}^{k'} - Ex_{it}^k x_{it}^{k'}) > z \right ) \nonumber \\
&\leq \sum_{k, k'} \prob \left (|(x_{it}^k x_{it}^{k'} - Ex_{it}^k x_{it}^{k'})| > \frac{z}{p^2 (M')^2} \right ) \label{equation:xsquare_tail} 
\end{align}

Note that $\prob(|Z| > z) \leq e^{1 - (z/f)^{d_2}}$ does not imply that $C \cdot \prob(|Z| > z)$ satisfies a tail bound of the same form (possibly with different constants $f, d_2$) if $C > 1$. 
However, a calculation shows that for any $C > 1$, there exist $f', d_2'$ such for all $z > 0$, $\min(1, Ce^{1 - (z/f)^{d_2}}) \leq \min(1, e^{1 - (z/f')^{d_2'}})$, so this is not a problem. 
This shows that the final term in \ref{equation:xsquare_tail} above satisfies a tail bound of the required form. \\

We now apply the lemma to each of the terms in equation \ref{equation:zbound}. 
Choose $\eta$ such that $4 \eta M M' + (4 \eta + 2M) \eta < \frac{1}{3}\clusterdistmin$.
Let $g_{it} \equiv E(x_{it}'(\theta^0(c') - \theta^0(c)))^2$ and $T'$ such that $\min_i \frac{1}{T} \sum_{t=1}^{T'} g_{it} > (1/2) \clusterdistmin$, using assumption \ref{assumptions:inference:limits}. 
Then for $T > T'$, the first term in $\ref{equation:zbound}$ is 
\begin{align*}
\max_i \prob \left ( \frac{1}{T}\sum_t \left ([x_{it}'(\theta^0(c') - \theta^0(c))]^2 - g_{it} \right ) \leq 4 \eta M M' + (4 \eta + 2M) \eta - \frac{1}{T} \sum_t g_{it} \right ) \\
\leq \max_i \prob \left ( \left | \frac{1}{T}\sum_t [x_{it}'(\theta^0(c') - \theta^0(c))]^2 - g_{it} \right | \geq \frac{1}{6} \clusterdistmin \right ) = o(T^{-a})
\end{align*}

where the last line follows from applying lemma \ref{lemma:bonhomme} with $z_{it} = [x_{it}'(\theta^0(c') - \theta^0(c))]^2 - g_{it}$. 
A similar argument using assumptions \ref{assumptions:inference:mixing}, \ref{assumptions:inference:tails}, \ref{assumptions:inference:limits} on the process $\{e_{it} x_{it}\}_t$ shows that the second term in equation $\ref{equation:zbound}$ is also $o(T^{-a})$, and the final term is just as assumption \ref{assumptions:inference:xbound}. \\

Then for $\epsilon > 0$, the Markov inequality gives
\begin{align*}
\prob \left (T^a \sup_{\theta \in \neighborhood_{\eta}} \frac{1}{N} \sum_{i=1}^N \one(\wh{c}_i(\theta) \not = c_i^0) >  \epsilon \right ) &\leq T^a \frac{1}{\epsilon} \mathbb{E} \sup_{\theta \in \neighborhood_{\eta}} \frac{1}{N} \sum_{i = 1}^N \sum_{c \in \mathcal{C}} \max_{c' \not = c} Z_{ic}(c', \theta) \\
&\leq T^a \frac{1}{\epsilon} \frac{1}{N} \sum_{i = 1}^N \sum_{c \in \mathcal{C}} \sum_{c' \not = c} \mathbb{E} \sup_{\theta \in \neighborhood_{\eta}} Z_{ic}(c', \theta) \\
&\leq T^a \frac{1}{\epsilon} \frac{1}{N} \sum_{i = 1}^N |\mathcal{C}|^2 \max_{c \not = c'} \max_{i \in [N]} \mathbb{E} \sup_{\theta \in \neighborhood_{\eta}} Z_{ic}(c', \theta) = o(1) \\
\end{align*}

This completes the proof of the lemma. 
\end{proof}

In what follows, we let $\eta$ satisfy the conditions posited in \ref{lemma:cluster_errors}. 
Recall the sample risk with oracle cluster membership $\qoracle \equiv \wh{Q}(\theta, \gamma^0)$. 
We show that for every $a > 0$, $\sup_{\theta \in \mathcal{N}_{\eta}} (\wh{Q} - \qoracle)(\theta) = o_P(T^{-a})$. 
For any $\theta \in \neighborhood_{\eta}$, we can write
\begin{align}
| (\wh{Q} - \qoracle)(\theta) | &= \left | \frac{1}{NT} \sum_{i, t} [y_{it} - x_{it}'\theta(\wh{c}_i(\theta))]^2 - [y_{it} - x_{it}'\theta(c_i^0)]^2 \right | \leq \left | \frac{1}{NT} \sum_{i, t} 2e_{it} x_{it}'(\theta(c_i^0) - \theta(\wh{c}_i(\theta))) \right | \nonumber \\
&+ \left | \frac{1}{NT} \sum_{i, t} [x_{it}'(\theta^0(c_i^0) - \theta(c_i^0))]^2  - [x_{it}'(\theta^0(c_i^0) - \theta(\wh{c}_i(\theta)))]^2   \right | \label{equation:qqtilde}
\end{align}

The first term on the right hand side is bounded above by
\begin{align*}
\frac{1}{N} \sum_{i} \bigg | (\theta(c_i^0) - \theta(\wh{c}_i(\theta)))' \frac{1}{T} \sum_t & 2e_{it} x_{it} \one(\wh{c}_i(\theta) \not = c_i^0) \bigg | \leq \frac{2M}{N} \sum_{i}  \one(\wh{c}_i(\theta) \not = c_i^0) \left \|\frac{1}{T}  \sum_t e_{it} x_{it}  \right \| \\
&\leq \left ( \frac{1}{N} \sum_i \one(\wh{c}_i(\theta) \not = c_i^0) \right )^{1/2} \left (\frac{1}{N} \sum_i \left \| \frac{1}{T} \sum_t e_{it} x_{it} \right \|^2 \right )^{1/2} \\
&= o_P(T^{-(2a)/2}) \left ( \frac{1}{NT^2} \sum_{i} \sum_{t, s} e_{it} e_{is} x_{it}'x_{is} \right )^{1/2} \\
&= o_P(T^{-(2a)/2}) o_P(1) = o_P(T^{-a})
\end{align*}

where the last line follows by lemma \ref{lemma:cluster_errors} and assumption \ref{consistency:dependence}. 
The second term in equation \ref{equation:qqtilde} can be expanded as 
\begin{align*}
\left | \frac{1}{NT} \sum_{i, t} [x_{it}'(\theta^0(c_i^0) - \theta(\wh{c}_i(\theta))  + \theta(\wh{c}_i(\theta))  - \theta(c_i^0))]^2 - [x_{it}'(\theta^0(c_i^0) - \theta(\wh{c}_i(\theta)))]^2 \right | \\
\leq \left | \frac{1}{NT} \sum_{i, t} 2 x_{it}'(\theta^0(c_i^0) - \theta(\wh{c}_i(\theta))) x_{it}'(\theta(\wh{c}_i(\theta))  - \theta(c_i^0)) \right | + \left | \frac{1}{NT} \sum_{i, t} (x_{it}'(\theta(\wh{c}_i(\theta)) - \theta(c_i^0)))^2 \right | 
\end{align*}

For instance, the second term can be rewritten
\begin{align}
\bigg | \frac{1}{N} \sum_{i} \frac{1}{T} &\sum_t (x_{it}'(\theta(\wh{c}_i(\theta)) - \theta(c_i^0)))^2 \one(\wh{c}_i(\theta) \not = c_i^0) \bigg | \leq \left | \frac{M^2}{N} \sum_{i} \one(\wh{c}_i(\theta) \not = c_i^0) \frac{1}{T} \sum_t \|x_{it} \|^2  \right | \nonumber \\
&\leq M^2 \left ( \frac{1}{N} \sum_i \one(\wh{c}_i(\theta) \not = c_i^0) \right )^{\frac{1}{2}} \left (\frac{1}{N} \sum_i \left ( \frac{1}{T} \sum_t \|x_{it}\|^2 \right )^2 \right )^{\frac{1}{2}} \leq o_P(T^{-a}) \label{equation:inference:Op1}
\end{align}

where the last inequality uses lemma \ref{lemma:cluster_errors} and assumption \ref{assumptions:inference:xnorm}.
It follows that 

\begin{equation} \label{equation:qoracle_equivalence}
\sup_{\theta \in \neighborhood_{\eta}} |(\wh{Q} - \qtilde)(\theta)| = o_P(T^{-a}) 
\end{equation}

We claim that $\thetaoracle - \theta^0 = o_P(1)$. 
Note that since $\qoracle(\theta) = \wh{Q}(\theta, \gamma^0)$, it suffices to check that the assumptions in \ref{assumptions:consistency} hold for $\Gamma' \equiv \{\gamma^0\}$.
The only thing we need to check is assumption \ref{consistency:eigenvalue}, which is clear since $\{\gamma^0\} \subset \Gamma$ implies $\inf_{c', \gamma \in \{\gamma^0\}} \max_{c} \rho(c, c', \gamma) \geq  \inf_{c', \gamma \in \Gamma} \max_{c} \rho(c, c', \gamma) \geq \delta - o_P(1)$ as $N, T \to \infty$ by assumption \ref{consistency:eigenvalue}. 
This shows $\thetaoracle - \theta^0 = o_P(1)$. \\

Next, we will show that for any $a > 0$

\begin{equation} \label{equation:conv_on_qoracle}
\qoracle(\wh{\theta}) - \qoracle(\thetaoracle) = o_P(T^{-a}) 
\end{equation}

Let $a > 0$ and $\epsilon > 0$. 
Define the event $E_T \equiv \{ T^a (\qoracle(\wh{\theta}) - \qoracle(\thetaoracle)) > \epsilon\}$.  
\begin{align*}
\prob(E_T) \leq \prob(E_T \cap \{\wh{\theta}, \thetaoracle \in \neighborhood_{\eta}\}) + \prob(\wh{\theta} \not \in \neighborhood_{\eta} \; \text{or} \; \thetaoracle \not \in \neighborhood_{\eta}) = \prob(E_T \cap \{\wh{\theta}, \thetaoracle \in \neighborhood_{\eta}\}) + o(1)
\end{align*}

The final equality follows from a union bound and consistency of $\wh{\theta}$ and $ \thetaoracle$.
On the event $E_T \cap \{\wh{\theta}, \thetaoracle \in \neighborhood_{\eta}\}$, we have
\begin{align*}
0 \leq \qoracle(\wh{\theta}) - \qoracle(\thetaoracle) &= (\qoracle(\wh{\theta}) - \qconc(\wh{\theta})) + (\qconc(\wh{\theta}) - \qconc(\thetaoracle)) + (\qconc(\thetaoracle) - \qoracle(\thetaoracle)) \\
&\leq 2 \sup_{\theta \in \neighborhood_{\eta}} |(\qoracle - \qconc)(\theta)|
\end{align*}

where we used that $(\qconc(\wh{\theta}) - \qconc(\thetaoracle)) \leq 0$ by the definition of $\wh{\theta}$. 
Then using the inequality above, apparently 
\[
\prob(E_T) \leq \prob \left (  T^a \cdot 2 \sup_{\theta \in \neighborhood_{\eta}} |(\qoracle - \qconc)(\theta)| > \epsilon \right ) + o(1) = o(1)
\]

by equation \ref{equation:qoracle_equivalence}.  
This completes the proof of \ref{equation:conv_on_qoracle}. 
We now show a curvature lower bound for $\qoracle$.
For every $1 \leq \ell \leq G$ and each $x \in [k_{\ell}]$, $\thetaoracle \in \argmin_{\theta \in \Theta} \qoracle(\theta)$ implies 

\begin{equation} \label{equation:qoracle_foc}
0 = \nabla_{\theta_{\ell x}} \qoracle(\thetaoracle) = \frac{2}{NT} \sum_{i: c_{i \ell}^0 = x} \sum_{t=1}^T (y_{it} - x_{it}' \thetaoracle(c_i^0)) x_{it}^{\ell}
\end{equation}

Define $\erroracle_{it} \equiv (y_{it} - x_{it}'\thetaoracle(c_i^0)$ and compute
\begin{align*}
\qoracle(\wh{\theta}) - \qoracle(\thetaoracle) &= \frac{1}{NT} \sum_{i, t} (y_{it} - x_{it}' \wh{\theta}(c_i^0))^2 - \frac{1}{NT} \sum_{i, t}(y_{it} - x_{it}' \thetaoracle(c_i^0))^2 \\
&=  \frac{1}{NT} \sum_{i, t} (y_{it} - x_{it}'\thetaoracle(c_i^0) + x_{it}' [\thetaoracle(c_i^0) - \wh{\theta}(c_i^0)])^2 - \frac{1}{NT} \sum_{i, t}(y_{it} - x_{it}' \thetaoracle(c_i^0))^2 \\
&=  \frac{1}{NT} \sum_{i, t} (x_{it}' [\thetaoracle(c_i^0) - \wh{\theta}(c_i^0)])^2 + \frac{1}{NT} \sum_{i, t} \erroracle_{it} x_{it}'[\thetaoracle(c_i^0) - \wh{\theta}(c_i^0)]
\end{align*}

We claim that the second term is identically zero.
Define a map\footnote{For $S_1$ and $S_2$ subsets of the same vector space, we define $S_1 - S_2 \equiv \{s_1 - s_2: s_i \in S_i, i=1,2\}$.} $F: \Theta - \Theta \to \mathbb{R}$ by $F(\theta) = \sum_{i, t} \erroracle_{it} x_{it}'\theta(c_i^0)$.
Note that for any $1 \leq \ell \leq G$, we can write 
\begin{align*}
F(\theta) &= \sum_{i=1}^N \sum_{t=1}^T \erroracle_{it} x_{it}'\theta(c_i^0) = \sum_{t=1}^T \sum_{x \in [k_{\ell}]} \sum_{i: c^0_{i\ell} =x} \erroracle_{it} x_{it}'\theta(c_i^0) = \sum_{t=1}^T \sum_{x \in [k_{\ell}]} \sum_{i: c^0_{i\ell} =x} \sum_{\ellaux} \erroracle_{it} \lag x^{\ellaux}_{it}, \theta^{\ellaux}(c_i^0) \rag \\
&= \sum_{t=1}^T \sum_{x \in [k_{\ell}]} \sum_{i: c^0_{i\ell} =x} \sum_{\ellaux \not = \ell} \erroracle_{it} \lag x^{\ellaux}_{it}, \theta^{\ellaux}(c_i^0) \rag +  \sum_{x \in [k_{\ell}]} \sum_{i: c^0_{i\ell} =x} \sum_{t=1}^T \erroracle_{it} \lag x^{\ell}_{it}, \theta^{\ell}(c_i^0) \rag \\
&= \sum_{t=1}^T \sum_{x \in [k_{\ell}]} \sum_{i: c^0_{i\ell} =x} \sum_{\ellaux \not = \ell} \erroracle_{it} \lag x^{\ellaux}_{it}, \theta^{\ellaux}(c_i^0) \rag  
\end{align*}

where we have used that $\sum_{i: c^0_{i\ell} =x} \sum_{t=1}^T \erroracle_{it} (x^{\ell}_{it})'\theta^{\ell}(c_i^0) = 0$ for each $x$ by the first order condition \ref{equation:qoracle_foc}. 
Since the last expression doesn't involve $\theta^{\ell}$, we conclude that for any $\theta \in \text{Dom}(F)$, the equality $F(\theta^{\ell}, \theta^{-\ell}) = F(0, \theta^{-\ell})$ holds.
Applying this fact inductively, we find that $F = F(0) = 0$ identically. 
In particular, $F(\thetaoracle - \wh{\theta}) = 0$, which is what we needed to show. 
Then similar to the proof of \ref{thm:consistency}, we calculate 
\begin{align*}
\qoracle(\wh{\theta}) - \qoracle(\thetaoracle) &= \frac{1}{NT} \sum_{i, t} (x_{it}' [\thetaoracle(c_i^0) - \wh{\theta}(c_i^0)])^2 = \sum_{c \in \mathcal{C}} (\thetaoracle(c) - \wh{\theta}(c))' \left (\frac{1}{NT} \sum_{i, t} \one(c_i^0 = c) x_{it}x_{it}' \right )'(\thetaoracle(c) - \wh{\theta}(c)) \\
&\geq \sum_{c \in \mathcal{C}} \|\thetaoracle(c) - \wh{\theta}(c)\|^2 \lambda_{min}(M_{NT}^c) \geq \sum_{c \in \mathcal{C}} \|\thetaoracle(c) - \wh{\theta}(c)\|^2 \min_{c' \in \mathcal{C}} \lambda_{min}(M_{NT}^{c'}) 
\end{align*}

Define $W_{NT} \equiv \min_{c' \in \mathcal{C}} \lambda_{min}(M_{NT}^{c'})$, so that $W_{NT}^c \geq 0$ by positive semi-definiteness of $M_{NT}^c$ for all $N, T, c$. 
Also denote $E_{NT} = \{W_{NT} > \evalboundtwo / 2\}$.
Then by assumption \ref{assumptions:inference:evals}, $\prob(E_{NT}) = o(1)$. 
We have 
\begin{align*}
\sum_{c \in \mathcal{C}} \|\thetaoracle(c) - \wh{\theta}(c)\|^2 \min_{c' \in \mathcal{C}} \lambda_{min}(M_{NT}^{c'}) &= \sum_{c \in \mathcal{C}} \|\thetaoracle(c) - \wh{\theta}(c)\|^2 (\evalboundtwo / 2 + (W_{NT} - \evalboundtwo / 2)) \\
&\geq \sum_{c \in \mathcal{C}} \|\thetaoracle(c) - \wh{\theta}(c)\|^2 (\evalboundtwo/2 + (W_{NT} - \evalboundtwo / 2) \one(W_{NT} < \evalboundtwo /2)) \\
&=\sum_{c \in \mathcal{C}} \|\thetaoracle(c) - \wh{\theta}(c)\|^2 \evalboundtwo/2 + o_P(T^{-a})
\end{align*}

In the last line we used the compactness assumption \ref{consistency:compact}, the fact that $|W_{NT} - \evalboundtwo/2| \leq \evalboundtwo / 2$ on $E_{NT}^c$, and $T^a \one(E_{NT}^c) = o_P(1)$ for any $a > 0$ since $\prob(E_{NT}) \to 1$ by assumption \ref{assumptions:inference:evals}.  
Combining this with equation \ref{equation:conv_on_qoracle} shows that $\sup_{c \in \mathcal{C}} \|\thetaoracle(c) - \wh{\theta}(c) \| = o_P(T^{-a})$, which completes the proof of part \ref{equation:equivalence} of the theorem. \\

For the second part of the theorem \ref{equation:cluster_assignment} on cluster assignment, note that for $\eta$ satisfying the conditions in lemma \ref{lemma:cluster_errors}, using the bounds developed in the proof of the lemma we find that 

\begin{align*}
\prob(\exists i: \, \wh{c}_i \not = c_i^0) &\leq \prob ( \exists i: \, \exists \theta \in \neighborhood_{\eta}: \, \wh{c}_i(\theta) \not = c_i^0 \, \text{and} \, \wh{\theta} \in \neighborhood_{\eta}) + \prob(\wh{\theta} \not \in \neighborhood_{\eta})\\
&\leq \sum_i  \prob(\exists \theta \in \neighborhood_{\eta}: \wh{c}_i(\theta) \not = c_i^0) + o(1) = \sum_i \mathbb{E}[\sup_{\theta \in \neighborhood_{\eta}} \one(\wh{c}_i(\theta) \not = c_i^0)] + o(1)\\ 
&\leq \sum_i \mathbb{E} \sup_{\theta \in \neighborhood_{\eta}} \sum_{c \in \mathcal{C}} \sum_{c' \not = c} Z_{ic}(c', \theta) + o(1) \leq \sum_i \sum_{c \in \mathcal{C}} \sum_{c' \not = c} \mathbb{E} \sup_{\theta \in \neighborhood_{\eta}} Z_{ic}(c', \theta) + o(1) \\
&= o(NT^{-a}) + o(1)
\end{align*}

This completes the proof of the theorem.

\subsection{Proof of Theorem \ref{thm:cp}} \label{proof:cp} 

In this section, we prove consistency of model selection for the Cp criterion defined in the main text. 
The assumptions of theorem \ref{thm:cp} (stated in assumption \ref{assumptions:mo}) are imposed everywhere in this section. 
First we need some additional definitions. 
Let $\Theta^k = \prod_{\ell} \mathbb{R}^{d_{\ell} \times k_{\ell}}$ be the parameter space for a model with $k = (k_1, \dots, k_{\numgroups})$ clusters. 
Let $\mathcal{C}_k = \prod_i [k_i]$ and $\Gamma_{k} = \left [ \, [N] \to \mathcal{C}_k \, \right ]$ denote the set of possible cluster labels and cluster labelings of the cross-sectional units, where we may have $k \not = k^0$, the true number of clusters in each group. \\

Define $\baserisk = \samplerisk(\theta^0, \gamma^0) = \frac{1}{NT} \sum_{i, t} e_{it}^2$ to be the sample risk evaluated at the true model.
We begin with the following lemma on the sample risk of different models.  

\begin{lemma} \label{lemma:cp:risk_order}
The following hold
\begin{enumerate}[label=(\roman*), itemsep=1mm]
\item If $k = k^0$, then $\samplerisk(k) - \baserisk = O_p(\frac{1}{NT})$
\item If $k > k^0$, then $\samplerisk(k) - \baserisk = o_p(T^{-1 + 3 \epsilon})$
\item If $k$ is such that $k_i < k_i^0$ for some $i$, then $\samplerisk(k) - \baserisk = \Omega(1) + o_p(1)$ 
\end{enumerate}
\end{lemma}

\begin{proof}[Proof of (i) and (iii)]
Statement (i) follows from lemma \ref{lemma:model:qk} in the supplemental appendix. 
We note that if $k = k_0$, then $\wh{\theta}$ satisfies the conditions of lemma \ref{lemma:model:qk} by our inference result theorem \ref{thm:clt} and lemma \ref{lemma:cluster_errors} above on the convergence of average classification risk. \\

For the proof of part (iii), first define $\deltatheta^k_i = \theta^0(c_i^0) - \wh{\theta}^k(\wh{c}^k)$ and recall that 
\begin{equation} \label{equation:model:riskbound_small}
\samplerisk(k) - \baserisk = \frac{1}{NT}\sum_{i, t} (x_{it}'\deltatheta_i^k)^2 +  \frac{1}{NT} \sum_{i, t} e_{it} x_{it}' \deltatheta_i^k 
\end{equation}

The expression $\frac{1}{NT} \sum_{i, t} e_{it} x_{it}' \deltatheta_i$ was already shown to be $o_p(1)$ uniformly over $\deltatheta_i \in \Theta$ in equation \ref{equation:cross_term_bound} in the consistency proof.
Similarly, the first term was analyzed in equation \ref{equation:consistency:curvature}. 
The exact same argument as before shows that for arbitrary $(\theta^k, \gamma^k) \in \Theta^k \times \Gamma^k$
\begin{align*}
&\frac{1}{NT}\sum_{i, t} (x_{it}'(\theta^0(c_i^0) - \theta^k(c_i^k))^2 \\
&= \frac{1}{NT} \sum_{i, t} \sum_{c \in \mathcal{C}^k} \sum_{c' \in \mathcal{C}^{k_0}} (\theta^k(c) - \theta^0(c'))'x_{it}x_{it}'(\theta^k(c) - \theta^0(c')) \mathds{1}(c_i^0 = c')\mathds{1}(c_i^k = c) \\
&= \sum_{c \in \mathcal{C}^k} \sum_{c' \in \mathcal{C}^{k_0}} (\theta^k(c) - \theta^0(c'))' \left ( \frac{1}{NT} \sum_{i, t} x_{it} x_{it}' \mathds{1}(c_i^0 = c')\mathds{1}(c_i^k = c) \right )(\theta^k(c) - \theta^0(c'))  \\
&\geq \sum_{c \in \mathcal{C}^k} \sum_{c' \in \mathcal{C}^{k_0}} \|\theta^k(c) - \theta^0(c')\|^2 \rho(c, c', \gamma) \\
& \geq \sum_{c \in \mathcal{C}^k} \sum_{c' \in \mathcal{C}^{k_0}} \rho(c, c', \gamma) \max_{x \in \mathcal{C}^k} \|\theta^k(x) - \theta^0(c')\|^2 \\
& \geq \sum_{c' \in \mathcal{C}^{k_0}} \min_{\Tilde{c} \in \mathcal{C}^{k_0}} \min_{\gamma^k \in \Gamma^k} \max_{c \in \mathcal{C}^k} \rho(c, \Tilde{c}, \gamma) \min_{x \in \mathcal{C}^k} \|\theta^k(x) - \theta^0(c')\|^2 \\
& \geq (\delta - o_P(1)) \sum_{c' \in \mathcal{C}^{k_0}} \min_{x \in \mathcal{C}^k} \|\theta^k(x) - \theta^0(c')\|^2  
\end{align*}

We claim that $\max_{c' \in \mathcal{C}^{k_0}} \min_{x \in \mathcal{C}^k} \|\wh{\theta}^k(x) - \theta^0(c')\|^2 = \Omega(1)$. 
Let $1 \leq \ell \leq \numgroups$ be such that $k_{\ell} < k^0_{\ell}$ and define $\sigma(j) = \argmin_i \|\wh{\theta}^k_{\ell i} - \theta^0_{\ell j}\|$.
Since $k_{\ell} < k^0_{\ell}$, by the pigeonhole principle $\sigma(j) = \sigma(i)$ for some $i, j \in [k^0_{\ell}]$. 
Then by cluster separation (assumption \ref{assumptions:consistency}) 
\begin{equation*}
0 < \clusterdistmin \leq \|\theta^0_{\ell j} - \theta^0_{\ell i}\| \leq \|\theta^0_{\ell j} - \wh{\theta}^k_{\ell \sigma(j)} \| + \|\wh{\theta}^k_{\ell \sigma(j)} - \wh{\theta}^k_{\ell \sigma(i)} \| + \|\wh{\theta}^k_{\ell \sigma(i)} - \theta^0_{\ell i}\|
\end{equation*}

Since the middle term on the RHS is $0$, $\max (\|\theta^0_{\ell j} - \wh{\theta}^k_{\ell \sigma(j)} \|, \|\theta^0_{\ell i} - \wh{\theta}^k_{\ell \sigma(i)} \|) > \clusterdistmin / 2$.
Without loss suppose the max is achieved at $i$. 
Then for any $c'$ with $c'_{\ell} = i$, we have $\min_{x \in \mathcal{C}^k} \|\wh{\theta}^k(x) - \theta^0(c')\|^2 \geq (\clusterdistmin/2)^2$.
Plugging in $(\wh{\theta}^k, \wh{\gamma}^k)$ into our uniform bound above, we find
\begin{align*}
\frac{1}{NT}\sum_{i, t} (x_{it}'(\theta^0(c_i^0) - \wh{\theta}^k(\wh{c}^k_i))^2 \geq (\delta - o_P(1)) \sum_{c' \in \mathcal{C}^{k_0}} \min_{x \in \mathcal{C}^k} \|\wh{\theta}^k(x) - \theta^0(c')\|^2 \geq  \delta (\clusterdistmin / 2)^2 - o_p(1)  
\end{align*}

where we have used compactness of $\Theta$ in the final line.
Then we have shown that $\samplerisk(k) - \baserisk \geq \delta (\clusterdistmin / 2)^2 + o_p(1)$, which completes the proof of (ii). 
\end{proof}

For the proof of part (i), we need to develop some extra machinery. 
In this section, we denote $\model = (\theta, \gamma) \in \Theta \times \Gamma$, and let $\model^k$ and $\modeltrue$ be parameter, cluster label pairs in $\Theta^k \times \Gamma^k$ and $\Theta^{k_0} \times \Gamma^{k_0}$ respectively.  
We denote $\modelk = (\thetak, \gammak)$ and $\modeltrue = (\theta^0, \gamma^0)$. 
Define
\begin{equation*}
\avgdist(\model, \model') \equiv \frac{1}{N} \sum_i (\theta(c_i) - \theta'(c_i'))^2
\end{equation*}

The following key lemma forms the backbone of our inductive approach for establishing (near) $\sqrt{T}$-consistency for over-specified estimators.

\begin{lemma} \label{lemma:mo:recursion}
Let $k \geq k_0$ and $\rootrate \equiv T^{-\frac{1}{2} + \epsilon}$.
Then for any sequence $\thetarate = o(1)$, we have 
\begin{align}
\avgdist(\modelk, \modeltrue) = O_p(\thetarate) \implies \avgdist(\modelk, \modeltrue) = o_p(\thetarate^{1/2}\rootrate)
\end{align}
\end{lemma}

\begin{proof}
In what follows, let $(\thetak, \gammak) = \argmin_{\theta \in \Theta^k, \gamma \in \Gamma^k} Q(\theta, \gamma)$ and again let $\deltatheta^k_i = (\thetak(\wh{c}_i^k) - \theta^0(c_i^0))$. 
With $\qtilde$ defined as in our consistency proof in equation \ref{equation:qtilde_def}, we have 
\begin{align*}
|\qtilde(\thetak, \gammak) - \wh{Q}(\thetak, \gammak)| &= \left | -2 \frac{1}{N} \sum_i \left \lag \deltatheta^k_i,  \left ( \frac{1}{T} \sum_t e_{it} x_{it} \right ) \right \rag \right | \lesssim \frac{1}{N} \sum_i \|\deltatheta^k_i \| \left \|\frac{1}{T} \sum_t e_{it} x_{it} \right \| \\
&\leq \left ( \frac{1}{N} \sum_i  \|\deltatheta^k_i\|^2 \right )^{1/2} \left (\frac{1}{N} \sum_i \left \| \frac{1}{T} \sum_t e_{it} x_{it} \right \|^2 \right )^{1/2}\\
&\leq \left ( \frac{1}{N} \sum_i  \|\deltatheta^k_i\|^2 \right )^{1/2} \left (\sup_{i \in [N]} \left \| \frac{1}{T} \sum_t e_{it} x_{it} \right \| \right )\\
&= O_p(\thetarate^{1/2}) o_p(\rootrate) = o_p(\thetarate^{1/2} \rootrate)
\end{align*}

The second to last equality holds by our assumption and applying lemma \ref{lemma:mo:deviations}.  
Now we reason
\begin{align*}
\qtilde(\thetak, \gammak) &= \wh{Q}(\thetak, \gammak) + [\qtilde(\thetak, \gammak) - \wh{Q}(\thetak, \gammak)] \leq \wh{Q}(\thetak, \gammak) + |\qtilde(\thetak, \gammak) - \wh{Q}(\thetak, \gammak)| \\
&= \wh{Q}(\thetak, \gammak) +  o_p(\thetarate^{1/2} \rootrate)\\
&\leq \wh{Q}(\theta^0, \gamma^0) + o_p(\thetarate^{1/2} \rootrate) = \qtilde(\theta^0, \gamma^0) + o_p(\thetarate^{1/2} \rootrate) 
\end{align*}

The inequality holds because $k \geq k^0 \implies$ $(\theta^0, \gamma^0)$ is in the parameter space of the misspecified estimator.\footnote{Specifically, there exist $\theta^k \in \Theta^k$ and $\gamma^k \in \Gamma^k$ such that the $N \times p$ matrix with ith row $\theta^0(\gamma^0(i)) = \theta^k(\gamma^k(i))$ for all $i \in [N]$}
This shows that $0 \leq \qtilde(\thetak, \gammak) - \qtilde(\theta^0, \gamma^0) \leq o_P(\thetarate^{1/2} \rootrate)$. 
Then by above we have 
\begin{align*}
o_p(\thetarate^{1/2} \rootrate) &\geq \qtilde(\thetak, \gammak) - \qtilde(\theta^0, \gamma^0) = \frac{1}{NT} \sum_{i, t} (x_{it}'\deltatheta^k_i)^2 = \frac{1}{N} \sum_i (\deltatheta^k_i)' \left ( \frac{1}{T} \sum_t x_{it} x_{it}' \right )\deltatheta^k_i \\
&= \frac{1}{N} \sum_i (\deltatheta^k_i)' \left ( \frac{1}{T} \sum_t E[x_{it} x_{it}'] \right )\deltatheta^k_i + \frac{1}{N} \sum_i (\deltatheta^k_i)' \left ( \frac{1}{T} \sum_t (x_{it} x_{it}' - E[x_{it} x_{it}']) \right )\deltatheta^k_i \\
&\geq \frac{1}{N} \sum_i (\deltatheta^k_i)' \left ( \frac{1}{T} \sum_t E[x_{it} x_{it}'] \right )\deltatheta^k_i - \left | \frac{1}{N} \sum_i (\deltatheta^k_i)' \left ( \frac{1}{T} \sum_t (x_{it} x_{it}' - E[x_{it} x_{it}']) \right )\deltatheta^k_i \right |  
\end{align*}

Now again applying the triangle inequality, Cauchy-Schwarz, and the definition of an operator norm we have 
\begin{align*}
\left | \frac{1}{N} \sum_i (\deltatheta^k_i)' \left ( \frac{1}{T} \sum_t (x_{it} x_{it}' - E[x_{it} x_{it}']) \right )\deltatheta^k_i \right | &\leq \frac{1}{N} \sum_i \|\deltatheta^k_i\|^2 \left \|  \frac{1}{T} \sum_t (x_{it} x_{it}' - E[x_{it} x_{it}']) \right \| \\
&\leq \frac{1}{N} \sum_i \|\deltatheta^k_i\|^2 \sup_{j \in [N]} \left \|  \frac{1}{T} \sum_t (x_{jt} x_{jt}' - E[x_{jt} x_{jt}']) \right \|  \\
&=O_p(\thetarate) o_p(\rootrate) = o_p(\thetarate \cdot \rootrate)
\end{align*}

where the last equality uses lemma \ref{lemma:mo:deviations}. 
Then continuing the chain of inequalities above we have 
\begin{align*}
 o_p(\thetarate^{1/2} \rootrate) &\geq \qtilde(\thetak, \gammak) - \qtilde(\theta^0, \gamma^0) \geq \frac{1}{N} \sum_i  \|\deltatheta^k_i\|^2 \min_{j \in [N]} \lambda_{min} \left ( \frac{1}{T} \sum_t E[x_{jt} x_{jt}'] \right ) - o_p(\thetarate \cdot \rootrate) 
\end{align*}

By assumption $\thetarate = o(1)$, so collecting the $o_p$ terms on the LHS and defining $\lambda_{NT}$ to be the eigenvalue term on the RHS, we have
\begin{align*}
o_p(\thetarate^{1/2} \rootrate) &\geq \frac{1}{N} \sum_i  \|\deltatheta^k_i\|^2 \lambda_{NT} \geq \frac{1}{N} \sum_i  \|\deltatheta^k_i\|^2 (\eiglowerbound / 2 + (\lambda_{NT} - \eiglowerbound / 2) \one(\lambda_{NT} \leq \eiglowerbound / 2)) \\  
&\geq \eiglowerbound / 2 \frac{1}{N} \sum_i  \|\deltatheta^k_i\|^2 - \left | \frac{1}{N} \sum_i \|\deltatheta^k_i\|^2 (\lambda_{NT} - \eiglowerbound / 2) \one(\lambda_{NT} \leq \eiglowerbound / 2)) \right |  \\  
&\geq \eiglowerbound / 2 \frac{1}{N} \sum_i  \|\deltatheta^k_i\|^2 - (\eiglowerbound /2) M^2  \one(\lambda_{NT} \leq \eiglowerbound / 2)) \geq \eiglowerbound / 2 \frac{1}{N} \sum_i  \|\deltatheta^k_i\|^2 - o(\thetarate^{1/2} \rootrate) 
\end{align*}

The second to last inequality follows by assumption \ref{assumptions:mo:evals} and compactness. 
The final inequality holds because indicator functions that converge to $0$ do so at arbitrary rate. 
This completes the proof of the lemma. 
\end{proof}

\begin{cor} \label{cor:mo:rates}
For any $k \geq k_0$
\begin{align}
\avgdist(\modelk, \modeltrue) &= o_p(T^{-1 + 3 \epsilon}) \\
\samplerisk(k) - \samplerisk^0 &= o_p(T^{-1 + 3 \epsilon}) 
\end{align}
\end{cor}
\begin{proof} 
We claim that for all $r \geq 0$, we have $\avgdist(\modelk, \modeltrue) = O_p \left (T^{c_r} \right )$, where $c_r = -(1 - \frac{1}{2^r}) + \epsilon \sum_{j = 0}^r 2^{-j}$.
The proof is by induction. 
The base case $c_0 = \epsilon$ is immediate since $\avgdist(\modelk, \modeltrue) = O_p(1)$ by compactness of $\Theta$.
Assume the statement is true for all $0 \leq m \leq r$, then by lemma \ref{lemma:mo:recursion}
\begin{align*}
\avgdist(\modelk, \modeltrue) = O_p(T^{c_r}) \implies \avgdist(\modelk, \modeltrue) = o_p \left (T^{c_r / 2} \cdot T^{-1/2 + \epsilon}\right)
\end{align*}

and $c_r/2 - 1/2 + \epsilon = -(1/2 - \frac{1}{2^{r + 1}}) + \epsilon \sum_{j = 0}^r 2^{-j - 1} - 1/2 + \epsilon = -(1 - \frac{1}{2^{r + 1}}) + \epsilon \sum_{j = 0}^{r + 1} 2^{-j}$, which completes the induction.  
In particular, the first statement of the corollary holds as soon as $2^{-r} \leq \epsilon$.
For the second statement of the corollary, recall that
\begin{align*}
\samplerisk(k) - \samplerisk^0 = \frac{1}{NT}\sum_{i, t} (x_{it}'\deltatheta_i^k)^2 +  \frac{1}{NT} \sum_{i, t} e_{it} x_{it}' \deltatheta_i^k 
\end{align*}

The proof of lemma \ref{lemma:mo:recursion}, showed that $\avgdist(\modelk, \modeltrue) = o_p(\thetarate) \implies  \frac{1}{NT} \sum_{i, t} e_{it} x_{it}' \deltatheta_i^k = o_p(\thetarate^{1/2} \rootrate)$. 
Under the same conditions 
\begin{align*}
\frac{1}{NT}\sum_{i, t} (x_{it}'\deltatheta_i^k)^2 \leq \frac{1}{NT}\sum_{i, t} \|x_{it}\|^2 \|\deltatheta_i^k\|^2  &\leq \frac{1}{N}\sum_i \|\deltatheta_i^k\|^2  \sup_{j \in [N]} \frac{1}{T} \sum_t \|x_{jt}\|^2 \\
&\leq o_p(\thetarate) O_p(1) = o_p(\thetarate)
\end{align*}

That $\sup_{i \in [N]} \frac{1}{T} \sum_t \|x_{it}\|^2$ is $O_p(1)$ can easily be shown by a union bound in combination with assumption \ref{assumptions:inference:xbound} (as long as $NT^{-a} = o(1)$ for some $a > 0$).
Putting this together, we get that $\samplerisk(k) - \samplerisk^0 = o_p(T^{-1 + 3 \epsilon}) + o_p(T^{-1/2 + 3 \epsilon/2 - 1/2 + \epsilon}) = o_p(T^{-1 + 3 \epsilon})$. 
This completes the proof of the corollary and of the first part of lemma \ref{lemma:cp:risk_order}. 
\end{proof}

\begin{prop} \label{prop:mo:maxmin_rate}
For any $k \geq k^0$
\begin{align}
\forall c \in \clusterspace^{k_0} \quad \min_{x \in \mathcal{C}^k} \|\wh{\theta}^k(x) - \theta^0(c)\|^2 &= o_p(T^{-1 + 3 \epsilon}) 
\end{align}
\end{prop}

\begin{proof}
Applying corollary \ref{cor:mo:rates}, we find that
\begin{align*}
o_p(T^{-1 + 3 \epsilon}) &\geq \qtilde(\thetak, \gammak) - \qtilde(\theta^0, \gamma^0) = \frac{1}{NT}\sum_{i, t} (x_{it}'(\theta^0(c_i^0) - \wh{\theta}^k(\cest^k))^2 \\
& \geq \min_{\Tilde{c} \in \mathcal{C}^{k_0}} \min_{\gamma^k \in \Gamma^k} \max_{c \in \mathcal{C}^k} \rho(c, \Tilde{c}, \gamma) \max_{c' \in \mathcal{C}^{k_0}} \min_{x \in \mathcal{C}^k} \|\theta^k(x) - \theta^0(c')\|^2 \\
&= \max_{c' \in \mathcal{C}^{k_0}} \min_{x \in \mathcal{C}^k} \|\theta^k(x) - \theta^0(c')\|^2 (\delta / 2 + (\rho^k_{NT} - \delta/2) \one(\rho^k_{NT} - \delta/2 \leq 0)) \\ 
&= (\delta / 2) \cdot \max_{c' \in \mathcal{C}^{k_0}} \min_{x \in \mathcal{C}^k} \|\theta^k(x) - \theta^0(c')\|^2 - o_p(T^{-1 + 3 \epsilon}) \\ 
&\implies \max_{c \in \mathcal{C}^{k_0}} \min_{x \in \mathcal{C}^k} \|\wh{\theta}^k(x) - \theta^0(c)\|^2 = o_p(T^{-1 + 3 \epsilon}) 
\end{align*}

The final equality again follows by compactness of $\Theta$, positivity of $\rho^k_{NT}$, and because indicator functions that converge to $0$ (in probability) do so at arbitrary rates. 
Since the square norm above is additively separable in the norms of each block of the covariate vector, for any $c, c' \in \clusterspace^{k_0}$ with $c_{\ell} = c'_{\ell}$, we must have $(\argmin_{x \in \clusterspace^k} \|\wh{\theta}^k(x) - \theta^0(c)\|^2)_{\ell} = (\argmin_{x \in \clusterspace^k} \|\wh{\theta}^k(x) - \theta^0(c')\|^2)_{\ell}$. 
This shows that setting $\sigma_{\ell}(a) = (\argmin_{x \in \clusterspace^k} \|\wh{\theta}^k(x) - \theta^0(c)\|^2)_{\ell}$ for any $c \in \clusterspace^{k_0}$ with $c_{\ell} = a$ is well-defined.  
\end{proof}

The following proposition is our analogue of Theorem 3.2 in \cite{LiuOverspecifiedGroups}. 
We use a recursive argument to give a faster rate for the worst case cross-sectional unit error in our setting.

\begin{prop} \label{prop:worst_case_rate}
For any $k \geq k^0$
\begin{align}
\sup_{i \in [N]} \|\thetak(\wh{c}^k_i) - \theta^0(c_i^0)\| = o_p(T^{-\frac{1}{2} + 2 \epsilon})
\end{align}
\end{prop}

\begin{proof}
Define $\wh{Q}_i(\theta, c_i) = \frac{1}{T} \sum_t (y_{it} - x_{it}'\theta(c_i))^2$ and $\qtilde_i(\theta, c_i) = \frac{1}{T} \sum_t (x_{it}'(\theta^0(c_i^0) - \theta(c_i)))^2 + \frac{1}{T} \sum_t e_{it}^2$.
Recall the random cluster mapping $\sigma: \clusterspace^{k_0} \to \clusterspace^k$ defined above. 
Then since $\wh{c}_i$ is the optimal cluster choice given estimated parameters $\wh{\theta}$,  
\begin{align*}
\wh{Q}_i(\thetak, \wh{c}^k_i) &\leq \wh{Q}_i(\thetak, \sigma(c_i^0)) \implies \qtilde_i(\thetak, \wh{c}^k_i) \leq \qtilde_i(\thetak, \sigma(c_i^0)) + (\wh{Q}_i - \qtilde_i)(\thetak, \sigma(c_i^0)) + (\qtilde_i - \wh{Q}_i)(\thetak, \wh{c}^k_i) \\
&\leq \qtilde_i(\thetak, \sigma(c_i^0)) + |\wh{Q}_i - \qtilde_i|(\thetak, \sigma(c_i^0)) + |\qtilde_i - \wh{Q}_i|(\thetak, \wh{c}^k_i) 
\end{align*}

The second term above has
\begin{align*}
\sup_i |\wh{Q}_i - \qtilde_i|(\thetak, \sigma(c_i^0)) &= \sup_i \left | (\theta^0(c_i^0) - \thetak(\sigma(c_i^0)))' \frac{1}{T} \sum_t e_{it} x_{it} \right | \\
&\leq \max_{c \in \mathcal{C}^{k_0}} \min_{x \in \mathcal{C}^k} \|\wh{\theta}^k(x) - \theta^0(c)\| \sup_{i \in [N]} \left \|\frac{1}{T} \sum_t e_{it} x_{it} \right \| \leq o_p(T^{-1 + \frac{5}{2} \epsilon})
\end{align*}

where we apply proposition \ref{prop:mo:maxmin_rate} and lemma \ref{lemma:mo:deviations}.  
Similarly, the third term is 
\begin{align*} 
\sup_i |\qtilde_i - \wh{Q}_i|(\thetak, \wh{c}^k_i) &\leq \sup_{i \in [N]}\|\theta^0(c_i^0) - \thetak(\wh{c}_i^k)\| \sup_{j \in [N]} \left \|\frac{1}{T} \sum_t e_{jt} x_{jt} \right \| \\
&= O_p \left ( \sup_{i \in [N]}\|\theta^0(c_i^0) - \thetak(\wh{c}_i^k)\| \right ) o_p(T^{-\frac{1}{2} + \epsilon})
\end{align*}

Moreover, we reason 
\begin{align*}
\sup_{i \in [N]} |\qtilde_i(\thetak, \sigma(c_i^0)) &- \qtilde_i(\theta^0, c_i^0)| = \sup_{i \in [N]} \frac{1}{T} \sum_t (x_{it}'(\theta^0(c_i^0) - \thetak(\sigma(c_i^0))))^2\\
&\leq \sup_{i \in [N]} \frac{1}{T} \sum_t \|x_{it}\|^2 \|\deltatheta^k(c_i^0, \sigma(c_i^0))\|^2 \leq \sup_{i \in [N]} \frac{1}{T} \sum_t \|x_{it}\|^2 \sup_{j \in [N]} \|\deltatheta^k(c_j^0, \sigma(c_j^0))\|^2 \\
&\leq \sup_{i \in [N]} \frac{1}{T} \sum_t \|x_{it}\|^2 \max_{c \in \clusterspace^{k_0}} \min_{x \in \mathcal{C}^k} \|\wh{\theta}^k(x) - \theta^0(c)\|^2 = O_p(1) o_p(T^{-1 + 3 \epsilon}) 
\end{align*}

That $\sup_{i \in [N]} \frac{1}{T} \sum_t \|x_{it}\|^2$ is $O_p(1)$ can easily be shown by a union bound in combination with assumption \ref{assumptions:inference:xbound} (as long as $NT^{-a} = o(1)$ for some $a > 0$).
Putting this all together, we have
\begin{align*}
0 &\leq \sup_{i \in [N]} [\qtilde_i(\thetak, \wh{c}^k_i) - \qtilde_i(\theta^0, c_i^0)] \\
&\leq \sup_{i \in [N]} [\qtilde_i(\thetak, \sigma(c_i^0)) + \sup_{j \in [N]} [\qtilde_i(\thetak, \wh{c}^k_j) - \qtilde_j(\thetak, \sigma(c_j^0))] - \qtilde_i(\theta^0, c_i^0)]  \\
&= o_p(T^{-1 + 3 \epsilon}) + O_p \left ( \sup_{i \in [N]}\|\theta^0(c_i^0) - \thetak(\wh{c}_i^k)\| \right ) o_p(T^{-\frac{1}{2} + \epsilon})
\end{align*}

where the first $o_p(1)$ is from work above and the second by lemma \ref{lemma:mo:deviations}. 
Now 
\begin{align*}
\sup_{i \in [N]} |\qtilde_i(\thetak, \wh{c}^k_i) - \qtilde_i(\theta^0, c_i^0)| &= \sup_{i \in [N]} |\deltatheta^k_i(\wh{c}^k_i, c_i^0)' \frac{1}{T} \sum_t x_{it} x_{it}'\deltatheta^k_i(\wh{c}^k_i, c_i^0)| \\
&\geq \sup_{i \in [N]} |\deltatheta^k_i(\wh{c}^k_i, c_i^0)' \frac{1}{T} \sum_t E[x_{it} x_{it}']\deltatheta^k_i(\wh{c}^k_i, c_i^0)| \\
&- \sup_{i \in [N]} |\deltatheta^k_i(\wh{c}_i^k, c_i^0)' \frac{1}{T} \sum_t (x_{it} x_{it}' - E[x_{it}x_{it}'])\deltatheta^k_i(\wh{c}_i^k, c_i^0)| \\
&\geq \sup_{i \in [N]} \|\deltatheta^k_i(\wh{c}^k_i, c_i^0)\|^2 \inf_{j \in [N]} \lambda_{min} \left (\frac{1}{T} \sum_t E[x_{jt} x_{jt}'] \right ) - \unifboundxxsup 
\end{align*}

where similar arguments show that
\begin{align*}
\unifboundxxsup = O_p \left ( \sup_{i \in [N]} \|\deltatheta^k_i(\wh{c}^k_i, c_i^0)\|^2 \right ) o_p(T^{-\frac{1}{2} + \epsilon})
\end{align*}

The indicator function trick used in the proof of lemma \ref{lemma:mo:recursion} above then shows that 
\begin{align}
\sup_{i \in [N]} \|\deltatheta^k_i(\wh{c}^k_i, c_i^0)\|^2 &=  O_p \left ( \sup_{i \in [N]} \|\deltatheta^k_i(\wh{c}^k_i, c_i^0)\|^2 \right ) o_p(T^{-\frac{1}{2} + \epsilon}) + o_p(T^{-1 + 3 \epsilon}) \nonumber \\
&+ O_p \left ( \sup_{i \in [N]}\|\deltatheta^k_i(\wh{c}^k_i, c_i^0)\| \right ) o_p(T^{-\frac{1}{2} + \epsilon}) \label{equation:worst_case_recursion}
\end{align}

The remainder of the proof follows by induction. 
For the base case, using compactness in the expression above shows that $\sup_{i \in [N]} \|\deltatheta^k_i(\wh{c}^k_i, c_i^0)\|^2 = o_p(T^{-\frac{1}{2} + \epsilon})$.
The inductive step follows from the recursion in equation \ref{equation:worst_case_recursion}. 
This completes the proof. 
\end{proof}

\subsubsection{Bias}
We also need the following lemma on the order of our proposed bias correction
\begin{lemma} \label{lemma:cp:bias_order}
The following are true 
\begin{enumerate}[label=(\roman*), itemsep=1mm]
\item If $k \geq k^0$, then $\biascorr(k) = O_p(\frac{1}{NT})$
\item If $k$ is such that $k_i < k_i^0$ for some $i$, then $\biascorr(k) = o_p(1)$
\end{enumerate}
\end{lemma}

\begin{proof}
TBD, current $\cp$ criterion uses non bias-corrected sample risk. 
\end{proof}

We are now ready to complete the proof of model selection consistency using our $\cp$ criterion. 
For completeness, suppose that we choose $\wh{k}$ uniformly (independently) at random in the case of a tie. 
Denote $k > k'$ if $k_i \geq k_i'$ for all $i$ and strictly for some index.

\begin{proof}[Proof of Theorem \ref{thm:cp}]
We reason that 
\begin{align*}
\prob(\wh{k} \not = k^0) &\leq \prob(\exists k \not = k^0 \, s.t. \, \cp(k) \leq \cp(k^0)) \\
&\leq \sum_{\substack{k: \exists k_i < k_i^0 \\ k \leq k_{\text{max}}}} \prob(\cp(k) \leq \cp(k^0)) + \sum_{\substack{k > k^0 \\ k \leq k_{\text{max}}}} \prob(\cp(k) \leq \cp(k^0))
\end{align*}

For $k$ in the first summation (with $k_i < k_i^0$ for some $i$), we have
\begin{align*}
\prob (\cp(k) \leq \cp(k^0)) &= \prob \left (\samplerisk(k) - \samplerisk(k^0) + \biascorr(k) - \biascorr(k^0) \leq \sum_i (k_i^0 - k_i) f(N, T) \right)\\
&=\prob \left ([\samplerisk(k) - \baserisk] - [\samplerisk(k^0) - \baserisk] + \biascorr(k) - \biascorr(k^0) \leq \sum_i (k_i^0 - k_i) f(N, T) \right) \\
&=\prob \left (\Omega(1) + o_p(1) + O_p(1/NT) + o_p(1) \leq o(1) \right) = \prob(\Omega(1) \leq o_p(1)) \to 0
\end{align*}

Where we have applied lemmas \ref{lemma:cp:risk_order} and \ref{lemma:cp:bias_order} in the final line. 
Similarly, for $k$ in the second summation (with $k \geq k^0$ and $k_j > k_j^0$ for some $j$) 
\begin{align*}
\prob (\cp(k) \leq \cp(k^0)) &=\prob \left ([\samplerisk(k) - \baserisk] - [\samplerisk(k^0) - \baserisk] + \biascorr(k) - \biascorr(k^0) \leq \sum_i (k_i^0 - k_i) f(N, T) \right) \\
&\leq \prob \left (2 \cdot O_p(1/NT) + o_p(T^{-1 + 3 \epsilon}) \leq -f(N, T) \right)\\
&= \prob \left (O_p \left (\frac{1}{NT^{3 \epsilon}} \right ) + o_p(1) \leq -T^{1-3\epsilon} f(N, T) \right ) \to 0
\end{align*}

since $T^{1-3 \epsilon} f(N, T) \to \infty$ by assumption.
Because each of the sums above is finite, this shows that $\prob(\wh{k} \not = k^0) = o(1)$, which completes the proof.
\end{proof}

\subsection{Variance Estimator Consistency} \label{proof:variance_estimation} 

In this section, we sketch how to adapt Hansen (2007)'s proof of the consistency of the Arellano (1987) HAC variance estimator to the estimator proposed in equation \ref{equation:variance_estimator} under asymptotics where $N, T \to \infty$ jointly. 
To use Hansen's proof, we impose the following assumptions, as well as compactness \ref{consistency:compact} and mixing conditions \ref{assumptions:inference:mixing}.
Throughout, we assume that $1 \in x_{it}$.

\begin{assumption} \label{assumptions:variance} 
Impose the following assumptions
\begin{enumerate}[label={(\alph*)}, ref ={\ref*{assumptions:variance}.(\alph*)}, itemindent=.5pt, itemsep=.5pt]
    \item $(e_{it}, x_{it}, c^0_i$) are cross-sectionally independent \label{assumptions:variance:independence}
    \item There exists $c < \infty$ and $\delta > 0$ such that for all $i, t$, and components $x_{itq}$ of $x_{it}$, we have $\mathbb{E}|x_{itq}|^{4 + \delta} < c$ and $\mathbb{E}|e_{it}|^{4 + \delta} < c$ \label{assumptions:variance:moments}
\end{enumerate} 
\end{assumption}

First note that for any $a \in [k_{\ell}]$ and $b \in [k_s]$ and $\epsilon > 0$ and $\delta > 0$, we have 

\begin{align}
\prob( \|\wh{\Omega}_{\ell a, sb}(\wh{c}) - \Omega_{\ell a, sb} \| > \epsilon) &\leq \prob( \|\wh{\Omega}_{\ell a, sb}(c^0) - \Omega_{\ell a, sb} \| > \epsilon) + \prob \left (\exists i \in [N] \, s.t. \, \wh{c}_i \not = c_i^0 \right ) \nonumber \\
&=\prob( \|\wh{\Omega}_{\ell a, sb}(c^0) - \Omega_{\ell a, sb} \| > \epsilon) + o(1) + O(NT^{-\delta}) \label{equation:variance_decomposition} 
\end{align}

So for consistency it suffices to focus on the estimator $\wh{\Omega}(c^0)$ defined by \ref{equation:variance_estimator} evaluated at the true cluster membership matrix. \\

Let $Z_i = \one(c_{i\ell}^0 = a) \one(c_{is}^0 = b) \in \{0, 1\}$. 
Define $\deltatheta(c) \equiv \theta^0(c) - \wh{\theta}(c)$ and $\deltatheta_i \equiv \deltatheta(c_i^0)$. 
Using this notation, we have
\begin{align*}
\wh{\Omega}_{\ell a, sb} &= \frac{1}{NT} \sum_{i, t, t'} \wh{e}_{it} \wh{e}_{it'}x_{it \ell} x_{it's}' Z_i \\
&= \frac{1}{NT} \sum_{i, t, t'} (e_{it} e_{it'} + e_{it} x_{it'}'\deltatheta_i + x_{it}' \deltatheta_i e_{it'} + \deltatheta_i'x_{it'} x_{it}' \deltatheta_i) x_{it \ell} x_{it's}' Z_i 
\end{align*}

We focus on just one term in the $d_{\ell} \times d_s$ matrix $x_{it\ell}x_{it's}'$, which we denote $x_{itp}x_{it'q}$. 
Then, for instance, the second term in the preceding expansion can be written as 
\begin{align*}
\frac{1}{NT} \sum_i Z_i \left ( \sum_t e_{it} x_{itp} \right ) &\left (\sum_{t'} x_{it'} x_{it'q} \right )' \deltatheta_i \\
&= \frac{1}{NT} \sum_i Z_i \sum_{c \in \mathcal{C}} \one(c_i^0 = c) \left ( \sum_t e_{it} x_{itp} \right ) \left (\sum_{t'} x_{it'} x_{it'q} \right )' \deltatheta_i \\
&=\frac{1}{NT} \sum_{c \in \mathcal{C}} \sum_i Z_i \one(c_i^0 = c) \left ( \sum_t e_{it} x_{itp} \right ) \left (\sum_{t'} x_{it'} x_{it'q} \right )' \deltatheta(c) 
\end{align*}

Note that $\mathcal{C}$ is finite. 
Each term inside the sum $\sum_{c \in \clusterspace}$ above 

\[
\frac{1}{NT} \sum_i Z_i \one(c_i^0 = c) \left ( \sum_t e_{it} x_{itp} \right ) \left (\sum_{t'} x_{it'} x_{it'q} \right )' \deltatheta(c)
\]

has the form of equation (O.2) in the supplementary appendix of \cite{Hansen2007Robust}, up to the extra term $Z_i \one(c_i^0 = c)$.
However, since $E\|Z_i \one(c_i^0 = c) v\| \leq E\|v\|$ for any vector $v$, these extra terms preserve the moment bounds needed for application of the Markov LLN (Hansen, Lemma A.2). \\

To show the moment bound $E\|(\sum_t e_{it} x_{itp})(\sum_t' x_{it'} x_{it'q}) \|^{1 + \delta}$ needed for the Markov LLN, Hansen's Theorem 3 and Lemma A.4 assume decay rates on the mixing coefficients $\alpha(t)$ of $(x_{it}, e_{it})$.
Our exponential mixing assumption \ref{assumptions:inference:mixing} is already sufficient for the polynomial rate used in his proof. 
Thus, Hansen's results apply to show that $\frac{1}{NT} \sum_{i, t, t'} e_{it} x_{it'}'\deltatheta_i x_{it \ell} x_{it's}' = O_p(\frac{1}{\sqrt{n}})$. 
The arguments from Hansen's proof similarly show that under the conditions in assumption \ref{assumptions:variance} the third term in equation \ref{equation:variance_decomposition} is $O_p(\frac{1}{\sqrt{n}})$, the fourth term is $O_p(\frac{1}{n})$, and the first term converges to the limit postulated in \ref{assumptions:clt:omega}.

\section{Supplementary Lemmas} \label{section:supplement}
In the following lemma, we show that $\max_i \var(\ov{x}_i' \Delta \theta) = o(1)$, needed for the proof of \ref{thm:fe:inference}.  
The proof is an application of methods developed in \cite{Rio1993}. 
Also see \cite{Rio2017} for a more complete exposition of covariance inequalities for strongly-mixing processes. 

\begin{lemma} \label{lemma:fe:variance_xi}
Under the fixed effects assumptions \ref{assumptions:fe:inference}, $\max_i \var(\ov{x}_i' \Delta \theta) \to 0$ as $T \to \infty$. 
\end{lemma}

\begin{proof}
For a sequence of mixing coefficients $\{\alpha(t)\}_{t \geq 0}$ define $\alpha \inv(u) = \sum_{t \geq 0} \one(\alpha(t) > u)$ for $0 \leq u \leq 1$.
Also, for scalar random variable $X$ we let $Q_X(u) \equiv \inf\{t \geq 0: P(|X| > t) \leq u\}$ be the reversed quantile function of $|X|$.
First note that
\begin{align*}
\var( \ov{x}_i' \Delta \theta) &= \var \left ( \sum_{k=1}^p \Delta \theta_k \ov{x}_{ik} \right) \leq \sum_{k,j} \var(\ov{x}_{ik})^{1/2} \var(\ov{x}_{ij})^{1/2} |\Delta \theta_k| |\Delta \theta_j| \\
&\leq M^2 \left ( \sum_k \var(\ov{x}_{ik})^{1/2} \right)^2 \leq pM^2 \sum_{k=1}^p \var(\ov{x}_{ik})
\end{align*}

The first inequality is from Cauchy-Schwarz, the second from compactness, and the final from Jensen's inequality. 
Then apparently it suffices to prove that $\max_i \var(\ov{x}_{ik}) \to 0$ as $T \to \infty$ for each $k$. 
Thus, in what follows we assume that $x_{it}$ is a scalar random variable. 
Corollary 1.1 of \cite{Rio2017} gives the bound 
\begin{equation} \label{equation:rio_bound}
\var(\ov{x}_i) \leq \frac{4}{T^2} \sum_{t \geq 0} \int_0^1 \alpha \inv(u) Q^2_{x_{it}}(u) du
\end{equation}

We claim that for any random variable $x \in L^1$, the inequality $Q_{x}(u) \leq |Ex| + Q_{x-Ex}(u)$ holds. 
Note that for $t \geq 0$, 
\begin{align*}
\prob(|x - Ex| > t) \leq u \implies \prob(|x| > t + |Ex|) \leq \prob(||x| - |Ex|| > t) \leq \prob(|x - Ex| > t) \leq u 
\end{align*}

by the reverse triangle inequality. 
Then we have shown that 
\begin{align*}
&\{t + |Ex|: t \geq 0, \prob( |x-Ex|>t) \leq u \} \subset \{t \geq 0: \prob(|x| > t) \leq u\} \\ 
&\implies  \inf \{t \geq 0: \prob(|x| > t) \leq u\} \leq |Ex| + \inf \{t \geq 0: \prob( |x-Ex|>t) \leq u \}\\
&\iff Q_{x}(u) \leq |Ex| + Q_{x - Ex}(u) \leq E|x| + Q_{x - Ex}(u)
\end{align*}

In what follows, denote $z_{it} = x_{it} - E(x_{it})$.
Then using this inequality in \ref{equation:rio_bound}, we get the bound
\begin{align}
T^2\var(\ov{x}_i) &\leq \sum_{t=1}^T E|x_{it}|^2 \int_0^1 \alpha \inv (u) du + 2 \sum_{t=1}^T E|x_{it}| \int_0^1 \alpha \inv(u)Q_{z_{it}}(u) du \nonumber \\
&+ \sum_{t=0}^T \int_0^1 \alpha \inv(u) Q_{z_{it}}(u)^2 du \label{equation:fe:technical_lemma}
\end{align}

where we applied Jensen's inequality to reduce $(E|x|)^2 \leq E|x|^2$. 
For the first term, note that $\int_0^1 \alpha \inv(u) du = \sum_{s \geq 0} \int_0^1 \one(\alpha(s) \geq u) du = \sum_{s \geq 0} \alpha(s)$. 
For the second term, we need a bound on the function $Q_{z_{it}}(u)$ for each $i, t$.
Note that from assumption \ref{assumptions:fe:inference:tails}, we have $\prob(|x_{it} - Ex_{it}| > t) \leq e^{1-(t/f)^{d_2}}$, giving 
\begin{align*}
\sup_{i, t} Q_{z_{it}}(u) &= \sup_{i, t} \inf \{t \geq 0: \prob(|x_{it} - Ex_{it}| > t) \leq u\} \leq \inf \{t \geq 0: e^{1-(t/f)^{d_2}} \leq u \} \\
&= f(1-\log(u))^{1/d_2}
\end{align*}

where the last line is just the inverse of the tail bound.  
We claim that for all $a > 1$ and $u \in (0, 1]$, we have $1- \log(u) \leq a u^{-1/a}$.
Note that $1-log(1) = 1 \leq a(1)^{-1/a} = a$.
Moreover, for all $u \in (0, 1]$, $-\frac{\partial}{\partial u} (1 - \log u) = 1/u \leq u^{-1/a - 1} = -\frac{\partial}{\partial u} au^{-1/a}$. 
This proves the claim. 
Let $a > 2 / d_2 \vee 1 $, then our work shows $(1-\log(u))^{2/d_2} \leq \frac{1}{u^{1-\epsilon(d_2)}}$, for some $\epsilon(d_2) \in (0, 1)$ 
\begin{align*}
\int_0^1 \alpha \inv(u) Q_{z_{it}}(u)^2 du &= \sum_{s \geq 0} \int_0^{\alpha(s)} Q_{z_{it}}(u)^2 \leq \sum_{s \geq 0} \int_0^{\alpha(s)} f^2 a(d_2)^2 u^{-1 + \epsilon(d_2)} du \\
& \leq c(f, d_2) \sum_{s \geq 0} \alpha(s)^{\epsilon(d_2)}du  \leq c(f, d_2) \sum_{s \geq 0} e^{-b\epsilon(d_2) s^{d_1}} \\
&\equiv K(b, f, d_1, d_2) < \infty
\end{align*}

where the final sum can easily be shown to converge by standard methods. 
Moreover, since $d_2$ is arbitrary, writing $(1-\log(u))^{1/d_2} = (1-\log(u))^{2/\wt{d}_2}$, the same proof shows that $\int_0^1 \alpha \inv (u) Q_{z_{it}}(u) du < K'(b, f, d_1, d_2)$. 
Our work showed that $\sum_{s \geq 0} \alpha(s)^{\epsilon} = K(b, f, d_1, d_2) < \infty$ for some $\epsilon \in (0, 1)$.
Then apparently $\sum_{s \geq 0} \alpha(s) \leq K''(b, f, d_1, d_2) < \infty$ for some different constant $K''$, because the sequence spaces $\ell_p$ are nested and increasing in $p > 0$.  \\

The work above in equation \ref{equation:fe:moment_sum} shows that $ \sup_i \frac{1}{T} \sum_t E|x_{it}|^p = O(1)$ for $p=1, 2$. 
Then the decomposition in \ref{equation:fe:technical_lemma} above becomes
\begin{align*}
\var(\ov{x}_i) &\leq \frac{1}{T^2} O(T)O(1) + \frac{2}{T^2} \left (\sum_{t=1}^T E|x_{it}|^2 \right)^{1/2} \left ( \sum_{t=0}^T \left (\int_0^1 \alpha \inv(u)Q_{z_{it}}(u) du \right)^2 \right)^{1/2} \\
&+ \frac{1}{T^2} \sum_{t=0}^T K(b, f, d_1, d_2) \\
&= O(1/T) +  2 \left ( \frac{1}{T} \sum_{t=1}^T E|x_{it}|^2 \right)^{1/2} \left ( \frac{1}{T^3} \sum_{t=0}^T K''(b, f, d_1, d_2)^2 \right)^{1/2} \\
&= O(1/T) + O(1)O(1/T) 
\end{align*}

where all the order statements above hold uniformly in $i$.
This completes the proof. 
\end{proof}

\subsection{Model Selection Lemmas}
The following lemma gives conditions under which the sample risk deviation from the irreducible sample risk $\samplerisk(\wh{\theta}, \wh{\gamma}) - \baserisk$ converges at the rate needed for theorem \ref{thm:cp}. 
\begin{lemma} \label{lemma:model:qk}
Suppose that $(\wh{\theta}, \wh{\gamma}) \in \Theta^{k_0} \times \Gamma^{k_0}$ has rate $\wh{\theta} - \theta^0 = O_p(\rate)$ and satisfies $\cest = \cest(\wh{\theta})$ for all $i$ (as defined in lemma \ref{lemma:cluster_errors}).  
Also, suppose there exists a neighborhood $\mathcal{N}$ of $\theta^0$ such that 
\begin{equation} \label{equation:mo:cluster_errors}
\sup_{\theta \in \mathcal{N}} \frac{1}{N} \sum_{i=1}^N \one(\wh{c}_i(\theta) \not = c_i^0) = o_p(T^{-a}) 
\end{equation}

for any $a > 0$.
Then $\samplerisk(\wh{\theta}, \wh{\gamma}) = O_p(\frac{\rate}{\sqrt{NT}}) + O_p(\rate^2) + o_p(T^{-a})$. 
\end{lemma}

\begin{proof}
In what follows, denote $\deltatheta(c) = (\theta^0(c) - \wh{\theta}(c))$ and $\deltatheta(c, c') = (\theta^0(c) - \wh{\theta}(c'))$. 
By definition, we have 
\begin{equation} \label{equation:mo:sample_risk_decomp}
\samplerisk(\wh{\theta}, \wh{\gamma}) - \baserisk = \frac{1}{NT} \sum_{i, t} e_{it} x_{it}'\deltatheta(c_i^0, \cest) + \frac{1}{NT} \sum_{i, t} (x_{it}'\deltatheta(\ctrue, \cest))^2
\end{equation}

We consider each term separately. 
The first term is 
\begin{align*}
\frac{1}{NT} \sum_{i, t} e_{it} x_{it}'\deltatheta(c_i^0, \cest) &=  \frac{1}{NT} \sum_{i, t} e_{it} x_{it}'\deltatheta(\ctrue) \one(\cest = \ctrue) + \frac{1}{NT} \sum_{i, t} e_{it} x_{it}'\deltatheta(c_i^0, \cest) \one(\cest \not = \ctrue) \\
&=\frac{1}{NT} \sum_{i, t} e_{it} x_{it}'\deltatheta(\ctrue)\one(\cest = \ctrue) + \frac{1}{NT} \sum_{i, t} e_{it} x_{it}'\deltatheta(c_i^0, \cest) \one(\cest \not = \ctrue) \\
&- \frac{1}{NT} \sum_{i, t} e_{it} x_{it}'\deltatheta(\ctrue)(1-\one(\cest = \ctrue)) \\
&=\frac{1}{NT} \sum_{i, t} e_{it} x_{it}'\deltatheta(\ctrue) + \frac{1}{NT} \sum_{i, t} e_{it} x_{it}'(\deltatheta(c_i^0, \cest) - \deltatheta(\ctrue)) \one(\cest \not = \ctrue) 
\end{align*}

Consider the first term in the final line. 
This can be written
\begin{align*}
\frac{1}{NT} \sum_{i, t} e_{it} x_{it}'\deltatheta(\ctrue) &= \sum_{c \in \clusterspace^{k_0}} \frac{1}{NT} \sum_{i, t} e_{it} x_{it}'\deltatheta(c) \one(\ctrue = c)\\
&= \sum_{c \in \clusterspace^{k_0}} \deltatheta(c)' \left (\frac{1}{NT} \sum_{i, t} e_{it} x_{it} \one(\ctrue = c) \right )\\
&= \sum_{c \in \clusterspace^{k_0}} O_p(\rate)'O_p(1/\sqrt{NT}) = O_p(\frac{\rate}{\sqrt{NT}})
\end{align*}

where we used assumption \ref{assumptions:mo:ci_uncorrelated} in the final line. 
Using Cauchy-Schwarz and recalling $M = \text{Diam}(\Theta)$, the second term can be written
\begin{align*}
\frac{1}{NT} &\sum_{i, t} e_{it} x_{it}'(\deltatheta(c_i^0, \cest) - \deltatheta(\ctrue)) \one(\cest \not = \ctrue) \\
&\leq \frac{1}{NT} \left (\sum_i  \|\wh{\theta}(\ctrue) - \wh{\theta}(\cest)\|^2 \one(\cest \not = \ctrue) \right)^{1/2} \left( \sum_i \left \| \sum_t e_{it} x_{it}\right \|^2 \right)^{1/2} \\
&\leq M^2 \left (\frac{1}{N}\sum_i \one(\cest \not = \ctrue) \right)^{1/2} \left( \frac{1}{NT^2} \sum_{i, t, s} e_{it} e_{is} x_{it}'x_{is} \right)^{1/2} \\ 
&\leq o_p(1) \left (\frac{1}{N}\sum_i \one(\cest(\wh{\theta}) \not = \ctrue) \right)^{1/2}( \one(\wh{\theta} \in \neighborhood) + \one(\wh{\theta} \not \in \neighborhood)) \\
&\leq o_p(1) \left ( \sup_{\theta \in \neighborhood} \frac{1}{N}\sum_i \one(\cest(\theta) \not = \ctrue) \right)^{1/2} + o_p(1) \one(\wh{\theta} \not \in \neighborhood) \\
&\leq o_p(1) o_p(T^{-a}) + o_p(1) o_p(T^{-a}) = o_p(T^{-a})
\end{align*}

The third inequality uses assumption \ref{assumptions:consistency}, and the fourth uses equation \ref{equation:mo:cluster_errors}. 
For the final inequality, note that by consistency of $\wh{\theta}$, the indicator $\one(\wh{\theta} \not \in \neighborhood)$ converges in probability to $0$ at arbitrary rate. 
We deal with the second term in equation \ref{equation:mo:sample_risk_decomp} similarly. 
Note that 
\begin{align*}
\frac{1}{NT} \sum_{i, t} (x_{it}'\deltatheta(\ctrue, \cest))^2 &= \frac{1}{NT} \sum_{i, t} (x_{it}'\deltatheta(\ctrue))^2 \one(\cest = \ctrue) + \frac{1}{NT} \sum_{i, t} (x_{it}'\deltatheta(\ctrue, \cest))^2 \one(\cest \not = \ctrue) \\
&= \frac{1}{NT} \sum_{i, t} (x_{it}'\deltatheta(\ctrue))^2 + \frac{1}{NT} \sum_{i, t} ((x_{it}'\deltatheta(\ctrue, \cest))^2 - (x_{it}'\deltatheta(\ctrue))^2) ( \one(\cest \not = \ctrue))
\end{align*}

Again, we argue the first term above is 
\begin{align*}
\sum_{c \in \clusterspace^{k_0}}\frac{1}{NT} \sum_{i, t} (x_{it}'\deltatheta(c))^2 \leq \sum_{c \in \clusterspace^{k_0}} \|\deltatheta(c)\|^2 \frac{1}{NT} \sum_{i, t}  \|x_{it}\|^2 = O_p(\rate^2) O_p(1)
\end{align*}

where we have used the tail assumption \ref{assumptions:inference:tails} and $\deltatheta(c) = O_p(1/\sqrt{NT})$ for all $c$ in the final equality. 
Now, for instance, we can break the second term into parts and compute 
\begin{align*}
\frac{1}{NT} \sum_{i, t} (x_{it}'\deltatheta(\ctrue))^2  \one(\cest \not = \ctrue) &\leq \frac{1}{NT} \sum_{i, t} \|x_{it}\|^2 \|\deltatheta(\ctrue)\|^2  \one(\cest \not = \ctrue) \\
&\leq M^2 \left ( \frac{1}{N} \sum_i \one(\cest \not = \ctrue) \right)^{1/2} \left (\frac{1}{NT^2}\sum_i \sum_{t, s} \|x_{it}\|^2 \|x_{is}\|^2 \right )^{1/2} \\
&\leq O_p(1)\left ( \frac{1}{N} \sum_i \one(\cest(\wh{\theta}) \not = \ctrue) \right)^{1/2} \leq O_p(1) o_p(T^{-a}) 
\end{align*}

In the final line we have used assumption \ref{assumptions:inference:xnorm} as well as our analysis of the sum of indicator functions above. 
An identical proof shows that $\frac{1}{NT} \sum_{i, t} ((x_{it}'\deltatheta(\ctrue, \cest))^2 = o_p(T^{-a})$. 
Putting this all together gives the claimed result.
\end{proof}

The following lemma is needed to establish rates of convergence for strongly mixing sequences. 

\begin{lemma} \label{lemma:mo:deviations}
Impose the mixing and tail assumptions in \ref{assumptions:inference:mixing} and \ref{assumptions:inference:tails}, and suppose also that $\log N = o(T^{\epsilon})$ for some $\epsilon$ with $d/2 > \epsilon > 0$. 
Then 
\begin{align*}
\sup_{i \in [N]} \left \| \frac{1}{T} \sum_{t=1} (x_{it} x_{it}' - E[x_{it} x_{it}']) \right \|_{2,2} &= o_p(T^{-\frac{1}{2} + \epsilon}) \\
\sup_{i \in [N]} \left \| \frac{1}{T} \sum_{t=1} e_{it} x_{it} \right \| &= o_p(T^{-\frac{1}{2} + \epsilon})
\end{align*}
where the first line uses the standard operator norm.
\end{lemma}

\begin{proof}
For $(z_t)_{t \geq 0}$ a mean zero-process satisfying the assumptions in \ref{assumptions:inference:tails} and \ref{assumptions:inference:mixing}, let $s(z)^2 = \sup_t E z_t^2 + 2 \sum_{s > t} |E z_t z_s | < \infty$.
Let $d \equiv \frac{d_1 d_2}{d_1 + d_2}$. 
Then setting $\lambda = \frac{Tz}{4}$ in equation (1.7) in \cite{Rio2011}, for any $r \geq 1$ we have 
\begin{align*}
\prob \left (\frac{1}{T} \left | \sum_t z_t \right | > z \right) \leq 4 \left(1 + \frac{T z^2}{16rs(z)^2} \right)^{-r/2} + \frac{16C}{z} \exp \left(-b \frac{(Tz)^d}{(4fr)^d} \right )
\end{align*}

Where $C$ is a constant only depending on the mixing and tail parameters $b, f, d_1, d_2$. 
In particular, plugging in $z = xT^{-\frac{1}{2} + \epsilon}$ and $r = T^{\epsilon}$ gives 
\begin{align*}
\prob \left (\frac{1}{T^{\frac{1}{2} + \epsilon}} \left | \sum_t z_t \right | > x \right) \leq 4 \left(1 + \frac{T^{\epsilon}x^2}{16s(z)^2} \right)^{-T^{\epsilon}/2} + \frac{16CT^{\frac{1}{2} - \epsilon}}{x} \exp \left(-b \frac{T^{d/2} x^d}{(4f)^d} \right ) \equiv f(T, x, s(z))
\end{align*}

Let $v \equiv \sup_{i, q} s((e_{it}x_{itq})_t)$. 
Then, for instance, applying this to the second expression above we get for any $x > 0$
\begin{align*}
\prob \left ( T^{\frac{1}{2} - \epsilon} \sup_{i \in [N]} \left \| \frac{1}{T} \sum_{t=1} e_{it} x_{it} \right \| > x \right ) &\leq  \prob \left ( T^{\frac{1}{2} - \epsilon} \sup_{i \in [N]} \left \| \frac{1}{T} \sum_{t=1} e_{it} x_{it} \right \|_1 > x \right ) \\
&\leq Np \sup_{q \in [p]} \sup_{i \in [N]} \prob \left (\frac{1}{T^{\frac{1}{2} + \epsilon}} \left | \sum_t e_{it} x_{itq} \right | > x/p \right ) \\
&\lesssim N f(T, x/p, v) \to 0
\end{align*}

as $N, T \to \infty$ if $\log N = o(T^{\epsilon \wedge (d/2)})$ and $\sup_{i, q} s((e_{it}x_{itq})_t)) < \infty$. 
Note that covariance inequalities from \cite{Rio2017} can be used to show that 
\[
\sup_{i} \sup_{1 \leq a \leq p} s((e_{it}x_{ita})_t)) < \infty \quad \text{and} \quad \sup_{i} \sup_{1 \leq a, b \leq p} s((x_{ita}x_{itb} - E[x_{ita} x_{itb}])_t) < \infty
\] 
under the assumptions \ref{assumptions:inference:tails} and \ref{assumptions:inference:mixing} in our setting, as noted in BM. 
This completes the proof for the second term.
For the first term, note that by equivalence of finite-dimensional vector space norms, there is a constant $c(p)$ depending only on the dimension such that
\begin{align*}
\left \| \frac{1}{T} \sum_{t=1} (x_{it} x_{it}' - E[x_{it} x_{it}']) \right \|_{2,2} \leq  c(p) \left \| \frac{1}{T} \sum_{t=1} (x_{it} x_{it}' - E[x_{it} x_{it}']) \right \|_1 
\end{align*}
where $\|A\|_1 \equiv \sum_{i, j} |a_{ij}|$ for a matrix $A \in \mathbb{R}^{p \times p}$. 
The first statement of the lemma then follows by exactly the same argument, substituting $x_{it}x_{it}' - E[x_{it} x_{it}']$ for $e_{it} x_{it}$. 
\end{proof}

\begin{cor}
The following rates hold
\begin{align*}
\frac{1}{N} \sum_{i=1} \left \| \frac{1}{T} \sum_{t=1} (x_{it} x_{it}' - E[x_{it} x_{it}']) \right \| &= o_p(T^{-\frac{1}{2} + \epsilon}) \\
\frac{1}{N} \sum_{i=1} \left \| \frac{1}{T} \sum_{t=1} e_{it} x_{it} \right \| &= o_p(T^{-\frac{1}{2} + \epsilon})
\end{align*}
\end{cor}

\begin{proof}
Immediate by lemma \ref{lemma:mo:deviations}, noting that for any positive real numbers $(a_i)_{i=1}^N$ we have $\frac{1}{N} \sum_i a_i \leq \sup_{i \in [N]} a_i$.
\end{proof}

\section{Extensions and Supplementary Material}
\subsection{Fixed Effects Model} \label{proof:fe} 
In this section, we consider an extension of the main specification with individual fixed effects.

\begin{equation} \label{equation:fixed_effects}
y_{it} = x_{it}' \theta(c_i^0) + a_i + e_{it}
\end{equation}

We propose to estimate equation \ref{equation:fixed_effects} by (1) de-meaning the time series for each cross-sectional unit followed by (2) applying Lloyd's Algorithm to the de-meaned data. 
In other words, defining $\Tilde{z}_{it} \equiv z_{it} - \frac{1}{T} \sum_t z_{it} = z_{it} - \ov{z}_i$ for any variable $z_{it}$, we apply Lloyd's algorithm to the model $\wt{y}_{it} = \wt{x}_{it}'\theta(c_i^0) + \wt{e}_{it}$. \\

The main challenge in extending our results to this setting is that the differencing operation changes the autocorrelation structure of the data, so that the mixing conditions in assumption \ref{assumptions:inference:mixing} may no longer be satisfied. 
In the remainder of this section, we overload notation and let $\wh{\theta}$ and $\wh{\gamma}$ refer to the fixed effects estimates defined by 

\begin{equation} \label{fe:estimator}
 (\wh{\theta}, \wh{\gamma}) = \argmin_{\gamma \in \Gamma, \theta \in \Theta} \frac{1}{NT}\sum_{i=1}^N\sum_{t=1}^T (\wt{y}_{it} - \wt{x}_{it}'\theta(c_i))^2
\end{equation}

Letting the cluster permutation $\sigma_{\ell}$ be defined analagously to equation \ref{equation:permutation_function} in the main text, we have  

\begin{lemma} \label{lemma:fe:labeling}
Under the assumptions in \ref{assumptions:consistency} with $(x_{it}, e_{it})$ replaced by $(\wt{x}_{it}, \wt{e}_{it})$, $\mathbb{P}(\sigma_{\ell} \, \text{invertible}) \to 1$ as $N, T \to \infty$
\end{lemma}

Relabeling $\wh{\theta}_{\ell \sigma(a)} \to \wh{\theta}_{\ell a}$ (which is well-defined w.h.p. as $N, T \to \infty$ by the lemma), we have

\begin{thm} \label{thm:fe:consistency}
Under the assumptions in \ref{assumptions:consistency} with $(x_{it}, e_{it})$ replaced by $(\wt{x}_{it}, \wt{e}_{it})$, for all blocks $\ell$ and $a \in [k_{\ell}]$, we have $\|\theta^0_{\ell a} - \wh{\theta}_{\ell a} \| = o_P(1)$ as $N, T \to \infty$. 
\end{thm}

\begin{proof}

Immediate from lemma \ref{lemma:permutation} and theorem \ref{thm:consistency} applied to the the de-meaned data $(\wt{y}, \wt{x}, \wt{e})_{it}$.  
\end{proof}

For the analogue of theorem \ref{thm:inference}, we modify assumption \ref{assumptions:inference} to the following 

\begin{assumption} \label{assumptions:fe:inference}
Consider the following assumptions
\begin{enumerate}[label={(\alph*)}, ref ={\ref*{assumptions:fe:inference}.(\alph*)}, itemindent=.5pt, itemsep=.5pt]
    \item $\max_i \frac{1}{T^2} \sum_{t, s} E(\|x_{it}\|^2 \|x_{is}\|^2) = O(1)$ as $T \to \infty$ \label{assumptions:fe:inference:xnorm}
    \item Let \ref{assumptions:inference:evals} hold with $(e, x)_{it}$ replaced by $(\wt{e}, \wt{x})_{it}$ 
    Also, let assumptions \ref{assumptions:inference:mixing} and \ref{assumptions:inference:xbound} on mixing conditions of $x_{it} e_{it}$, and large deviations of $\sum_t \|x_{it}\|^2$ hold exactly as in assumption \ref{assumptions:inference} from the main theorem \label{assumptions:fe:inference:old} 
    \item The uniform limits $\max_{i \in [N]} \frac{1}{T} \sum_t E[e_{it} x_{it}] \to 0$ and $\min_{i \in [N]} \frac{1}{T} \sum_t \mathbb{E}(\wt{x}_{it}'(\theta(c) - \theta(c')))^2 \to \wt{\clusterdist}(c, c')$ hold as $T \to \infty$, and $\wt{d}(c, c') \geq \wt{d}_{min} > 0$ for $c \not = c'$. \label{assumptions:fe:inference:limits}
    \item There exist constants $f$ and $d_2$ such that for all $i \in [N]$ and all $z > 0$, for all components $x_{it}^j$, $x_{it}^{j'}$ of the vector $x_{it}$ we have $\prob(|x_{it}^j x_{it}^{j'}  - E(x_{it}^j x_{it}^{j'}) | > z)$, $\prob( |e_{it} x_{it}^j - E e_{it} x_{it}^j |  > z)$ and $\prob( |x_{it}^j - E x_{it}^j |  > z)$ are bounded above by $e^{1 - (z/f)^{d_2}}$ \label{assumptions:fe:inference:tails} 
    \item The covariate vector $x_{it}$ contains a constant \label{assumptions:fe:inference:one} 
\end{enumerate} 
\end{assumption}

Then, analogously letting $\wt{\theta}$ be the infeasible estimator that minimizes \ref{fe:estimator} with the true cluster identities $c_i^0$ plugged in, we have  

\begin{thm} \label{thm:fe:inference}
Let the assumptions needed for consistency (assumption \ref{assumptions:consistency}) hold with $(x, e)_{it}$ replaced by $(\wt{x}, \wt{e})_{it}$, and let the assumptions in \ref{assumptions:fe:inference} hold.
Then for any $a > 0$ and as $N, T \to \infty$, we have the following theorem

\begin{equation} 
\wh{\theta} = \thetaoracle + o_P(T^{-a})
\end{equation}

Moreover, individual cluster estimates satisfy

\begin{equation} 
\prob \left (\exists i \in [N] \, s.t. \, \wh{c}_i \not = c_i^0 \right ) = o(1) + o(NT^{-a})
\end{equation}
\end{thm}

The analogue of theorem \ref{thm:clt} follows immediately from theorem \ref{thm:fe:inference} by replacing $y_{it}$, $x_{it}$ and $e_{it}$ with the appropriate de-meaned variables. 
Specifically, define
\begin{align*}
\covxmatrix_{\ell a, s b} &= \frac{1}{NT} \sum_{i=1}^N \sum_{t=1}^T \wt{x}_{it \ell} \wt{x}_{its}' \one(c_{is} = b) \one(c_{i\ell} = a) \\
\cltvec_{\ell a} &= \frac{1}{NT} \sum_{i=1}^N \sum_{t=1}^T \wt{e}_{it} \one(c_{i \ell} = a) \wt{x}_{it\ell}
\end{align*}

and let

\[
\frac{1}{NT} \sum_{i, j = 1}^N \sum_{t, t' = 1}^T E[(e_{it} - \ov{e}_i)(e_{jt'} - \ov{e}_j) \one(c_{i \ell}^0 = a) \one(c_{js}^0 = b) (x_{it\ell} - \ov{x}_{i\ell}) (x_{jt's} - \ov{x}_{js})'] \to \Omega_{\ell a, s b} \label{equation:fe:covariance}
\]

as $N, T$ goes to infinity (we assume the limit exists), then the analogue of \ref{thm:clt} is 

\begin{thm} \label{thm:fe:clt}
Suppose that the assumptions in \ref{assumptions:clt} are satisfied with $(x, e)_{it}$ replaced by $(\wt{x}, \wt{e})_{it}$ and the matrices $\covxmatrix$ and $\Omega$ as defined above. 
Also let there $r > 0$ such that $\sqrt{N}T^{-r} = o(1) $. 
Then  

\begin{equation}
\sqrt{NT}(\vect(\wh{\theta} - \theta^0)) \convd \mathcal{N}(0, M \inv \Omega M)
\end{equation}
\end{thm}

We propose to use the HAC estimator defined in equation \ref{equation:variance_estimator}, with $\wh{e}_{it}$ replaced by the residuals from the fixed effects problem \ref{fe:estimator}. 

\begin{proof}[Proof of Theorem \ref*{thm:fe:inference}]

Following the same arguments as in the proof of theorem \ref{thm:inference}, equation \ref{equation:zbound} becomes
\begin{align}
&\max_{i} \prob \left ( \frac{1}{T}\sum_{t=1}^T [\wt{x}_{it}'(\theta^0(c') - \theta^0(c))]^2 \leq 4 \eta MK + (4 \eta + 2M) \eta \right ) \nonumber \\
&+ \max_{i}\prob \left ( \frac{1}{T}\sum_t \|\wt{x}_{it}\|^2 > M' \right ) +  \max_{i} \prob \left (  \left \| \frac{1}{T} \sum_{t = 1}^T \wt{e}_{it} \wt{x}_{it} \right \| > \eta \right ) \label{equation:fe:zbound} 
\end{align}

where we have replaced $M'$ with an arbitrary positive constant $K$, to be determined below.
For the second term, note that by assumption \ref{assumptions:inference:xbound}, we have $\max_i \prob(\sum_t \|x_{it}\|^2 > M') = o(T^{-a})$ for each $a > 0$. 
Note that for any $C > 0$, since $z \leq z^2 + 1$ on $\mathbb{R}_{\geq 0}$ 
\begin{align*}
\prob \left (\frac{1}{T} \sum_t \|x_{it} \| > C \right) \leq \prob \left ( \frac{1}{T} \sum_t (\|x_{it} \|^2 + 1) > C \right) \leq \prob \left ( \frac{1}{T} \sum_t \|x_{it} \|^2  > C - 1 \right) 
\end{align*}

Also note that $\|x_{it} - \ov{x}_i\|^2 \leq 2(\|x_{it}\|^2 + \|\ov{x}_i\|^2)$, and $\prob(\|\ov{x}_i\|^2 > C_1) \leq \prob(\frac{1}{T} \sum_t \|x_{it}\| > C_1^{1/2})$ by Cauchy-Schwarz. 
Putting this all together, we find that
\begin{align*}
\prob \left( \frac{1}{T} \sum_t \|x_{it} - \ov{x}_i\|^2 > K \right) &\leq \prob \left (\frac{1}{T} \sum_t \|x_{it} \|^2 > K/4 \right ) + \prob \left (\|\ov{x}_i\|^2 > K/4 \right) \\
&\leq \prob \left (\frac{1}{T} \sum_t \|x_{it} \|^2 > K/4 \right ) + \prob \left(\frac{1}{T} \sum_t \|x_{it}\|^2 > \sqrt{K}/2 - 1\right)\\
&= o(T^{-a})
\end{align*}

The first inequality follows from $\|a + b\|^2 \leq 2(\|a\|^2 + \|b\|^2)$ and a union bound. 
Moreover, the $o(T^{-a})$ statement holds uniformly over $i$ as long as $K \geq 4(M' + 1)^2$, where $M'$ is as in the uniform large deviations bound on $\|x_{it}\|^2$ in assumption \ref{assumptions:inference:xbound}. \\

For the third term, note that $\frac{1}{T} \sum_t \wt{e}_{it} \wt{x}_{it} = \frac{1}{T} \sum_t (e_{it} - \ov{e}_i)(x_{it} - \ov{x}_i) = \frac{1}{T}\sum_t e_{it} x_{it} - \ov{e}_i \ov{x}_i$.
The term $\frac{1}{T} \sum_t e_{it} x_{it}$ is uniformly $o_p(T^{-a})$ for any $a > 0$ by assumption \ref{assumptions:fe:inference:old}, as shown in the proof of \ref{thm:inference}.
For $c > 0$, the second term $\ov{e}_i \ov{x}_i$ has 
\begin{align*}
\prob(\|\ov{e}_i \ov{x}_i\|> c) &\leq \prob(\|\ov{x}_i\| > M' + 1) + \prob \left ( |\ov{e}_i| > \frac{c}{M' + 1} \right) \\
&\leq \prob(\frac{1}{T}\sum_t \|x_{it}\|^2 > M') + \prob \left ( |\ov{e}_i| > \frac{c}{M' + 1} \right) 
\end{align*}

The first term is uniformly $o(T^{-a})$ (in the sense of equation \ref{equation:zbound}) by assumption, and the second term is uniformly $o(T^{-a})$ by the same type of argument in the main proof using assumptions \ref{assumptions:fe:inference:one}, \ref{assumptions:inference:tails}, and \ref{assumptions:inference:limits} to invoke lemma \ref{lemma:bonhomme} on tail bounds for strongly mixing processes. \\

Let $K = 4(M' + 1)^2$ and $\eta$ such that $4 \eta MK + (4 \eta + 2M)\eta < \wt{d}_{min} / 2$. 
With $g_{it} = E((x_{it}'(\theta^0(c)-\theta^0(c')))^2)$ and $T'$ such that $\frac{1}{T'} \sum_{t=1}^{T'} g_{it} \geq \frac{1}{3} \wt{d}_{min}$. 
Then for $T > T'$, the first term is 
\begin{align*}
\prob \left ( \frac{1}{T}\sum_t \left ([\wt{x}_{it}'(\theta^0(c') - \theta^0(c))]^2 - g_{it} \right ) \leq (1/3)\wt{d}_{min} - \frac{1}{T} \sum_t g_{it} \right ) \\
\leq \prob \left ( \left | \frac{1}{T}\sum_t [\wt{x}_{it}'(\theta^0(c') - \theta^0(c))]^2 - g_{it} \right | \geq \frac{1}{6} \wt{d}_{min} \right )
\end{align*}

Setting $\Delta \theta \equiv \theta^0(c') - \theta^0(c)$, we can expand each term in the sum on the left hand side as
\begin{align*}
 [\wt{x}_{it}'(\theta^0(c') - \theta^0(c))]^2 - g_{it} &= ((\wt{x}_{it}' \Delta \theta)^2 - E(\wt{x}_{it}' \Delta \theta)^2) = (x_{it}' \Delta \theta)^2 - E(x_{it}' \Delta \theta)^2 \\
 &- 2(\ov{x}_i' \Delta \theta \Delta \theta' x_{it} - E  \ov{x}_i' \Delta \theta \Delta \theta' x_{it}) + ((\ov{x}_i' \Delta \theta)^2 - E(\ov{x}_i' \Delta \theta)^2) \\
& \equiv B^1_{iT} + B^2_{iT} + ((\ov{x}_i' \Delta \theta)^2 - E(\ov{x}_i' \Delta \theta)^2) 
\end{align*}

Fix $C > 0$. 
The first term has $\prob(\frac{1}{T}\sum_t (x_{it}' \Delta \theta)^2 - E(x_{it}' \Delta \theta)^2 > C) = o(T^{-a})$ uniformly over $i$ by applying lemma \ref{lemma:bonhomme} exactly as in the proof of the main theorem.  
For the second term, note that $\frac{1}{T}\sum_t (\ov{x}_i'\Delta \theta \Delta \theta' x_{it} - \frac{1}{T}\sum_t E \ov{x}_i' \Delta \theta \Delta \theta' x_{it}) = (\ov{x}_i - E\ov{x}_i)'\Delta \theta \Delta \theta' \ov{x}_i$, and 
\begin{align*}
|(\ov{x}_i - E\ov{x}_i)'\Delta \theta \Delta \theta' \ov{x}_i| \leq \|(\ov{x}_i - E\ov{x}_i)\| \|\Delta \theta \Delta \theta' \ov{x}_i \| \leq \|(\ov{x}_i - E\ov{x}_i)\| \|\Delta \theta \|^2 \|\ov{x}_i \| 
\end{align*}

Then 
\begin{align}
\prob(\|(\ov{x}_i - E\ov{x}_i)\| \|\Delta \theta \|^2 \|\ov{x}_i \| > C) &\leq \prob(\|\Delta \theta \|^2 \|\ov{x}_i \| > M^2 (M' + 1)) + \prob \left (\|(\ov{x}_i - E\ov{x}_i)\| > \frac{C}{M^2 (M' + 1)}\right ) \nonumber \\
&\leq \prob(\frac{1}{T} \sum_t \|x_{it}\|^2 > M') + \prob \left (\|(\ov{x}_i - E\ov{x}_i)\| > \frac{C}{M^2 (M' + 1)}\right ) \nonumber \\ 
&= o(T^{-a}) \label{equation:fe:multiplicative}
\end{align}

uniformly in $i$, where the second inequality follows by assumption $\ref{consistency:compact}$ and the same algebra used above to bound $\|\ov{x}_i\|$ in probability using assumption \ref{assumptions:inference:xbound}. 
That the final term is uniformly $o(T^{-a})$ follows by lemma \ref{lemma:bonhomme}, using the tail conditions and strong mixing assumed in \ref{assumptions:fe:inference:old}. 
The final term above is 
\begin{align*}
(\ov{x}_i' \Delta \theta)^2 - E(\ov{x}_i \Delta \theta)^2 &= (\ov{x}_i' \Delta \theta)^2 - \var(\ov{x}_i' \Delta \theta) - (E\ov{x}_i' \Delta \theta)^2\\
&= -\var(\ov{x}_i' \Delta \theta) + (\ov{x}_i' \Delta \theta - E(\ov{x}_i \Delta \theta))(\ov{x}_i' \Delta \theta + E(\ov{x}_i' \Delta \theta))
\end{align*}

Note that $|E(\ov{x}_i' \Delta \theta)| \leq M E \|\ov{x}_i\| \leq M E\frac{1}{T} \sum_t \|x_{it}\| \leq M E\frac{1}{T} \sum_t \|x_{it}\|^2$ by Cauchy-Schwarz, compactness of $\Theta$, and monotonicity of Lp norms, respectively. 
Let $E_{iT} = \one(\frac{1}{T} \sum_t \|x_{it}\|^2 > M')$, then 
\begin{align}
\sup_i E \left ( \frac{1}{T} \sum_t \|x_{it}\|^2 \right ) &\leq \sup_i E \left (\one(E_{it}) \frac{1}{T}\sum_t \|x_{it}\|^2 \right) + M' \nonumber \\
&\leq M' + \sup_i \prob(E_{iT})^{1/2}\left (\frac{1}{T^2} \sum_{t,s} E\|x_{it}\|^2 \|x_{is}\|^2 \right )^{1/2} \nonumber \\
&= M' + o(T^{-a})O(1) = O(1) \label{equation:fe:moment_sum}
\end{align}

where the second inequality uses Cauchy-Schwarz, and the last line uses assumptions \ref{assumptions:inference:xbound} and \ref{assumptions:fe:inference:xnorm}. 
Noting that $\ov{x}_i'\Delta \theta - E(\ov{x}_i' \Delta \theta = o_p(T^{-a})$ by mixing and tail assumptions on $x_{it}$, compactness, and lemma \ref{lemma:bonhomme}, the product term can now be shown to have $\prob[(\ov{x}_i' \Delta \theta - E(\ov{x}_i \Delta \theta))(\ov{x}_i' \Delta \theta + E(\ov{x}_i' \Delta \theta) > C] = o(T^{-a})$ using the same type of argument as in equation \ref{equation:fe:multiplicative}.
Lemma \ref{lemma:fe:variance_xi} in the supplemental appendix shows that $\max_i \var(\ov{x}_i'\Delta \theta) = O(1/T)$ under our assumptions. 
In particular, we can choose $T''$ such that $\max_i \var(\ov{x}_i'\Delta \theta) < (\wt{d}_{min} / 12)$ for all $T > T''$. \\

Finally, define $B^3_{iT} = (\ov{x}_i' \Delta \theta - E(\ov{x}_i \Delta \theta))(\ov{x}_i' \Delta \theta + E(\ov{x}_i' \Delta \theta)$ and $B^k_T \equiv \frac{1}{T} \sum_t B^k_{iT}$. 
Then for $T > \max(T', T'')$, for all $i$ we have  
\begin{align*}
&\prob \left ( \left | \frac{1}{T}\sum_t [\wt{x}_{it}'(\theta^0(c') - \theta^0(c))]^2 - g_{it} \right | \geq \frac{1}{6} \wt{d}_{min} \right ) \leq \prob \left (B^1_T + B^2_T + B_3^T + \var(\ov{x}_i'\Delta \theta) \geq \frac{1}{6}\wt{d}_{min} \right ) \\
&\leq \prob \left (B^1_T + B^2_T + B_3^T + \geq \frac{1}{12}\wt{d}_{min} \right ) \leq \sum_{k=1}^3 \prob \left (B^k_T > \frac{1}{36} \wt{d}_{min} \right) = o(T^{-a}) 
\end{align*}

The first inequality follows from the triangle inequality, the second from $T > T''$ and the final inequality from a union bound. 
The $o(T^{-a})$ holds uniformly in $i$ by the arguments above. 
This complets the proof that equation \ref{equation:fe:zbound} is $o(T^{-a})$ for any $a > 0$. \\

One final issue is the use of assumption \ref{assumptions:inference:xnorm} in equation \ref{equation:inference:Op1}. 
For the fixed effects case, we replaced this assumption with assumption $\ref{assumptions:fe:inference:xnorm}$; however, one can show that $\frac{1}{T^2} \sum_{t, s} \|\wt{x}_{it}\|^2 \|\wt{x}_{is}\|^2 \leq \frac{16}{T^2} \sum_{t, s} \|x_{it}\|^2 \|x_{is}\|^2$, so that $\max_i \frac{1}{T^2} \sum_{t, s} E(\|\wt{x}_{it}\|^2 \|\wt{x}_{is}\|^2) = O(1)$. 
Then $\frac{1}{NT^2} \sum_i \sum_{t, s} \|\wt{x}_{it}\|^2 \|\wt{x}_{is}\|^2 = O_P(1)$ by the Markov inequality.  
The remainder of the proof follows exactly as in the proof of theorem $\ref{thm:inference}$, substituting $(\wt{x}, \wt{e})_{it}$ for $(x, e)_{it}$. 
\end{proof}

\subsection{Discussion of Assumption \ref*{consistency:eigenvalue}} \label{appendix:eigenvalue_discussion}

Let $\mathcal{C}_k = \prod_{\ell} [k_{\ell}]$ and $\Gamma_{k} = \left [ \, [N] \to \mathcal{C}_k \, \right ]$ denote cluster space and possible cluster labelings of the cross-sectional units when we allow $k$ to be misspecified $k \not = k^0$, as in section \ref{section:model}.  
Recall that the general version of assumption \ref{consistency:eigenvalue} requires that there exist a $\delta > 0$ such that  
\[
\rho^k_{NT} \equiv \min_{c' \in \clusterspace^{k_0}} \min_{\gamma \in \Gamma^k} \max_{c \in \clusterspace^k} \lambda_{min} \left (\frac{1}{NT} \sum_{i, t} x_{it} x_{it}' \mathds{1}(c_i^0 = c')\mathds{1}(c_i = c) \right )\geq \delta - o_p(1)   
\]

Consider the inner term $\min_{\gamma \in \Gamma^k} \max_{c \in \clusterspace^k} \rho(c, c', \gamma)$ (recall that $\rho(c, c', \gamma)$ was defined to be the inner minimum eigenvalue). 
One interpretation of this term is given by the following two-period game: (1) An adversary colors each unit $i \in [N]$ using at most $|\clusterspace|$ colors after which (2) the econometrician chooses a color $c^*$ and forms the sample covariance matrix
\[
\frac{1}{N} \sum_{i \in I_c} \frac{1}{T} \sum_{t=1}^T x_{it}x_{it}'
\]
using only units with that color $i \in I_{c^*}$ so that its minimum eigenvalue $\lambda_{min}$ is large. \\ 
   
In the following lemma, we analyze the stochastic convergence of the eigenvalues generated by this process under full independence. 
We will need the following assumptions
\begin{assumption} \label{assumptions:lemmas:eigenvalue}
Consider the following assumptions
\begin{enumerate}[label={(\alph*)}, ref ={\ref*{assumptions:lemmas:eigenvalue}.(\alph*)}, itemindent=.5pt, itemsep=.5pt]
    \item There exists $\sigma^2 > 0$ such that $x_{it} \sim \text{SubG}(\sigma^2)$ for all $i, t$. \label{assumptions:lemmas:eigenvalue:subg}
    \item $x_{it}$ are jointly independent for all $(i, t)$.  
\end{enumerate} 
\end{assumption}

\begin{lemma}
Let assumptions \ref{assumptions:consistency} and \ref{assumptions:lemmas:eigenvalue} hold. 
Then there exists $\delta > 0$ such that 
\begin{align*}
\prob \left (\inf_{c' \in \mathcal{C}} \inf_{\gamma \in \Gamma_k} \sup_{c \in \mathcal{C}_k} \rho(c, c', \gamma) < \delta \right) = O(\poly(N) \cdot e^{-N})
\end{align*}
\end{lemma}

\begin{proof}
TBD. 
\end{proof}

\subsection{Computation} \label{appendix_computation}

As described in section \ref{section:model}, to solve problem \ref{problem} we primarily rely on Lloyd's algorithm, which performs coordinate descent on $\Theta \times \Gamma$. 
It is well known that this problem may be nonconvex, so in general coordinate descent will only yield a local minimum. 
To mitigate this issue, we rely on multiple random initializations.
In our simulations we choose initial $\theta \sim \mathcal{N}(0, \sigma^2 I)$ and each $c_{i\ell} \sim \text{Unif}([k_{\ell}])$ independently. \\ 

\textbf{Convergence Over Initializations} - In this section, we give some evidence on the convergence of our algorithm for different data-generating processes. 
Given $1 \leq v \leq S$ random initializations, let $(\wh{\theta}_v, \wh{\gamma}_v)$ be the estimator achieved on the $v^{th}$ initialization.   
Define $\qopt_s \equiv \min_{1 \leq v \leq s} \samplerisk(\wh{\theta}_v, \wh{\gamma}_v)$ and $\thetaopt_s$ to be the estimator that achieves $\qopt_s$ (out of the first $s$ initializations).
Define the mean relative errors
\begin{equation} \label{equation:computation:relative_errors}
r_Q(s) = \mathbb{E} \left [\frac{\qopt_s - \qopt_S}{\qopt_S} \right ] \quad r_{\theta}(s) = \mathbb{E} \left [\frac{\|\thetaopt_s - \thetaopt_S\|}{\|\thetaopt_S\|} \right ]
\end{equation}
where each expectation is taken over the joint distribution of $(x_{it}, y_{it})$ and the sequence of random initializations $(\theta, \gamma)^{init}_s$. 
Monte Carlo approximations of the paths $r_Q(s)$ and $r_{\theta}(s)$ for DGP's mirroring those used in section \ref{section:monte_carlo} are shown in figures \ref{figure:computation:convergence_angle}, \ref{figure:computation:convergence_nt}, and \ref{figure:computation:convergence_cluster}. 
Note in particular that ``angle'' refers to a measure of cluster separation, as in the simulation design in section \ref{section:monte_carlo}.
In table \ref{table:computation:relative_error}, we show the number of initializations required to achieve $0.1\%$ relative error for each DGP.
Each simulation reports results up to $S=200$, calculated using $200$ independent sample paths.\\  

The results show that problem \ref{problem} becomes significantly easier as $(N, T)$ increases.
Problems with fewer, well-separated clusters also converge more quickly to a stable solution (specifically, no improvements with additional random initializations up to $S$). 
In particular, $r_Q(50) \leq 0.01\%$ for all DGP's.
Thus, we use $S = 50$ for our Monte Carlo simulations, reported in section \ref{section:monte_carlo}. 
There is a large literature in computer science on heuristics for the least-squares partitioning problem, as well as some recent work on exact methods. 
See BM Appendix S1 and the references therein for more details. \\ 

\textbf{Algorithm Hyperparameters} - The hyperparameters for our implementation are $(S, tol, itermax)$, where $(tol, itermax)$ define a stopping rule for coordinate descent. 
With $j$ denoting the number of update cycles ((2) and (3) in our algorithm), we stop if either $j > itermax$ or  $\|\wh{\theta}_j - \wh{\theta}_{j - 1}\| < tol$. 
We use $tol=1\cdot 10^{-8}$ and $itermax=400$. 
We found solutions to be very insensitive to both hyperparameters for $tol > 1 \cdot 10^{-6}$ and $itermax > 100$. 
We use $S=50$, as described above. \\ 

\clearpage

\section{Tables and Figures} \label{section:tables} 

\begin{table}[htbp]
  \centering
  \caption{Effect of Cluster Separation}
    \begin{tabular}{lllllllllllll}
          &       &       &       &       &       &       &       &       &       &       &       &  \\
    \midrule
          &       & \multicolumn{2}{c}{Coverage} &       & \multicolumn{2}{c}{Bootstrap Coverage} &       & \multicolumn{2}{c}{Param. MSE} &       & \multicolumn{2}{c}{Cluster Loss} \\
    \multicolumn{1}{c}{Angle ($\alpha$)} &       & \multicolumn{1}{c}{AR(1)} & \multicolumn{1}{c}{HK} &       & \multicolumn{1}{c}{AR(1)} & \multicolumn{1}{c}{HK} &       & \multicolumn{1}{c}{AR(1)} & \multicolumn{1}{c}{HK} &       & \multicolumn{1}{c}{AR(1)} & \multicolumn{1}{c}{HK} \\
\cmidrule{1-1}\cmidrule{3-4}\cmidrule{6-7}\cmidrule{9-10}\cmidrule{12-13}    \multicolumn{1}{c}{1.57} &       & \multicolumn{1}{c}{0.90} & \multicolumn{1}{c}{0.89} &       & \multicolumn{1}{c}{0.88} & \multicolumn{1}{c}{0.87} &       & \multicolumn{1}{c}{0.047} & \multicolumn{1}{c}{0.053} &       & \multicolumn{1}{c}{0.044} & \multicolumn{1}{c}{0.050} \\
    \multicolumn{1}{c}{1.26} &       & \multicolumn{1}{c}{0.86} & \multicolumn{1}{c}{0.84} &       & \multicolumn{1}{c}{0.83} & \multicolumn{1}{c}{0.81} &       & \multicolumn{1}{c}{0.055} & \multicolumn{1}{c}{0.061} &       & \multicolumn{1}{c}{0.074} & \multicolumn{1}{c}{0.083} \\
    \multicolumn{1}{c}{0.94} &       & \multicolumn{1}{c}{0.75} & \multicolumn{1}{c}{0.72} &       & \multicolumn{1}{c}{0.73} & \multicolumn{1}{c}{0.69} &       & \multicolumn{1}{c}{0.059} & \multicolumn{1}{c}{0.067} &       & \multicolumn{1}{c}{0.13} & \multicolumn{1}{c}{0.14} \\
    \multicolumn{1}{c}{0.63} &       & \multicolumn{1}{c}{0.53} & \multicolumn{1}{c}{0.50} &       & \multicolumn{1}{c}{0.51} & \multicolumn{1}{c}{0.49} &       & \multicolumn{1}{c}{0.058} & \multicolumn{1}{c}{0.064} &       & \multicolumn{1}{c}{0.22} & \multicolumn{1}{c}{0.23} \\
    \multicolumn{1}{c}{0.31} &       & \multicolumn{1}{c}{0.25} & \multicolumn{1}{c}{0.25} &       & \multicolumn{1}{c}{0.24} & \multicolumn{1}{c}{0.24} &       & \multicolumn{1}{c}{0.052} & \multicolumn{1}{c}{0.059} &       & \multicolumn{1}{c}{0.37} & \multicolumn{1}{c}{0.38} \\
    \multicolumn{1}{c}{0.16} &       & \multicolumn{1}{c}{0.20} & \multicolumn{1}{c}{0.19} &       & \multicolumn{1}{c}{0.20} & \multicolumn{1}{c}{0.19} &       & \multicolumn{1}{c}{0.048} & \multicolumn{1}{c}{0.054} &       & \multicolumn{1}{c}{0.44} & \multicolumn{1}{c}{0.44} \\
    \midrule
    \multicolumn{13}{l}{\textit{Notes: N=150, T=10}} \\
    \end{tabular}%
  \label{table:separation}%
\end{table}%

\begin{table}[htbp]
  \centering
  \caption{Effect of Sample Size (N, T)}
   \begin{adjustbox}{width=1.1\columnwidth,center} 
    \begin{tabular}{r|cccccccccccccccccc}
    \multicolumn{1}{c}{} &       &       &       &       &       &       &       &       &       &       &       &       &       &       &       &       &       &  \\
    \midrule
    \multicolumn{1}{r}{} &       & \multicolumn{5}{c}{Param. MSE} &       & \multicolumn{5}{c}{Coverage (Analytical)} &       & \multicolumn{5}{c}{Cluster Loss} \\
\cmidrule{3-7}\cmidrule{9-13}\cmidrule{15-19}    \multicolumn{1}{l}{Errors} & T     & N=50  & 100   & 150   & 200   & 250   &       & 50    & 100   & 150   & 200   & 250   &       & 50    & 100   & 150   & 200   & 250 \\
    \midrule
          & 5     & 0.154 & 0.140 & 0.137 & 0.134 & 0.133 &       & 0.761 & 0.763 & 0.772 & 0.768 & 0.756 &       & 0.134 & 0.125 & 0.125 & 0.123 & 0.123 \\
    \multicolumn{1}{l|}{AR(1)} & 10    & 0.051 & 0.047 & 0.047 & 0.046 & 0.046 &       & 0.873 & 0.889 & 0.898 & 0.899 & 0.901 &       & 0.046 & 0.044 & 0.044 & 0.044 & 0.044 \\
          & 15    & 0.021 & 0.019 & 0.018 & 0.018 & 0.018 &       & 0.909 & 0.927 & 0.939 & 0.927 & 0.926 &       & 0.018 & 0.017 & 0.017 & 0.017 & 0.017 \\
          & 20    & 0.009 & 0.008 & 0.008 & 0.008 & 0.007 &       & 0.916 & 0.939 & 0.935 & 0.935 & 0.935 &       & 0.007 & 0.007 & 0.007 & 0.007 & 0.007 \\
          & 25    & 0.005 & 0.004 & 0.004 & 0.004 & 0.003 &       & 0.930 & 0.937 & 0.936 & 0.941 & 0.943 &       & 0.003 & 0.003 & 0.003 & 0.003 & 0.003 \\
    \multicolumn{1}{r}{} &       &       &       &       &       &       &       &       &       &       &       &       &       &       &       &       &       &  \\
          & 5     & 0.160 & 0.146 & 0.144 & 0.143 & 0.140 &       & 0.718 & 0.748 & 0.738 & 0.738 & 0.736 &       & 0.137 & 0.130 & 0.130 & 0.130 & 0.128 \\
    \multicolumn{1}{l|}{HK} & 10    & 0.059 & 0.055 & 0.053 & 0.052 & 0.051 &       & 0.856 & 0.875 & 0.886 & 0.888 & 0.887 &       & 0.053 & 0.051 & 0.050 & 0.049 & 0.049 \\
          & 15    & 0.027 & 0.024 & 0.023 & 0.023 & 0.022 &       & 0.905 & 0.914 & 0.916 & 0.920 & 0.917 &       & 0.023 & 0.022 & 0.021 & 0.022 & 0.021 \\
          & 20    & 0.013 & 0.011 & 0.011 & 0.010 & 0.010 &       & 0.915 & 0.927 & 0.930 & 0.936 & 0.941 &       & 0.010 & 0.010 & 0.010 & 0.009 & 0.009 \\
          & 25    & 0.007 & 0.006 & 0.005 & 0.005 & 0.005 &       & 0.924 & 0.936 & 0.946 & 0.946 & 0.942 &       & 0.004 & 0.004 & 0.005 & 0.004 & 0.004 \\
    \bottomrule
    \end{tabular}%
    \end{adjustbox}
  \label{table:nt}%
\end{table}%

\begin{table}[htbp]
  \centering
  \caption{Effect of Number of Clusters $(k_1, k_2)$}
   \begin{adjustbox}{width=1.1\columnwidth,center} 
    \begin{tabular}{lrrrrrrrrrrrrrr}
          &       &       &       &       &       &       &       &       &       &       &       &       &       &  \\
    \midrule
    \multicolumn{1}{c}{\# Clusters} & \multicolumn{2}{c}{Coverage} &       & \multicolumn{2}{c}{Bootstrap Coverage} &       & \multicolumn{2}{c}{Param. MSE} &       & \multicolumn{2}{c}{Function MSE} &       & \multicolumn{2}{c}{Cluster Loss} \\
    \multicolumn{1}{c}{(k1, k2)} & \multicolumn{1}{c}{AR(1)} & \multicolumn{1}{c}{HK} &       & \multicolumn{1}{c}{AR(1)} & \multicolumn{1}{c}{HK} &       & \multicolumn{1}{c}{AR(1)} & \multicolumn{1}{c}{HK} &       & \multicolumn{1}{c}{AR(1)} & \multicolumn{1}{c}{HK} &       & \multicolumn{1}{c}{AR(1)} & \multicolumn{1}{c}{HK} \\
\cmidrule{2-3}\cmidrule{5-6}\cmidrule{8-9}\cmidrule{11-12}\cmidrule{14-15}    \multicolumn{1}{c}{(1, 2)} & \multicolumn{1}{c}{0.90} & \multicolumn{1}{c}{0.88} &       & \multicolumn{1}{c}{0.88} & \multicolumn{1}{c}{0.86} &       & \multicolumn{1}{c}{0.026} & \multicolumn{1}{c}{0.030} &       & \multicolumn{1}{c}{0.026} & \multicolumn{1}{c}{0.094} &       & \multicolumn{1}{c}{0.035} & \multicolumn{1}{c}{0.040} \\
    \multicolumn{1}{c}{(2, 2)} & \multicolumn{1}{c}{0.86} & \multicolumn{1}{c}{0.83} &       & \multicolumn{1}{c}{0.83} & \multicolumn{1}{c}{0.81} &       & \multicolumn{1}{c}{0.054} & \multicolumn{1}{c}{0.063} &       & \multicolumn{1}{c}{0.054} & \multicolumn{1}{c}{0.184} &       & \multicolumn{1}{c}{0.074} & \multicolumn{1}{c}{0.085} \\
    \multicolumn{1}{c}{(2, 3)} & \multicolumn{1}{c}{0.84} & \multicolumn{1}{c}{0.81} &       & \multicolumn{1}{c}{0.82} & \multicolumn{1}{c}{0.79} &       & \multicolumn{1}{c}{0.070} & \multicolumn{1}{c}{0.078} &       & \multicolumn{1}{c}{0.070} & \multicolumn{1}{c}{0.218} &       & \multicolumn{1}{c}{0.090} & \multicolumn{1}{c}{0.100} \\
    \multicolumn{1}{c}{(3, 3)} & \multicolumn{1}{c}{0.81} & \multicolumn{1}{c}{0.78} &       & \multicolumn{1}{c}{0.80} & \multicolumn{1}{c}{0.76} &       & \multicolumn{1}{c}{0.085} & \multicolumn{1}{c}{0.095} &       & \multicolumn{1}{c}{0.085} & \multicolumn{1}{c}{0.258} &       & \multicolumn{1}{c}{0.106} & \multicolumn{1}{c}{0.119} \\
    \multicolumn{1}{c}{(3, 4)} & \multicolumn{1}{c}{0.80} & \multicolumn{1}{c}{0.76} &       & \multicolumn{1}{c}{0.78} & \multicolumn{1}{c}{0.74} &       & \multicolumn{1}{c}{0.099} & \multicolumn{1}{c}{0.109} &       & \multicolumn{1}{c}{0.099} & \multicolumn{1}{c}{0.281} &       & \multicolumn{1}{c}{0.118} & \multicolumn{1}{c}{0.131} \\
    \multicolumn{1}{c}{(4, 4)} & \multicolumn{1}{c}{0.77} & \multicolumn{1}{c}{0.73} &       & \multicolumn{1}{c}{0.76} & \multicolumn{1}{c}{0.72} &       & \multicolumn{1}{c}{0.114} & \multicolumn{1}{c}{0.123} &       & \multicolumn{1}{c}{0.114} & \multicolumn{1}{c}{0.306} &       & \multicolumn{1}{c}{0.132} & \multicolumn{1}{c}{0.143} \\
    \midrule
    \multicolumn{15}{l}{\textit{Notes: N=150, T=10}} \\
    \end{tabular}%
    \end{adjustbox}
  \label{table:clusters}%
\end{table}%

\begin{table}[htbp]
  \centering
  \caption{Effect of Misspecified Blocking of Covariates - Estimation with $B^0=1$ and $B=2$}
    \begin{tabular}{lrrrrrrr}
          &       &       &       &       &       &       &  \\
    \midrule
    \multicolumn{1}{c}{Errors} & \multicolumn{1}{c}{\# Clusters} &       & \multicolumn{2}{c}{Param. MSE} &       & \multicolumn{2}{c}{Function MSE} \\
          & \multicolumn{1}{c}{(k1, k2)} &       & \multicolumn{1}{c}{B=1} & \multicolumn{1}{c}{B=2} &       & \multicolumn{1}{c}{B=1} & \multicolumn{1}{c}{B=2} \\
\cmidrule{2-2}\cmidrule{4-5}\cmidrule{7-8}          & \multicolumn{1}{c}{(1, 2)} &       & 0.027 & 0.026 &       & 0.077 & \multicolumn{1}{c}{0.026} \\
          & \multicolumn{1}{c}{(2, 2)} &       & 0.058 & 0.054 &       & 0.157 & \multicolumn{1}{c}{0.054} \\
    \multicolumn{1}{c}{AR(1)} & \multicolumn{1}{c}{(2, 3)} &       & 0.079 & 0.070 &       & 0.206 & \multicolumn{1}{c}{0.070} \\
          & \multicolumn{1}{c}{(3, 3)} &       & 0.108 & 0.085 &       & 0.273 & \multicolumn{1}{c}{0.085} \\
          & \multicolumn{1}{c}{(3, 4)} &       & 0.137 & 0.099 &       & 0.334 & \multicolumn{1}{c}{0.099} \\
          & \multicolumn{1}{c}{(4, 4)} &       & 0.163 & 0.114 &       & 0.383 & \multicolumn{1}{c}{0.114} \\
          &       &       &       &       &       &       &  \\
          &       &       & \multicolumn{1}{c}{B=1} & \multicolumn{1}{c}{B=2} &       & \multicolumn{1}{c}{B=1} & \multicolumn{1}{c}{B=2} \\
\cmidrule{2-2}\cmidrule{4-5}\cmidrule{7-8}          & \multicolumn{1}{c}{(1, 2)} &       & 0.031 & \multicolumn{1}{c}{0.030} &       & 0.097 & \multicolumn{1}{c}{0.094} \\
          & \multicolumn{1}{c}{(2, 2)} &       & 0.068 & \multicolumn{1}{c}{0.063} &       & 0.200 & \multicolumn{1}{c}{0.184} \\
    \multicolumn{1}{c}{HK} & \multicolumn{1}{c}{(2, 3)} &       & 0.090 & \multicolumn{1}{c}{0.078} &       & 0.252 & \multicolumn{1}{c}{0.218} \\
          & \multicolumn{1}{c}{(3, 3)} &       & 0.122 & \multicolumn{1}{c}{0.095} &       & 0.335 & \multicolumn{1}{c}{0.258} \\
          & \multicolumn{1}{c}{(3, 4)} &       & 0.149 & \multicolumn{1}{c}{0.109} &       & 0.393 & \multicolumn{1}{c}{0.281} \\
          & \multicolumn{1}{c}{(4, 4)} &       & 0.174 & \multicolumn{1}{c}{0.123} &       & 0.441 & \multicolumn{1}{c}{0.306} \\
    \midrule
    \multicolumn{8}{l}{\textit{Notes: N=150, T=10}} \\
    \end{tabular}%
  \label{table:misspec}%
\end{table}%

\begin{table}[htbp]
  \centering
  \caption{Effect of Dimension Imbalance}
   \begin{adjustbox}{width=1\columnwidth,center} 
    \begin{tabular}{lrrrrrrrrrrrrrr}
          &       &       &       &       &       &       &       &       &       &       &       &       &       &  \\
    \midrule
    \multicolumn{1}{c}{dim} & \multicolumn{2}{c}{Coverage-large} &       & \multicolumn{2}{c}{Coverage-small} &       & \multicolumn{2}{c}{Cluster loss-small} &       & \multicolumn{2}{c}{Cluster loss-large} &       & \multicolumn{2}{c}{Param. MSE} \\
    \multicolumn{1}{c}{(m, p-m)} & \multicolumn{1}{c}{AR(1)} & \multicolumn{1}{c}{HK} &       & \multicolumn{1}{c}{AR(1)} & \multicolumn{1}{c}{HK} &       & \multicolumn{1}{c}{AR(1)} & \multicolumn{1}{c}{HK} &       & \multicolumn{1}{c}{AR(1)} & \multicolumn{1}{c}{HK} &       & \multicolumn{1}{c}{AR(1)} & \multicolumn{1}{c}{HK} \\
\cmidrule{2-3}\cmidrule{5-6}\cmidrule{8-9}\cmidrule{11-12}\cmidrule{14-15}    \multicolumn{1}{c}{(1, 11)} & \multicolumn{1}{c}{0.921} & \multicolumn{1}{c}{0.918} &       & \multicolumn{1}{c}{0.599} & \multicolumn{1}{c}{0.538} &       & \multicolumn{1}{c}{0.156} & \multicolumn{1}{c}{0.170} &       & \multicolumn{1}{c}{0.000} & \multicolumn{1}{c}{0.000} &       & \multicolumn{1}{c}{0.009} & \multicolumn{1}{c}{0.009} \\
    \multicolumn{1}{c}{(2, 10)} & \multicolumn{1}{c}{0.930} & \multicolumn{1}{c}{0.931} &       & \multicolumn{1}{c}{0.841} & \multicolumn{1}{c}{0.840} &       & \multicolumn{1}{c}{0.069} & \multicolumn{1}{c}{0.060} &       & \multicolumn{1}{c}{0.001} & \multicolumn{1}{c}{0.001} &       & \multicolumn{1}{c}{0.008} & \multicolumn{1}{c}{0.007} \\
    \multicolumn{1}{c}{(3, 9)} & \multicolumn{1}{c}{0.932} & \multicolumn{1}{c}{0.931} &       & \multicolumn{1}{c}{0.917} & \multicolumn{1}{c}{0.910} &       & \multicolumn{1}{c}{0.029} & \multicolumn{1}{c}{0.028} &       & \multicolumn{1}{c}{0.001} & \multicolumn{1}{c}{0.001} &       & \multicolumn{1}{c}{0.006} & \multicolumn{1}{c}{0.006} \\
    \multicolumn{1}{c}{(4, 8)} & \multicolumn{1}{c}{0.935} & \multicolumn{1}{c}{0.930} &       & \multicolumn{1}{c}{0.922} & \multicolumn{1}{c}{0.916} &       & \multicolumn{1}{c}{0.016} & \multicolumn{1}{c}{0.019} &       & \multicolumn{1}{c}{0.001} & \multicolumn{1}{c}{0.001} &       & \multicolumn{1}{c}{0.005} & \multicolumn{1}{c}{0.005} \\
    \multicolumn{1}{c}{(5, 7)} & \multicolumn{1}{c}{0.931} & \multicolumn{1}{c}{0.934} &       & \multicolumn{1}{c}{0.927} & \multicolumn{1}{c}{0.934} &       & \multicolumn{1}{c}{0.006} & \multicolumn{1}{c}{0.007} &       & \multicolumn{1}{c}{0.003} & \multicolumn{1}{c}{0.002} &       & \multicolumn{1}{c}{0.005} & \multicolumn{1}{c}{0.005} \\
    \multicolumn{1}{c}{(6, 6)} & \multicolumn{1}{c}{0.936} & \multicolumn{1}{c}{0.938} &       & \multicolumn{1}{c}{0.939} & \multicolumn{1}{c}{0.935} &       & \multicolumn{1}{c}{0.003} & \multicolumn{1}{c}{0.003} &       & \multicolumn{1}{c}{0.004} & \multicolumn{1}{c}{0.004} &       & \multicolumn{1}{c}{0.004} & \multicolumn{1}{c}{0.004} \\
    \midrule
    \textit{Notes: N=150, T=10} &       &       &       &       &       &       &       &       &       &       &       &       &       &  \\
    \end{tabular}%
    \end{adjustbox}
  \label{table:imbalance}%
\end{table}%

\begin{table}[htbp]
  \centering
  \caption{Effect of Growing Model Dimension}
    \begin{tabular}{lrrrrrrrrrrrrrr}
          &       &       &       &       &       &       &       &       &       &       &       &       &       &  \\
\cmidrule{1-12}    \multicolumn{1}{c}{dim (p)} & \multicolumn{2}{c}{Coverage} &       & \multicolumn{2}{c}{Param. MSE} &       & \multicolumn{2}{c}{Function MSE} &       & \multicolumn{2}{c}{Cluster Loss} &       &       &  \\
    \multicolumn{1}{c}{Error} & \multicolumn{1}{c}{indep} & \multicolumn{1}{c}{AR(1)} &       & \multicolumn{1}{c}{indep} & \multicolumn{1}{c}{AR(1)} &       & \multicolumn{1}{c}{indep} & \multicolumn{1}{c}{AR(1)} &       & \multicolumn{1}{c}{indep} & \multicolumn{1}{c}{AR(1)} &       &       &  \\
\cmidrule{2-3}\cmidrule{5-6}\cmidrule{8-9}\cmidrule{11-12}    \multicolumn{1}{c}{1} & \multicolumn{1}{c}{0.93} & \multicolumn{1}{c}{0.92} &       & \multicolumn{1}{c}{0.046} & \multicolumn{1}{c}{0.078} &       & \multicolumn{1}{c}{0.023} & \multicolumn{1}{c}{0.065} &       & \multicolumn{1}{c}{0.011} & \multicolumn{1}{c}{0.019} &       &       &  \\
    \multicolumn{1}{c}{2} & \multicolumn{1}{c}{0.94} & \multicolumn{1}{c}{0.91} &       & \multicolumn{1}{c}{0.047} & \multicolumn{1}{c}{0.060} &       & \multicolumn{1}{c}{0.045} & \multicolumn{1}{c}{0.077} &       & \multicolumn{1}{c}{0.012} & \multicolumn{1}{c}{0.015} &       &       &  \\
    \multicolumn{1}{c}{3} & \multicolumn{1}{c}{0.92} & \multicolumn{1}{c}{0.93} &       & \multicolumn{1}{c}{0.051} & \multicolumn{1}{c}{0.059} &       & \multicolumn{1}{c}{0.067} & \multicolumn{1}{c}{0.095} &       & \multicolumn{1}{c}{0.012} & \multicolumn{1}{c}{0.014} &       &       &  \\
    \multicolumn{1}{c}{4} & \multicolumn{1}{c}{0.91} & \multicolumn{1}{c}{0.92} &       & \multicolumn{1}{c}{0.058} & \multicolumn{1}{c}{0.063} &       & \multicolumn{1}{c}{0.098} & \multicolumn{1}{c}{0.122} &       & \multicolumn{1}{c}{0.014} & \multicolumn{1}{c}{0.015} &       &       &  \\
    \multicolumn{1}{c}{5} & \multicolumn{1}{c}{0.92} & \multicolumn{1}{c}{0.90} &       & \multicolumn{1}{c}{0.064} & \multicolumn{1}{c}{0.068} &       & \multicolumn{1}{c}{0.127} & \multicolumn{1}{c}{0.150} &       & \multicolumn{1}{c}{0.015} & \multicolumn{1}{c}{0.016} &       &       &  \\
\cmidrule{1-12}    \multicolumn{15}{l}{\textit{Notes: N=150, T=10}} \\
    \end{tabular}%
  \label{table:dimension}%
\end{table}%

\begin{table}[htbp]
  \centering
    \caption{Number of Initializations for 0.1\% Rel. Error}
    \begin{tabular}{rrrrrrrrr}
          &       &       &       &       &       &       &       &  \\
    \midrule
          &       &       &       &       &       &       &       &  \\
    \multicolumn{1}{c}{(N, T)} &       & \multicolumn{1}{c}{$s_{\theta}$} & \multicolumn{1}{c}{$s_q$} &       & \multicolumn{1}{c}{Angle} &       & \multicolumn{1}{c}{$s_{\theta}$} & \multicolumn{1}{c}{$s_q$} \\
\cmidrule{1-1}\cmidrule{3-4}\cmidrule{6-6}\cmidrule{8-9}    \multicolumn{1}{c}{(20, 10)} &       & \multicolumn{1}{c}{84} & \multicolumn{1}{c}{9} &       & \multicolumn{1}{c}{1} &       & \multicolumn{1}{c}{14} & \multicolumn{1}{c}{0} \\
    \multicolumn{1}{c}{(50, 10)} &       & \multicolumn{1}{c}{65} & \multicolumn{1}{c}{1} &       & \multicolumn{1}{c}{0.8} &       & \multicolumn{1}{c}{98} & \multicolumn{1}{c}{0} \\
    \multicolumn{1}{c}{(100, 10)} &       & \multicolumn{1}{c}{28} & \multicolumn{1}{c}{0} &       & \multicolumn{1}{c}{0.6} &       & \multicolumn{1}{c}{56} & \multicolumn{1}{c}{0} \\
    \multicolumn{1}{c}{(150, 10)} &       & \multicolumn{1}{c}{9} & \multicolumn{1}{c}{0} &       & \multicolumn{1}{c}{0.4} &       & \multicolumn{1}{c}{109} & \multicolumn{1}{c}{2} \\
    \multicolumn{1}{c}{(250, 10)} &       & \multicolumn{1}{c}{7} & \multicolumn{1}{c}{0} &       & \multicolumn{1}{c}{0.2} &       & \multicolumn{1}{c}{139} & \multicolumn{1}{c}{4} \\
          &       &       &       &       & \multicolumn{1}{c}{0.1} &       & \multicolumn{1}{c}{157} & \multicolumn{1}{c}{5} \\
          &       &       &       &       &       &       &       &  \\
    \multicolumn{1}{c}{(N, T)} &       & \multicolumn{1}{c}{$s_{\theta}$} & \multicolumn{1}{c}{$s_q$} &       & \multicolumn{1}{c}{K} &       & \multicolumn{1}{c}{$s_{\theta}$} & \multicolumn{1}{c}{$s_q$} \\
\cmidrule{1-1}\cmidrule{3-4}\cmidrule{6-6}\cmidrule{8-9}    \multicolumn{1}{c}{(50, 5)} &       & \multicolumn{1}{c}{163} & \multicolumn{1}{c}{10} &       & \multicolumn{1}{c}{3} &       & \multicolumn{1}{c}{14} & \multicolumn{1}{c}{1} \\
    \multicolumn{1}{c}{(50, 15)} &       & \multicolumn{1}{c}{17} & \multicolumn{1}{c}{0} &       & \multicolumn{1}{c}{4} &       & \multicolumn{1}{c}{48} & \multicolumn{1}{c}{0} \\
    \multicolumn{1}{c}{(50, 20)} &       & \multicolumn{1}{c}{9} & \multicolumn{1}{c}{1} &       & \multicolumn{1}{c}{5} &       & \multicolumn{1}{c}{153} & \multicolumn{1}{c}{1} \\
    \multicolumn{1}{c}{(50, 25)} &       & \multicolumn{1}{c}{1} & \multicolumn{1}{c}{0} &       & \multicolumn{1}{c}{6} &       & \multicolumn{1}{c}{130} & \multicolumn{1}{c}{4} \\
          &       &       &       &       & \multicolumn{1}{c}{7} &       & \multicolumn{1}{c}{163} & \multicolumn{1}{c}{4} \\
    \midrule
    \multicolumn{9}{r}{} \\
    \end{tabular}%
  \label{table:computation:relative_error}%
\end{table}%

\clearpage

\begin{figure}
\centering
\caption{Algorithm Convergence and Cluster Separation}
  \begin{subfigure}[b]{0.45\textwidth}
    \includegraphics[width=\textwidth]{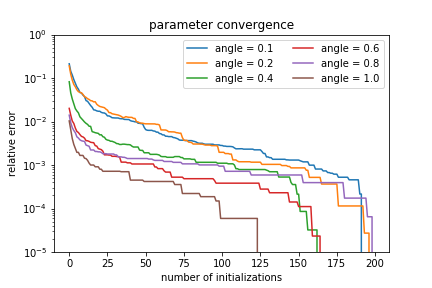}
    \caption{$r_{\theta}(s)$}
  \end{subfigure}
  \begin{subfigure}[b]{0.45\textwidth}
    \includegraphics[width=\textwidth]{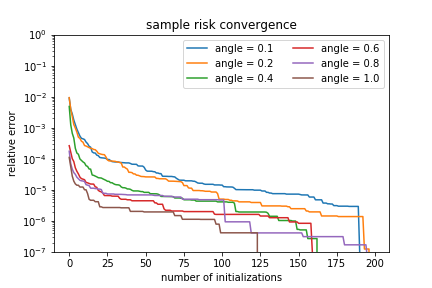}
    \caption{$r_{Q}(s)$}
  \end{subfigure}
  \label{figure:computation:convergence_angle}
\end{figure}

\begin{figure}
\centering
\caption{Algorithm Convergence and $(N, T)$}
  \begin{subfigure}[b]{0.45\textwidth}
    \includegraphics[width=\textwidth]{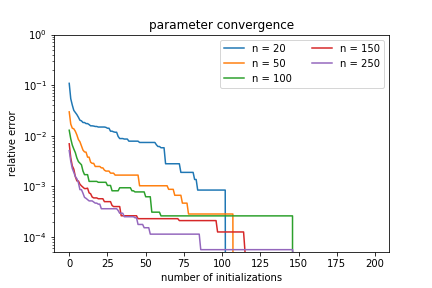}
    \caption{$r_{\theta}(s)$}
  \end{subfigure}
  \begin{subfigure}[b]{0.45\textwidth}
    \includegraphics[width=\textwidth]{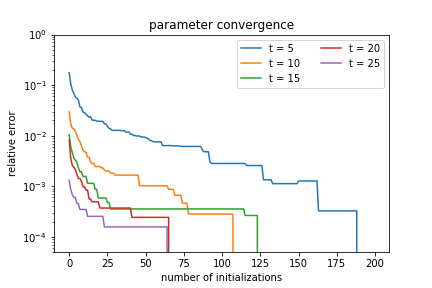}
    \caption{$r_{\theta}(s)$}
  \end{subfigure}
  \label{figure:computation:convergence_nt}
\end{figure}

\begin{figure}
\centering
\caption{Algorithm Convergence and Number of Clusters $k = k_1 + k_2$}
  \begin{subfigure}[b]{0.45\textwidth}
    \includegraphics[width=\textwidth]{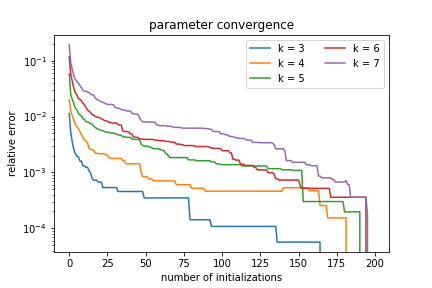}
    \caption{$r_{\theta}(s)$}
  \end{subfigure}
  \begin{subfigure}[b]{0.45\textwidth}
    \includegraphics[width=\textwidth]{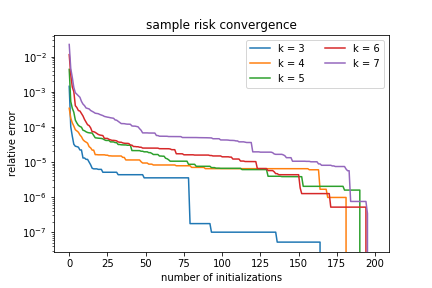}
    \caption{$r_{Q}(s)$}
  \end{subfigure}
  \label{figure:computation:convergence_cluster}
\end{figure}

\clearpage


\end{document}